\documentclass[a4paper,amsmath,onecolumn, floatfix,nofootinbib, prx, aps, usenames, dvipsnames,11pt, unpublished]{quantumarticle}
\pdfoutput=1

\usepackage[utf8]{inputenc}
\usepackage[english]{babel}
\usepackage[T1]{fontenc}
\usepackage{natbib}

\usepackage[usenames,dvipsnames]{xcolor}
\usepackage{graphicx}
\usepackage[ colorlinks = true, 
             linkcolor = MidnightBlue,
             urlcolor  = MidnightBlue,
             citecolor = orange,
             anchorcolor = ForestGreen,
]{hyperref} 
\usepackage{float}

\usepackage[sans]{dsfont} 
\usepackage{bbm}
\usepackage[mathscr]{euscript}
\usepackage{cancel} 
\usepackage{amsfonts}

\usepackage{tcolorbox}
\tcbuselibrary{theorems} 

\usepackage{amsmath}
\usepackage{amssymb}
\usepackage{units}
\usepackage{framed}
\usepackage{mdframed}
\usepackage{thmtools, thm-restate} 
\usepackage{ucs}
\usepackage{subcaption}
\usepackage{fullpage}
\usepackage{csquotes}
\usepackage{epigraph}

\usepackage{multicol}
\usepackage{enumerate}

\usepackage{tikz}
\usetikzlibrary{arrows}
\usepackage{rotating, xcolor}
\usetikzlibrary{decorations.pathreplacing,calligraphy}
\usetikzlibrary{shapes}
\usetikzlibrary{hobby}
\usetikzlibrary{decorations.markings}

\tcbset{colback=orange!5,colframe=orange!75!yellow, 
fonttitle= \bfseries, float=htb}

    \newcommand{\ket}[1]{\vert  #1 \rangle}
    \newcommand{\bra}[1]{\langle #1 |}
    \newcommand{\inprod}[2]{\langle #1 | #2 \rangle}

	\newcommand{\id}{\mathbbm{1}}

	\newcommand{\tr}{\operatorname{Tr}  }
	
	\renewcommand{\vec}[1]{\mathbf{#1}}

\newcommand{\footnoterecall}[1]{\hyperref[#1]{\footnotemark[\value{#1}]}}

\declaretheorem[]{axiom}
\declaretheorem[]{definition}
\declaretheorem[]{corollary}

\newenvironment{proof}{\paragraph{Proof:}}{\hfill$\square$}

\usepackage{cleveref}

\usetikzlibrary {positioning}
\definecolor{processblue}{cmyk}{0.96,0,0,0}

\usepackage{thmtools, thm-restate}

\usepackage[page]{totalcount}
\usepackage{etoolbox}

\usetikzlibrary{arrows,decorations.markings}
\begin{document}

\title{Any theory that admits a Wigner's Friend type multi-agent paradox is logically contextual}
\author{Nuriya Nurgalieva}
\affiliation{Institute for Theoretical Physics, ETH Z\"{u}rich, 8093 Z\"{u}rich, Switzerland}
\affiliation{Department of Physics, University of Z\"{u}rich, 8057 Z\"{u}rich, Switzerland}
\email{nuriya.nurgalieva@physik.uzh.ch}

\author{V. Vilasini}
\affiliation{Institute for Theoretical Physics, ETH Z\"{u}rich, 8093 Z\"{u}rich, Switzerland}
\affiliation{Inria, Université Grenoble Alpes, 38000 Grenoble, France}
\email{vilasini@inria.fr}

\maketitle

\date{}

\begin{abstract}

Wigner’s Friend scenarios push the boundaries of quantum theory by modeling agents, along with their memories storing measurement outcomes, as physical quantum systems. Extending these ideas beyond quantum theory, we ask: in which physical theories, and under what assumptions, can agents who reason logically about each other’s measurement outcomes encounter apparent paradoxes? To address this, we prove a link between Wigner's Friend type multi-agent paradoxes and contextuality in general theories: if agents who are modeled within a physical theory come to a contradiction when reasoning using that theory (under certain assumptions on how they reason and describe measurements), then the theory must admit contextual correlations of a logical form. This also yields a link between the distinct
fundamental concepts of Heisenberg cuts and measurement contexts in general theories, and in particular, implies that the quantum Frauchiger-Renner paradox is a proof of logical contextuality. Moreover, we identify structural properties of such paradoxes in general theories and specific to quantum theory. For instance, we demonstrate that theories admitting behaviors corresponding to extremal vertices of n-cycle contextuality scenarios admit Wigner's Friend type paradoxes without post-selection, and that any quantum Wigner's Friend paradox based on the n-cycle scenario must necessarily involve post-selection. Further, we construct a multi-agent paradox based on a genuine contextuality scenario involving sequential measurements on a single system, showing that Bell non-local correlations between distinct subsystems is not necessary for Wigner's Friend paradoxes. 
Our work offers an approach to investigate the structure of physical theories and their information-theoretic resources by means of
deconstructing the assumptions underlying multi-agent physical paradoxes.
\end{abstract}

\begingroup
\renewcommand\thefootnote{}
\footnotetext{\hspace{-2em}\emph{Both authors contributed equally to this work.}}
\addtocounter{footnote}{-1}
\endgroup

\setlength{\epigraphwidth}{4in}
\epigraph{What if opposites could be combined and transcended, paradox embraced, a whole life lived in contradictory case?}{Rachel Hartman, \emph{Tess of the Road}}

\epigraph{The `paradox' is only a conflict between reality and your feeling of what reality `ought to be'.}{Richard Feynman, \emph{Volume III, Chapter 18. Angular Momentum, The Feynman Lectures on Physics}}

\tableofcontents

\section{Introduction}
\label{sec:introduction}
This introduction starts with a lie. Or did it? The next sentence is true. The previous one is false. These are examples of logical paradoxes, which when mapped to scenarios where statements are attributed to agents: ``Alice says that Bob lies; Bob says that she told the truth'', constitute a multi-agent paradox (Figure~\ref{fig:epistemic-simple}). The agents, as information carriers and measurement makers, can be embedded in the physical world. We can then consider agents' statements about measurement outcomes that they observe, their associated inferences about the measurement outcomes of another agent and whether this can lead to a physical multi-agent paradox.

An example of such a paradox in a physical theory would be a situation where Alice and Bob measure some systems of the theory at different times, and “Alice upon observing her outcome to be $a=1$ concludes that Bob must have obtained outcome $b=0$” but through another chain of reasoning using the same theory, “Alice knows that upon seeing his outcome to be $b=0$, Bob must have concluded that she will obtain the opposite outcome $a=0$.'' A natural question then arises: can paradoxes and contradictions like these emerge in physical settings, from the experiments, observations and reasoning of truthful, rational agents? If yes, what properties must a physical theory possess, and under what precise assumptions must the reasoning take place, in order to give rise to such paradoxes? 
Developing formal methods to address such questions will shed light on fundamental limits of consistency of scientific reasoning in physical theories, while linking these limits to our physical understanding of measurements and information-theoretic resources of the theory. 

 Within quantum theory, a thought experiment was introduced by Frauchiger and Renner~\cite{Frauchiger2018} (FR)~\footnote{We revisit this scenario from a contextual point of view in Section~\ref{subsec:contextuality}, and discuss its generalizations in Section~\ref{sec:paradox}.}, where agents who  model each other as quantum systems \textit{and} are allowed to reason about each other's measurement outcomes, come to a logical contradiction. This conclusion is made under certain assumptions relating to how agents make inferences and how their measurements are modeled \cite{Nurgalieva2021}. The situation considered by FR constitutes a Wigner's Friend Scenario \cite{Wigner1961}, which are quantum protocols where the agents and their labs are also regarded as possibly unitarily evolving closed quantum systems, on which other agents can perform arbitrary quantum operations. Importantly, in this case, the paradox relies on the ambiguity in modelling quantum measurements: either as yielding definite classical outcomes or as unitary quantum evolutions that are not associated with definite classical records, the distinction being captured by the Heisenberg cut in quantum theory \cite{Vilasini2022}. It was demonstrated in \cite{Vilasini2022} that any multi-agent paradox arising in any Wigner's Friend scenario that can be considered in quantum theory (such as that of FR) can always be resolved by clearly accounting for the choice of Heisenberg cuts under which each prediction and logical statement is derived. 

Given that the Heisenberg cut does not typically play an explicit role in our standard usage of quantum theory, this raises the question: in which situations can the information about Heisenberg cuts be safely ignored and in which situations would this lead to apparent multi-agent paradoxes such as that of Frauchiger and Renner? Previously, \cite{Vilasini2022} showed that in standard quantum protocols where agents do not perform non-trivial quantum operations on each others' labs (in contrast to Wigner's Friend type protocols where they do), the Heisenberg cut can be safely and consistently ignored, recovering our usual intuitions about such quantum experiments. This still leaves open the question of when paradoxes can arise in Wigner's Friend type protocols in quantum or more general theories, and more generally on the necessary resources of a theory that lead to such paradoxes. Such resolutions and classifications of quantum paradoxes are relevant for creating foolproof, non-contradictory logical and inferential computing systems (in quantum and other non-classical theories), where one part of a larger system would need to reason about its other parts. Moreover analysing the landscape of these multi-agent paradoxes from basic theory-independent assumptions can provide valuable fundamental insights on the structure of a physical theory and on theory-independent resolutions, despite the fact that such paradoxes in quantum theory are only apparent and can be fully resolved as shown in \cite{Vilasini2022}.

Wigner's Friend type multi-agent paradoxes were first formulated in a theory-independent manner in our previous work \cite{VNdR}, and it was shown that a post-quantum theory, box-world also admits a Wigner's Friend type paradox akin to FR's proposal. Building on this prior work, we initiate an approach for a systematic study of the similarities and differences between quantum and more general theories in the context of multi-agent paradoxes. This can help understand the assumptions underpinning consistency and usability physical theories by reasoning agents, and shed further light on the much-researched foundational question of why quantum theory is special as compared to other possible theories, through the lens of multi-agent reasoning. 

In this paper, we define multi-agent paradoxes by the following assumptions: all agents know that they apply a common theory and setup to make their predictions (\textbf{common knowledge}); agents can share and use each others' knowledge about measurement outcomes as long as their measurements are compatible or jointly measurable (\textbf{reasoning about compatible agents}); setting labels that specify how a measurement is modelled (as producing classical records vs purely as a transformation on systems of the theory) and which capture the Heisenberg cut can be ignored in all agents’ statements (\textbf{setting-independence}); and, finally, contradictory conclusions about measurement outcomes cannot be reached by any agent (\textbf{non-contradictory outcomes}). These assumptions generalise those of FR's \cite{Frauchiger2018} along with additional implicit assumptions identified in future analyses of FR's arguments \cite{NL2018, Vilasini2022}. Therefore, when these assumptions are applied in quantum theory, as in the FR scenario, they can lead to a physical multi-agent paradox (more specifically, a Wigner's Friend type multi-agent paradox). However, once the assumptions are formulated in a general and theory-independent way (in the context of operational theories), can we identify more general properties implied by the existence of such paradoxes in a theory? 

Contextuality is a fundamental property which is inherent to some theories, including quantum theory. It has been linked to pre- and post-selection paradoxes~\cite{Pusey2015} which are a distinct class of paradoxes arising from studying the structure of states and measurements in the theory (rather than focusing on agents’ reasoning). In the realm of multi-agent paradoxes, quantum theory which admits the FR paradox is contextual and it has been shown that the Spekkens' toy theory which is non-contextual, does not admit an analogue of FR’s paradox~\cite{Hausmann2023}, but a post-quantum contextual GPT does~\cite{VNdR}. Apart from these specific instances, it has generally remained an open question whether any theory that admits such Wigner's Friend type multi-agent paradoxes must be contextual. Moreover questions on which forms of contextuality are necessary, and whether there are structural differences between such paradoxes in quantum vs more general theories were not systematically explored. In this paper, we formally explore this connection and associated questions.

\paragraph{Contribution.}
We prove that \textbf{any theory which admits a Wigner's Friend type multi-agent paradox also admits a logically contextual empirical model}. To achieve this, we identify a set of assumptions for defining a generalised theory-independent version of an FR-type paradox. This includes the definitions of compatibility of agents, 
modelling the concept of Heisenberg cuts (often associated with quantum theory) in this theory-independent context, specifying the rules of reasoning and the language (or set of statements) that agents are allowed to use. As a concrete example, we construct a quantum multi-agent paradox based on the 5-cycle KCBS contextuality scenario, an instance of a genuine contextuality scenario involving sequential measurements on a single system. This shows that Bell non-locality (a special case of contextuality involving multiple subsystems) is not necessary for such paradoxes. We then characterise properties of such paradoxes in quantum and general theories, such as the following.
\begin{itemize}
    \item {\bf Cyclic structure of paradoxes based on reference graphs.} Linking Wigner's Friend paradoxes to classical semantic paradoxes which are analyzed through the concept of reference graphs, we show that all Wigner's Friend multi-agent paradoxes defined here (which involve finitely many measurements) admit the structure of a Liar's cycle, necessarily having a directed cycle in the reference graph. Moreover, we present a conjecture for ruling out Wigner's Friend paradoxes based on infinite but acyclic reference graphs such as the Yablo's paradox, under certain assumptions on the theory and reasoning. 
    \item {\bf Role of post-selection in $n$-cycle scenarios.} Considering $n$-cycle contextuality scenarios with binary outcomes, we show that each extremal vertex of the associated non-disturbance polytope enables the construction of a Wigner's Friend paradox without post-selection (of which the PR-box based paradox of \cite{VNdR} is an instance). Further, any quantum Wigner's Friend paradox based on an $n$-cycle scenario must necessarily involve post-selection (as is the case in FR and the above-mentioned KCBS-based quantum paradox). 
    \item {\bf Paradoxes based on symmetric inferences.} We show that Wigner's Friend type paradoxes where each inference is a two-way implication (e.g., ``Alice upon seeing $a=0$ concludes that Bob must have seen $b=0$'' \emph{and vice versa}, as opposed to only one direction being true) cannot occur in quantum theory, although the PR-box based paradox of \cite{VNdR} is an example of such a paradox in box-world. We discuss the link between this property and the post-selection property of the previous result. 
\end{itemize}

Finally, we discuss how the assumptions going into our definition of multi-agent paradoxes can be modified or relaxed to define other classes of multi-agent paradoxes and to link them to other non-classical resources beyond logical contextuality.

\paragraph{Structure of the paper.}
In Section~\ref{sec:background}, we review the necessary background concerning two key notions relevant to the paper, namely, semantic paradoxes
and logical contextuality. We then define physical multi-agent setups and associated paradoxes in Section~\ref{sec:paradox-theory}, while discussing the assumptions therein. These constitute a class of Wigner's Friend type multi-agent paradoxes in general theories. Section~\ref{sec:paradox} first presents our main theorem which gives a general bridge between Wigner's Friend type multi-agent paradoxes and logical contextuality, and we illustrate an example of such a paradox (analogous to FR's) based on a genuine contextuality scenario, the $5$-cycle or KCBS contextuality scenario. In Section~\ref{sec:character}, we derive properties of Wigner's Friend type multi-agent paradoxes, while linking their structure to those of semantic paradoxes. We identify properties that hold in general theories (\Cref{sec: structure_general}), and ones holding in quantum theory in particular (\Cref{sec: structure_quantum}), while comparing and contrasting them. Finally, we discuss the implications of this work, and possible future directions in Section~\ref{sec:conclusions}.

\paragraph{Note added.} An earlier version of this work, including some of the results—specifically the Wigner's Friend paradox based on the KCBS contextuality scenario—was presented in the PhD thesis of Nuriya Nurgalieva \cite{NuriyaThesis}. While preparing this manuscript based on that thesis, an independent study by Walleghem et al. \cite{Walleghem2024-FR} appeared on arXiv, introducing a similar example to illustrate the non-necessity of Bell non-locality for such paradoxes, derived from quantum-theoretic assumptions. A key distinction in our analysis is that we derive the associated paradox by formulating theory-independent assumptions and explicitly considering the role of Heisenberg-cut (in)dependence, which was identified as a necessary condition for any such quantum paradox in \cite{Vilasini2022}. While our work and \cite{Walleghem2024-FR} overlap in this particular quantum-mechanical example, our primary focus extends beyond this to broader connections between FR-type Wigner's Friend paradoxes in general physical theories, as well as their underlying structural features, which were not explored in previous studies.

\section{Background}
\label{sec:background}
\begin{figure}[t]
\centering
\includegraphics[scale=0.4]{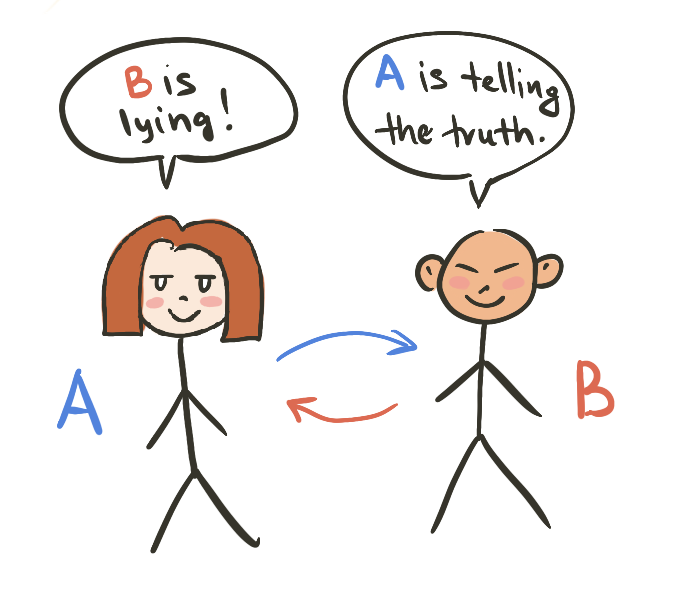}
\caption{{ \bf  A simple example of a multi-agent paradox.} The setting includes two agents, Alice and Bob. Alice says that Bob is lying; Bob says that Alice is telling the truth. Together, their statements are inconsistent with each other, meaning that there exist no global truth value assignment to their statements.}
\label{fig:epistemic-simple}
\end{figure}

In this section, we provide a pedagogical introduction reviewing two important notions in this work -- semantic paradoxes studied in classical logic, and the concept of logical contextuality studied in the quantum theoretic literature. Later in the paper, in Section~\ref{sec:paradox}, we will connect these two notions to Wigner's Friend type paradoxes.\footnote{ We note that connections between such semantic paradoxes and logical contextuality were initially derived in \cite{Abramsky11}, without considering Wigner's Friend type Scenarios.}

\subsection{Liar cycles and reference graphs}
\label{subsec:epistemic-general}

\begin{figure}[t]
\centering
\begin{subfigure}{0.9\textwidth}
\centering
\begin{tikzpicture}[-latex, auto, node distance=1 cm and 2cm, on grid, semithick, state/.style = {circle ,top color=white, bottom color =processblue!20, draw, processblue, text=blue, minimum width=1 cm}]
\node[state] (A) {$S_1$};
\node[state] (B) [right=of A] {$S_2$};
\node[] (N) [right=of B] {$\dots$};
\node[state] (C) [right=of N] {$S_{N-1}$};
\node[state] (D) [right=of C] {$S_N$};
\path (A) edge [right] (B);
\path (B) edge [right] (N);
\path (N) edge [right] (C);
\path (C) edge [right] (D);
\path (D) edge [bend left] (A);
\end{tikzpicture}
\caption{{\bf $G_\text{Liar}$:} graph of the reference relation of the Liar cycle.  Each sentence $S_k$ refers to the the one immediately after,  $S_{k+1}$, cyclically.}
\label{fig:liar-cycle}
\end{subfigure} 
\\
\begin{subfigure}{0.9\textwidth}
\centering
\vspace{1cm}
\begin{tikzpicture}[-latex, auto, node distance=1 cm and 2cm, on grid, semithick, state/.style = {circle ,top color=white, bottom color =processblue!20, draw, processblue, text=blue, minimum width=1 cm}]
\node[state] (A) {$S_1$};
\node[state] (B) [right=of A] {$S_2$};
\node[] (N) [right=of B] {$\dots$};
\node[state] (C) [right = of N] {$S_{N-1}$};
\node[] (D) [right = of C] {$\dots$};
\path (A) edge [right] (B);
\path (A) edge [bend right] (D);
\path (A) edge [bend right] (N);
\path (A) edge [bend right] (C);
\path (B) edge [right] (N);
\path (B) edge [bend right = -25] (C);
\path (B) edge [bend right = -40] (D);
\path (N) edge [right] (C);
\path (C) edge [right] (D);
\path (C) edge [bend right = -25] (D);
\end{tikzpicture}
\caption{{\bf $G_\text{Yablo}$: } graph  of the reference relation of Yablo's paradox. Each sentence $S_k$ refers to all (countably infinite) following sentences $\{S_n\}_{n >k}$.}
\label{fig:yablo}
\end{subfigure}
\caption{{\bf Reference relation graphs for paradoxical chains of statements, finite and infinite.}}
\end{figure}
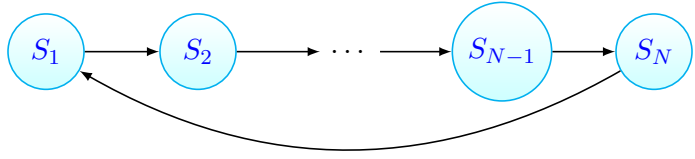
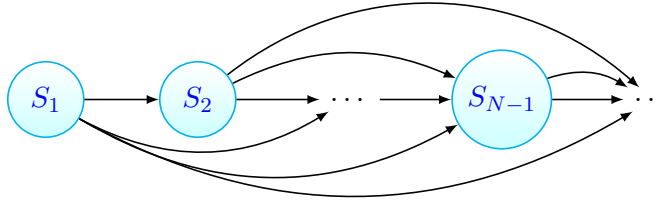

Suppose we have a set of statements that refer to the truth values of each other. Semantic paradoxes entail situations where such a set of statements are mutually inconsistent, i.e. the conditions they specify on truth values cannot be simultaneously satisfied~\cite{StanfordParadoxes}.

\paragraph{Liar's sentence and Liar cycle.}

The simplest example of a semantic paradox is the Liar's sentence $S$,
\begin{center}
    $S$: $S$ is false.
\end{center}

We can extend this construction to a chain of statements of arbitrary finite length, so-called Liar cycle, which also leads to a logical paradox. An example of such a chain consisting of $N$ statements is given below:
\begin{center}
    $S_1$: $S_2$ is true. \\
    $S_2$: $S_3$ is true. \\
    $\dots$ \\
    $S_{N-1}$: $S_N$ is true. \\
    $S_N$: $S_1$ is false.
\end{center}

Formally, this chain of statements can be expressed in terms of Boolean equations~\cite{Abramsky15}, where the paradox is captured by the inconsistency of the system of equations. If $s_k \in \{0,1\}$ is the truth value of statement $S_k$ (a Boolean variable), then we can express the setting as 
\begin{equation}
\label{eq: liar_bool}
    s_1=s_2,\quad s_2=s_3,\quad...,\quad s_{n-1}=s_n, \quad s_N=\neg s_1.
\end{equation}
By instantiating with the two options $s_1 = 0$ (false) and $s_1=1$ (true), we obtain two chains,
\begin{align}
\begin{split}
\label{eq: liar}
    &s_1 = 1 \Rightarrow s_2 = 1 \Rightarrow \dots \Rightarrow s_N = 1 \Rightarrow s_1 = 0, \\
    &s_1 = 0 \Rightarrow s_2 = 0 \Rightarrow \dots \Rightarrow s_N = 0 \Rightarrow s_1 = 1,    
\end{split}
\end{align}
which give us the contradiction. Notice that each chain in \Cref{eq: liar} is independently paradoxical i.e., a paradox ensures even if one such chain of statements is found. For this reason, we will refer to each such chain as a half-Liar's cycle and call it a full Liar's cycle when both the chains hold.

\paragraph{Structure and properties: finite chains.}

The structure of such statement chains can be encapsulated in so-called \emph{reference relation graphs (rfgs)}, where vertices represent the statements, and edges correspond to the references made in the statements~\cite{Rabern2013,Beringer2017}. 
The relationship between two statements of the theory $a$ and $b$ is captured in a \textit{reference relation} $(a,b)$ if $a$ depends on (refers to) $b$~\footnote{A formal way to define this in propositional language is: $a$ refers to $b$ if and only if there is a name $\alpha$ of $b$ such
that `$\alpha$ is true’ is a subformula of $a$. A straightforward way to generalize
this is to regard a quantified sentence as referring to all the sentences that
its instances are referring to. For more precise considerations and definitions, please refer to~\cite{Beringer2017}.}. We denote the set of such reference relations $R$. Consequently, we use these relations to construct a \textit{reference relation graph}.

\begin{definition}[Reference relation graph]
\label{def:rfg}
Given a set of propositions $\Sigma$, and a reference relation $R$ on it, the reference relation graph is defined as a directed graph $\mathcal{G} = (\Sigma,E)$ with $E=\{(a,b), (a,b)\in R; a,b\in\Sigma\}$.
\end{definition}

To identify if a particular rfg can be source of a paradox, we need to check if it admits a consistent assignment of truth values to all of its nodes. It is then possible to identify certain structural properties that all finite rfgs of a setting share as necessary and sufficient condition for its paradoxicality~\cite{Jongeling2002, Rabern2013, Beringer2017}, namely, it must contain a directed cycle. This means that one can find a semantic paradox in a finite setting\footnote{That is, involving a finite number of statements.} if and only if it can be reduced to a structure similar to the Liars cycle (the rfg of which can be seen on Figure~\ref{fig:liar-cycle}).

The intuition behind the proof is simple. Suppose that we are given a statement $S_1$ which refers to other statements $\{S_j\}_j$ (i.e., they have a dependence in truth values), which in turn refer to other sets of statements and so on. 
If the reference graph capturing these reference relations between the statements is finite and contains no self-referential loops (directed cycles), then it is possible, starting at the first statement $S_1$, to move down along the graph, in an unambiguous direction, to the ends of the branches (nodes which have no outgoing arrows). This allows to establish the truth values of the referred statements from the referencing statement in each step, in a clear order, such that subsequent truth assignments do not contradict previous ones (as there is no other directed path going from them towards earlier statements), and hence no paradox can arise.
However, the presence of a loop can result in a inconsistency as it introduces additional directed paths that defy an unambiguous ordering, and can contradict previous truth value assignments. 

\paragraph{Structure and properties: infinite chains.}
The structure of infinite chains forming infinite reference relation graphs, however, does not admit a clear characterization~\cite{Rabern2013}; it is also not agreed upon whether certain chains of this form can be categorized as a self-referential paradox~\cite{StanfordParadoxes}. The common example of such a setting is Yablo's paradox~\cite{Yablo1993}, which makes use of an infinite sequence of statements forming an infinite but acyclic reference graph (Figure~\ref{fig:yablo}),
\begin{center}
    $S_1$: $S_n$ is false $\forall n>1$ . \\
    $S_2$: $S_n$ is false $\forall n>2$. \\
    $\dots$ \\
    $S_k$: $S_n$ is false $\forall n>k$. \\
    $\dots$
\end{center}
Suppose that $S_1$ is true; then $S_2$ is inevitably false, and $\exists j>2$ such that $S_j$ is true. However, according to $S_1$, $S_j$ has to be false; we arrive to a contradiction. On the other hand, if we assume $S_1$ is false, then we are bound to conclude that $\exists j>1$ such that $S_j$ is true, and then applying the above reasoning to the infinite sequence starting from $S_j$, we would again obtain a contradiction. 
Denoting true and false using the binary values 0 and 1, we can also represent the Yablo chain as follows for $i,j\in\{1,2,3,...\}$.

\begin{align}
\label{eq: Yablo_statements}
\begin{split}
    &s_i=1 \implies s_j=0 \quad \forall j>i,\\
&s_i=0 \implies \exists j>i,\quad s_j=1.
\end{split}
\end{align}

Notice that the Yablo sequence does not yield any contradictions in the finite case (i.e., when $i,j\in \{1,...,n\}$ for finite $n$), as it can always be consistently resolved by setting all statements to be false (i.e., 0), except the last statement to be true (i.e., 1). 
While Yablo-type statement chains lead to a semantic paradox, it is not guaranteed that every infinite rfg leading to a paradox can be reduced to a Yablo-type sequence. The conjecture in~\cite{Rabern2013} however suggests that this is the case. 

\subsection{Logical contextuality}
\label{subsec:contextuality}

In this subsection, we discuss the notion of logical contextuality, and give an intuition of what it entails for an arbitrary measurement scenario in a theory. In our definitions, we follow the work of Abramsky et al~\cite{Abramsky11,Abramsky15}.

Generally, a contextual model can be understood as admitting a locally consistent, but globally inconsistent description. One intuitive example of such a (classical) construction are the Penrose stairs~\cite{Penrose1958}. Every flight of stairs is perfectly reasonable on its own, but from the global perspective of all flights combined, such configuration is impossible to build.

In the frameworks of quantum mechanics and similar physical theories (where systems are measured to obtain an outcome), we can characterize the initial setting by specifying: a set of variables one can measure on the system; subsets of variables (\textit{contexts}) which are measured together; and a set of outcomes that can be obtained after the measurement is carried out.

Note that for specifying the \textit{contexts}, we must specify when a subset of the variables can be measured together, we require a general notion of when a set of variables are \textit{compatible}, or \textit{jointly measurable}. This is formally defined below, for any theory $\mathbb{T}$ that has a well-defined notion of systems, measurement outcomes and probabilities, by extending the understanding of the notion in quantum theory (see, for example~\cite{Kunjwal2014}). 

\begin{definition}[Joint measurability or compatibility of measurements]
\label{def:meas_compatibility}

A set $\mathcal{M}_{1},...,\mathcal{M}_{k}$ of measurements with outcome sets $\{O_1, \dots, O_k\}$ acting on a set of systems $\mathcal{S}$ in a theory $\mathbb{T}$ are said the be compatible whenever the following holds. For every initial state $\rho$ of $\mathcal{S}$, there exists a measurement $\mathcal M_\text{joint}$ on $\mathcal{S}$ with outcome set $O_1 \times \dots \times O_k$ such that the probability rule of $\mathbb{T}$ yields a distribution $P(a_1,...,a_k|\rho)$ for the outcomes of $\mathcal M_\text{joint}$ when applied on $\rho$, such that the marginals of this distribution $P(a_i|\rho)=\sum_{j=1,j\neq i}^kP(a_1,...,a_k|\rho)$ equal the probability of the outcome $a_i$ when only the i$^{th}$ measurement $\mathcal{M}_i$ is performed on $\rho$. We can equivalently say that the measurements are jointly measurable, the statistics of each measurement in the set can be recovered as the marginal of a single joint distribution. 
    
\end{definition}

For the purposes of this paper, we emphasize the role of the initial state $\rho$ of the set of the systems on which the measurements are performed, and explicitly include it in definitions of measurement scenario and, consequently, empirical model. Given a set of variables, sets of compatible measurements and a set of outcomes, one can define a measurement scenario:

\begin{definition}[Measurement scenario]
\label{def: meas_scenario}
A measurement scenario is a quadruple $(X, \mathcal{C}, O, \rho)$, where
\begin{itemize}
    \item $X$ is a set of variables;
    \item $\mathcal{C}$ is a family of compatible subsets of $X$ (\textit{contexts});
    \item $O$ is a set of outcomes, or values for the variables, which can be refined to $O_x, x\in X$;
    \item $\rho$ is the initial state of the set of the systems on which the measurements are performed.
\end{itemize}
\end{definition}

Given a measurement scenario, we can start measuring all variables in a single context $C \in\mathcal{C}$, and obtain outcomes. This can be described as a mapping
\begin{gather*}
    s: C \to O^C,
\end{gather*}
where $O^C$ is the set of all possible outcomes of measurements of variables in $C$.
One can also restrict the mapping to a subcontext $C' \subseteq C$: $s|_{C'} : C \to O^{C'}$.
The function $s$ is only a local section, as it is defined only in the context $C$, and not in the whole set of variables $X$.

Bell non-locality scenarios involving non-signalling parties also fit into this framework. Consider a disjoint family of sets $\{X_i\}_{i\in I}$, where $I$ is the set of the parts of the system which are space-like separated, and $X_i$ is the set of measurements that are carried out in the part $i\in I$. Then the set of variables consists of all local variable sets $X=\cup_{i\in I} X_i$, and we define the set of contexts $\mathcal{C}$ to be those subsets of $X$ containing exactly one measurement from each part. This way, the measurements in different parts of the global system are compatible as they act on distinct subsystems, and multiple measurements in the same part are not allowed within the same context (as they may be incompatible). In this approach, Bell non-locality is a special case of contextuality, and all conclusions below can be applied to Bell non-locality scenarios. Here, we focus only on logical forms of contextuality, which subsume logical forms of Bell non-locality such as those of Hardy's scenario \cite{Hardy1993}.

Next we define empirical models for measurement scenarios by considering a probability distribution $e_C$ associated to variables in each context $C$.
\begin{definition}[Empirical model]
\label{def: emp_model}
An empirical model $e$ for the measurement scenario $(X, \mathcal{C}, O, \rho)$ is a set of local empirical probabilistic descriptions $e_C$ of all contexts $C$
\begin{gather*}
    e = \{e_C \in Prob(O^C) | C\in\mathcal{M}\}, 
\end{gather*}
where $Prob(O^C)$ is the set of probability distributions on the set $O^C$ of all possible outcomes of measurements of variables in $C$. 
\end{definition}

Note that the measurement scenario includes a particular initial state of the systems $\rho$, meaning that the empirical model is also given with respect to this particular $\rho$.

\begin{figure}[t]
    \centering
    \includegraphics[scale=0.3]{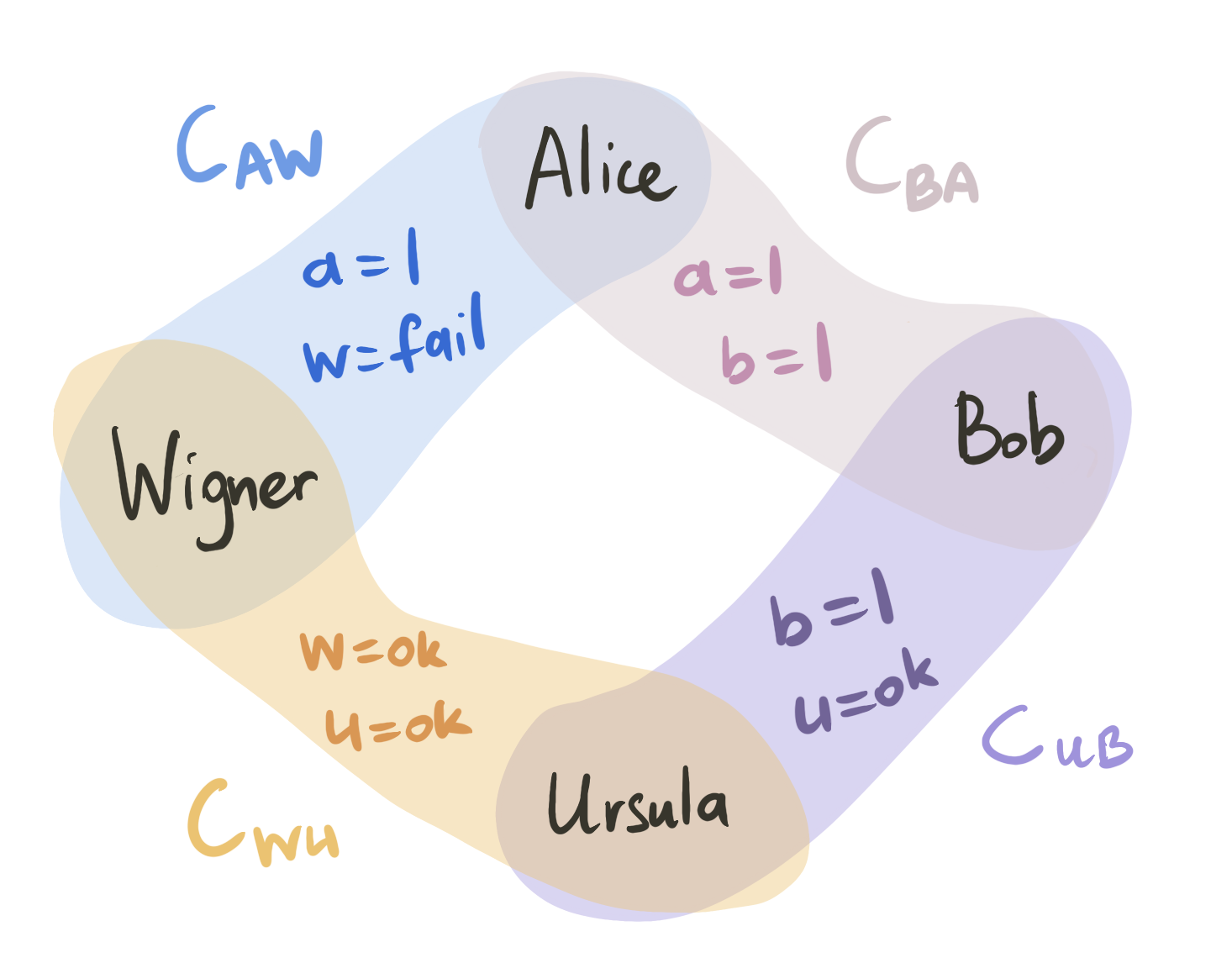}
    \caption{{ \bf Frauchiger-Renner setup~\cite{Frauchiger2018} setting and contextuality.} In Frauchiger-Renner setup, as well as in Hardy's paradox, there are four contexts $C_{UB}, C_{BA}, C_{AW}, C_{WU}$ in which the outcome assignments are made. However, there exists an outcome assignment (illustrated above) which does not belong to any compatible family of assignments, as it does not agree on the intersection between contexts $C_{AW}$ and $C_{WU}$, where the value assignment on the latter is the result of post-selection. For the specific details of the setup, please refer to Figure~\ref{fig:fr-entanglement}.}
    \label{fig:FR_contextuality}
\end{figure} 

Contextuality in general is characterised by locally consistent by globally inconsistent \textit{probability assignments} for an empirical model -- for example, a model which is non-contextual always admits a global probability assignment (a joint probability distribution on all variables in $X$) which agrees with all local probabilistic descriptions (the marginals over each context $C$). In this paper, we focus on a subset of contextuality scenarios where the contextuality is witnessed by locally consistent by globally inconsistent \textit{assignments of outcome values} for an empirical model, and are usually referred to as \textit{logically contextual}. The notion of logical contextuality is typically formulated within the sheaf-theoretic framework of \cite{Abramsky11, Abramsky15}. Here we briefly review the definitions as set out in these works, but without focusing on the underlying sheaf-theoretic structure, as this would not be directly relevant for the results of the current paper. 

In the above, $O^C$ was the set of outcomes that can be assigned to measurements in a context $C$. However, the probability distribution $e_C$ may associate zero probability to some of these value assignments, such that the set of all possible outcome assignments $s: C \to O^C$  for a measurement context $C$ is actually a subset $\mathcal{S}(C)$ of $O^C$. Explicitly, this set is defined as follows
\begin{gather*}
    \mathcal{S}(C):=\{s\in O^C: \forall C'\in \mathcal{M},\quad s|_{C\cap C'}\in \mathbf{supp}(e|_{C\cap C'})\},
\end{gather*}
where $\mathbf{supp}(e|_{C\cap C'})$ is the support of the marginal of $e_C'$ at $C\cap C'$, i.e. each possible assignment of outcomes for a context $C$ must be compatible with the support of the distribution $e_{C'}$ for all overlapping contexts $C'$. Further, it is assumed that every possible joint measurement of a context $C$ yields a possible joint outcome $s_C$ i.e., $\mathcal{S}(C)\neq \emptyset$, $\forall C \in \mathcal{C}$. 

We can now formally define what it means to have a compatible family of outcome assignments.

\begin{definition}[Compatible family of outcome assignments \cite{Abramsky15}] A compatible family of outcome assignments for a measurement scenario $(X,\mathcal{C},O,\rho)$ is a family $\{s_C\}_{C\in \mathcal{M}}$ with $s_C\in \mathcal{S}(C)$ such that for all $C, C'\in \mathcal{C}$,
\begin{gather*}
    s_C|_{C\cap C'} = s_{C'}|_{C\cap C'}.
\end{gather*}
Such a compatible family, if it exists, induces a unique global section $\mathcal{S}(X)$. 
\end{definition}

Having a contextual model means that some outcome assignment which is allowed in the model is not a member of such a compatible family. Here, one can distinguish between two types of contextuality: \textit{logical}, where we simply point out a particular incompatible outcome assignment, and \textit{strong}, where all outcome assignments are incompatible. We formally define this pair below.
 
\begin{definition}[Logical contextuality]
\label{def: logical_contextuality}
Consider a measurement scenario $(X,\mathcal{C},O,\rho)$ associated with a set $\mathcal{S}$ of possible outcome assignments as characterised above. This constitutes an empirical model. Then for any $C\in \mathcal{C}$, the empirical model is said to be \emph{logically contextual} if there exists an outcome assignment $s\in \mathcal{S}(C)$ that does not belong to any compatible family of outcome assignments. 
\end{definition}
 
\begin{definition}[Strong contextuality]
Consider a measurement scenario $(X,\mathcal{C},O,\rho)$ associated with a set $\mathcal{S}$ of possible outcome assignments as characterised above. Then for any $C\in \mathcal{C}$, the corresponding empirical model is said to be \emph{strongly contextual} if it is logically contextual at all $s\in \mathcal{S}(C)$. 
\end{definition}
 
Then, every strongly contextual empirical model is logically contextual and every logically contextual empirical model is contextual. Hence, none of these models admits a joint probability distribution over $X$. However, they may admit a global section $\mathcal{S}(X)$ in terms of the outcome assignments as we now explain. Logical contextuality by itself does not rule out the existence of a global section $\mathcal{S}(X)$ i.e., a possible assignment of outcome values to all the measurements. This is because a scenario may be logically contextual with respect to a particular assignment $s\in \mathcal{S}(C)$ for a context $C$, but this only rules out the existence of a compatible family of assignments that contain $s$, and we can nevertheless have a compatible family of assignments that contains a different assignment $s'\in \mathcal{S}(C)$ for measurements in $C$. An example is the measurement scenario for Hardy's paradox \cite{Hardy1993}/the original FR paradox \cite{Frauchiger2018}, where the paradox only arises when a particular set of outcomes are observed for the measurement context where both agents measure in the ``X'' basis i.e., when $u=w=ok$ is obtained in the FR case, which is why post-selection is required for the original FR paradox (see Figure~\ref{fig:FR_contextuality}, and the results of \Cref{sec: structure_quantum} regarding post-selection). Strong contextuality on the other hand, implies the non-existence of a global section i.e., $\mathcal{S}(X)=\emptyset$. An example is the PR box \cite{Popescu1994}/our PR box paradox \cite{VNdR} where the paradox ensures irrespective of the outcomes obtained in any of the measurement contexts and consequently no post-selection is required for obtaining the paradox.
We refer the reader to~\cite{Abramsky15} for an analysis of Hardy's paradox, the PR box examples and their logical/strong contextuality. Our results of \Cref{sec: paradox_implies_contextual} imply in particular, the logical contextuality of the FR \cite{Frauchiger2018} and PR-box based \cite{VNdR} Wigner's Friend paradoxes.

\section{Multi-agent setups: predictions, reasoning and paradox in general theories }
\label{sec:paradox-theory}

Having reviewed the notion of paradox in its logical and abstract sense, let us now consider what it means to encounter a paradox in a physical theory. While logical paradoxes can occur in classical scenarios involving lying agents, the paradoxes we consider in physical theories involve reliable agents who all apply the same physical theory in their reasoning. Thus, the paradox arises due to fundamental properties of the theory. Nevertheless, we will see that the paradoxical structures obtained are similar to those of classical liar paradoxes. Our analysis will apply to theories with a well-defined notion of systems, states of systems, transformations on systems (including the identity transformation, ``doing nothing'') and measurement outcomes along with a rule for computing the probabilities of measurement outcomes\footnote{Which are associated with positive and normalised probability distributions respecting the standard rules of conditional probabilities.} in a given scenario. Therefore these definitions and considerations can be instantiated in a number of existing frameworks for describing such general physical theories, such as formalisms for generalised, operational or categorical probabilistic theories \cite{Hardy01, Chiribella2010,dAriano2017,coecke2016categoricalquantummechanicsi, Coecke_Kissinger_2017,Ormrod2023}.

We formulate a paradox in a physical theory as a contradiction between a set of assumptions that refer to how agents in the theory conduct their reasoning about a given protocol or scenario. This approach is analogous to the style of the no-go theorem by Frauchiger and Renner~\cite{Frauchiger2018} for reasoning agents in quantum theory. However, following a recent work \cite{Vilasini2022}, we will make explicit certain assumptions (relating to the modeling of measurements and choices of Heisenberg cuts) that are implicit in FR's result. Although \cite{Vilasini2022} demonstrates that making these assumptions explicit and rigorous fully resolves FR-type apparent paradoxes and ensures the consistency of agents' use of quantum theory, our motivation to analyze such paradoxes stems from their independent capacity to shed light on the fundamental properties of the theory, irrespective of whether they are resolved in a physical theory. Thus, for brevity, we will continue to refer to these as paradoxes, even though they are only apparently so and can be resolved without regarding the theory as fundamentally pathological.

We therefore begin by explaining the role of measurements and their modeling in multi-agent scenarios before defining the general setup.

\subsection{Modelling measurements: generalising the Heisenberg cut}
\label{sec:meas_model}
\begin{figure}[t]
    \centering
    \includegraphics[width=0.6\linewidth]{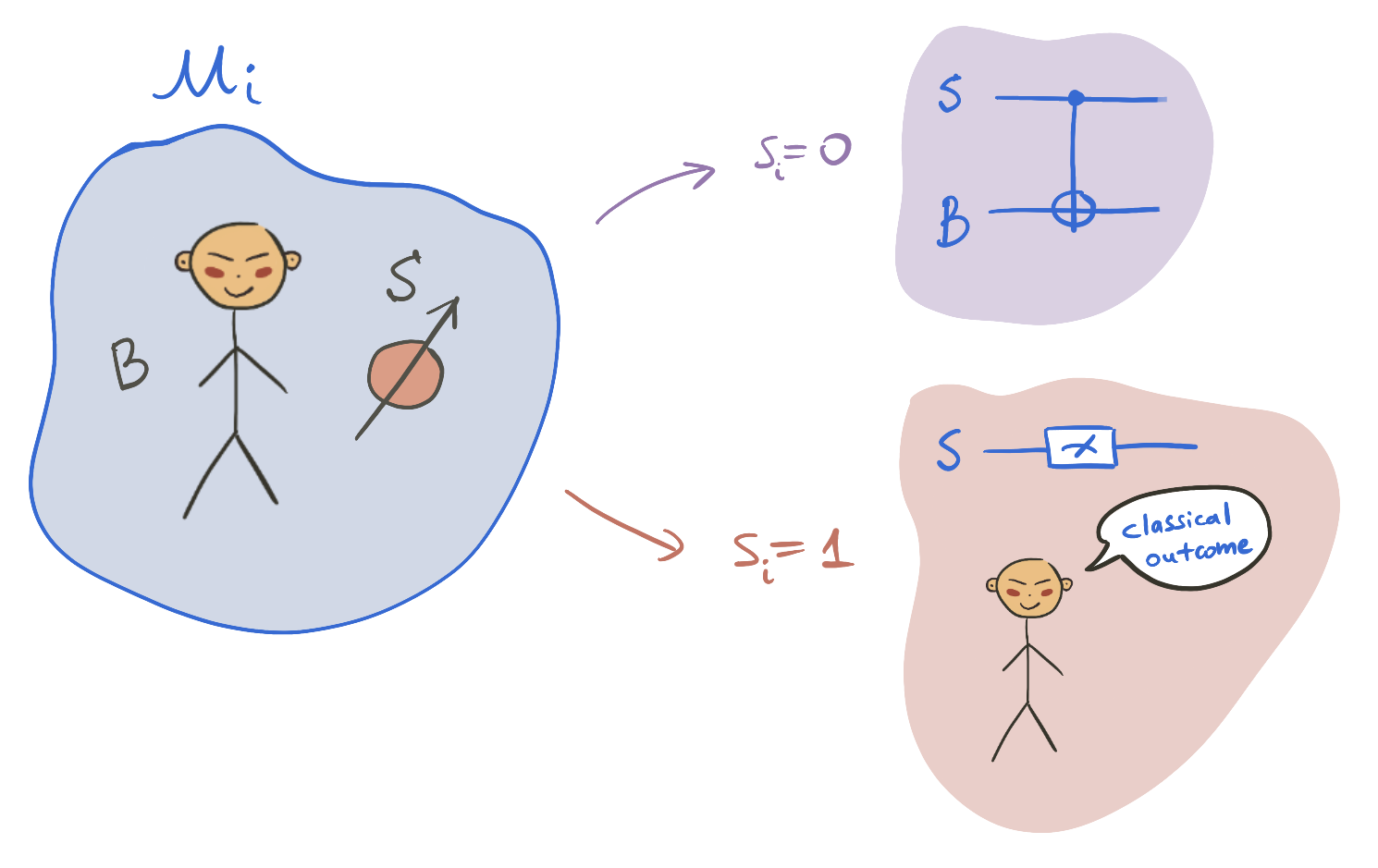}
    \caption{{\bf Two views on measurement in quantum theory.} Suppose that Bob is measuring a quantum system $S$ (measurement $\mathcal M_i$). There are two ways this measurement can be modelled in quantum theory, depending on whether Bob is considered to be a quantum system (under the Heisenberg cut) or not. This choice can be seen as a choice of a \textit{setting variable} $s_i$. If $s_i=0$ then the evolution is modelled as a unitary one; if $s_i=1$ then it is assumed that classical outcomes are observed (identified by the measurement projectors). This formalisation of measurement models/Heisenberg cuts in terms of settings was first introduced for quantum theory in \cite{Vilasini2022}. Here we consider an extension of the concept to general physical theories, and relate it to contextuality of the theory under other assumptions.}
    \label{fig:measurement-views}
\end{figure}

Consider a physical theory, a system $S$ in the theory, and a state $\rho_S$
of the system. The transformations allowed by the theory map between possible states of the system $S$. A measurement, however, is a special kind of operation that translates between states of the theory and classical outcomes. This connection to classical data is necessary to discuss the observable predictions of the theory. Alternatively, one may view the measurement as a physical transformation within the theory, including the apparatus and agent performing it as physical systems. See also previous works \cite{VNdR, Ormrod2023} for formal discussions on these different views of a measurement in general physical theories.

In quantum theory, these two models of a measurement are informed by the projection and unitarity postulates. When considering an agent who observes a classical measurement outcome, in usual quantum theory, we apply the projection postulate to describe their post-measurement state and use the same projectors in computing the outcome probabilities via the Born rule. On the other hand, when regarding an agent's lab as a closed quantum system (e.g., the view of an outside agent), the agents' measurement is modeled as a unitary evolution of the closed system that is their lab. This is a purely quantum description not associated with any classical records.
Thus, a physical theory can associate two types of descriptions to a measurement $\mathcal{M}_i$, depending on whether it is regarded as producing classical records (denoted as $s_i=1$) or as the evolution of a closed system of the theory (denoted as $s_i=0$). Following \cite{Vilasini2022}, we refer to these $s_i$ as ``settings''. This distinction can be understood as representing the Heisenberg cut.\footnote{While we have described the setting $s_i=1$ in the quantum case in terms of the projection postulate, this is not always necessary. The measurement model associated with classical records ($s_i=1$) can be understood as requiring knowledge of the projectors (or basis) identifying the classical outcome (needed to compute their Born rule probabilities) even when applying a unitary description, while the $s_i=0$ model in the quantum case does not invoke the knowledge of these projectors/basis. Thus it is also possible to resolve the quantum FR paradox without assuming the projection postulate, as discussed in \cite{Vilasini2022}.} In quantum theory, such a cut specifies which aspects of a given experiment are modeled classically versus quantum mechanically. Similarly, we can distinguish in a general physical theory the aspects modeled classically versus using that theory. In certain theories, this distinction may be irrelevant or the two cases may coincide, such as in fully classical theories.

In quantum theory, the $s_i=0$ setting corresponds to a unitary description of the measurement $\mathcal{M}_i$. In more general theories, the concept of an information-preserving memory update was introduced in our previous work \cite{VNdR} to capture measurements as evolutions of the theory. In a multi-agent context, this can be seen as the physical evolution associated with a measurement from the perspective of an outside agent (as opposed to that of an inside agent who performs the measurement and observes classical outcomes). While the choice of such settings does not matter in typical quantum experiments, Wigner's original thought experiment illustrates that this choice does have empirical consequences in quantum scenarios where the outside agent can perform arbitrary quantum operations on the lab of the inside agent.
It is shown in \cite{Vilasini2022} that explicitly accounting for the implicit choices of such Heisenberg cuts in the theory's predictions and agents' statements fully resolves multi-agent paradoxes in quantum theory. Therefore, to recover any FR-like apparent paradox in quantum theory, it is necessary to impose the independence of the predictions and statements from the choice of such settings (or Heisenberg cuts) in scenarios where this choice matters.

We will therefore model setting-independence as an explicit condition in our general formulation of multi-agent paradoxes in physical theories, and take these settings into account when defining the predictions of our model. However, our results do not rely on the exact description of the setting $s_i=1$ and $s_i=0$ models of a measurement, it does not require $s_i=0$ to correspond to an information-preserving memory update according to \cite{VNdR} or a reversible unitary evolution as in quantum theory, nor does it require $s_i=1$ in the quantum case to correspond to the projection postulate, these are only examples. These labels capture the distinction between a measurement regarded as producing classical records vs an evolution of systems of the theory, without specifying their exact description (in the interest of generality of the results). In certain specific theories, these two evolutions can coincide.

\subsection{Multi-agent setup}

We first define the physical setup that we consider. We assume that the theory under consideration has a well-defined notion of time according to which operations can be ordered. We first define a multi-agent setup in any physical theory that admits a well-defined notion of systems, states, operations and measurements, and provides a rule for computing empirical probabilities for classical outcomes of measurements. This is based on a definition of Extended Wigner's Friend Scenarios originally proposed in \cite{Vilasini2022} for the quantum case. This definition encompasses finite multi-agent protocols where agents' memories (in which they store the measurement outcome) or equivalently agents' labs are modelled as physical systems of the theory, and where one agent can have full control (perform arbitrary operations of the theory) over the labs of other agents in the scenario.

\begin{definition}[Multi-agent setup]
\label{def: MAsetup}
A multi-agent setup $\mathcal{M}\mathcal{A}$ in a theory $\mathbb{T}$ consists of 
\begin{itemize}

    \item A finite set of agents $\mathcal{A}:=\{A_i\}_{i=1}^N$,
    \item A finite set of systems $\mathcal{S}:=\{S_j\}_{j=1}^m$ of the theory,
    \item For each agent $A_i$, an additional system $L_i$ of the theory that models the memory/rest of the lab of the agent with $\mathcal{L}:=\{L_i\}_{i=1}^N$. We will simply refer to this as the memory and consider that each memory $L_i$ is initialised to some fixed state $\rho^0_{L_i}$.
    \item A set of measurements $\{\mathcal{M}_i\}_{i=1}^N$, where each measurement $\mathcal{M}_i$ is conducted by an agent $A_i$ on some subset $\mathcal{S}_i\subset \mathcal{S}\cup \mathcal{L} \backslash \{L_i\}$ of the all the systems and memories of the remaining agents, at some time $t_i$. The result of this measurement is stored in the memory $L_i$ of the agent $A_i$ and we take $t_i<t_j$ for $i<j$.
 \item For each agent $A_i$, a finite set $\mathtt{O}_i:=\{0,1,...,d_{\mathtt{S}_i}-1\}$ in which their outcome $a_i$ takes values.
    \item An initial joint state $\rho_{S_1,...,S_m}$ of all the systems in $\mathcal{S}$ prepared at time $t_0<t_i$.
\end{itemize} 
\end{definition}

We note that the original definition of \cite{Vilasini2022} additionally includes the possibility for operations to be performed between measurements, we ignore this here as the operation can be absorbed into the definition of the measurements. This way, the focus is entirely on a set of measurements and their properties such as joint measurability, which will be relevant for our main results.

{\bf Predictions of the theory.} The above definition does not refer to the settings ($s_i\in\{0,1\}$) used to describe each measurement, leaving an ambiguity in how the measurements are to be modelled. This becomes important to specify when considering the predictions or probabilities of measurement outcomes in a given multi-agent set-up. Agents will use such probabilities in their reasoning about each others' knowledge and to arrive at logical statements.

We assume that the physical theory gives a probability rule to compute probabilities of measurement outcomes given a description of the states, transformations (or channels) and measurements in the setup along with a specification of how each measurement is modelled. In quantum theory, the Born rule does this job. In the following, for any set $a_{j_1},...,a_{j_p}$ of outcomes in a setup, we will conveniently denote this as a vector $\vec{a}_j:=(a_{j_1},...,a_{j_p})$, and assigning values $v_i$ to each $a_i$ in the set, is then denoted as $\vec{a}_j=\vec{v}_j$. In a mild but harmless abuse of notation, we will use $a_i\in \vec{a}_j$ to denote that the outcome $a_i$ belongs to the set of outcomes that defines the vector $\vec{a}_j$, and $a_i\in\vec{a}_j\cup\vec{a}_l$ to denote $a_i$ that belongs in the union of the sets associated with two such vectors.

\begin{definition}[Default predictions of a multi-agent setup]
\label{def: default_predictions}
 A default prediction of a multi-agent setup $\mathcal{MA}$ involving $N$ agents is a conditional probability $P(\vec{a}_j|\vec{a}_l,\vec{s})$ where $\vec{a}_j$ and $\vec{a}_l$ denote outcomes of disjoint subsets of measurements in  $\mathcal{MA}$, where the former set is non-empty and the latter possibly empty. Further, $\vec{s}=(s_1,...,s_N)$ is a vector encoding the values of the settings of all $N$ measurements in the scenario, where the settings take values as follows: for all $a_i\in \vec{a}_j \cup \vec{a}_l$, $s_i=1$ and $s_i=0$ otherwise. The prediction is computed through the probability rule of the theory for the given setup and choice of settings. 
\end{definition}

The above concept of a default prediction is based on a definition of a setting-conditioned prediction introduced in \cite{Vilasini2022}, where the default value assignment corresponds to the implicitly used model of measurements in existing Wigner's Friend arguments and no-go theorems (including FR). Notice that this default assignment is relative to the given prediction (or probability), and not a global assignment of setting choices for all measurements. Here, all the classical outcomes ($\vec{a}_j \cup \vec{a}_l$) appearing in the given probability expression $P(\vec{a}_j|\vec{a}_l,\vec{s})$ are assigned setting 1, which means that the associated measurements are regarded as producing classical records (indeed the very distribution under consideration refers to these records). On the other hand, all outcomes which do not appear in the probability, the corresponding measurement is modelled purely as an evolution of the given physical theory (in the quantum case, as a unitary evolution). The description of this evolution, is to be specified by the theory, the details of which will not affect our results.  See also \cref{fig:fr-entanglement} for details of the settings and default predictions in the quantum FR scenario.

In \cite{Vilasini2022}, other setting choices from the default assignment were also considered in order to compare different interpretations of quantum theory. In this work, only the default setting assignment described above will be relevant, we will therefore refer to the default predictions simply as predictions of a multi-agent setup.

\paragraph{Statements of a physical theory.} 
Based on the measurement outcomes agents observe and using the correlations encoded in the predictions they compute, they can make various statements. For example, ``I observed the outcome $a=1$.'', or ``Given that I observed the outcome $a=0$ and given my knowledge of the protocol, I am certain that Bob will observe the outcome $b=1$''. Notice that the first statement is only possible if the theory assigns a non-zero probability to the outcome $a=1$ and the second entails that the theory assigns a conditional probability of 1 for $b=1$ given $a=0$ and other information about the protocol/scenario.

Building on this intuition, we now specify the structure of statements we will consider in a concrete multi-agent setup in a physical theory, while determining whether the theory can or cannot lead to paradoxes. This forms the language of the theory, for the given setup. 

\begin{figure}
\centering
    \includegraphics[scale=0.3]{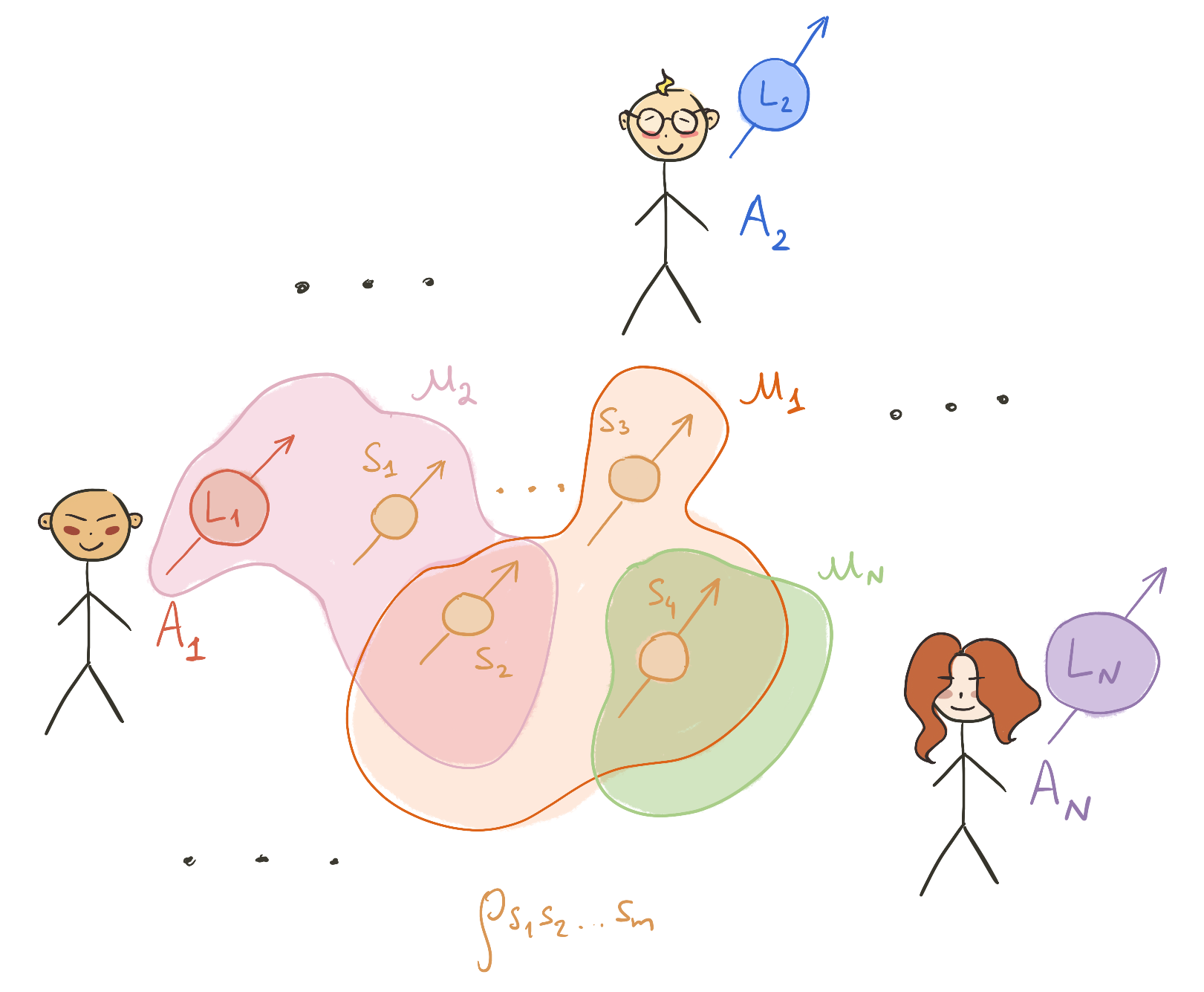}
    \caption{{ \bf Multi-agent setup.} Consists of: a set of agents, a set of their memories, a set of additional systems, and a set of measurements agents perform. The measurement of one agent can act on any subset of the memories of other agents and the additional systems, and for simplicity we assume each agent carries out only one measurement. Each memory $L_i$ is initialised to some state $\rho^0_{L_i}$ and the systems start out in the initial state $\rho_{S_1,...,S_m}$.}
        \label{fig:multi-setup}    
\end{figure}

\begin{definition}[Set of statements of a setup]
\label{def:language_v2}
Let $\mathcal{MA}$ be a multi-agent setup in a theory $\mathbb{T}$. The set $\Sigma_{\mathcal{MA}}$ of statements of the setup $\mathcal{MA}$ consists of all statements of the following types that can be derived in $\mathcal{MA}$ using $\mathbb{T}$.

    \begin{itemize}
        \item \textit{Atomic outcome statements}, which concern the outcomes of a set of measurements: Consider $\phi_{\vec{s}}:=\Big(\vec{a}_j = \vec{v}_j\Big)_{\vec{s}}$, where $\vec{a}_j$ denotes a set of measurement outcomes (each outcome associated with one agent) in the scenario and $\vec{v}_j$ denotes a set of values, one associated with each outcome in $\vec{a}_i$ and $\vec{s}$ denotes a setting choice for the scenario under which the statement is made. Then, $\phi_{\vec{s}}\in \Sigma_{\mathcal{MA}}$ whenever the corresponding prediction of the scenario is non-zero, i.e., $P(\vec{a}_j=\vec{v}_j|\vec{s})>0$.
        \item \textit{Atomic inferences}, which combine two atomic outcome statements into an inferential chain: Consider $\psi_{\vec{s}}:=\Big((\vec{a}_l = \vec{v}_l) \Rightarrow (\vec{a}_j = \vec{v}_j)\Big)_{\vec{s}}$. Then, $\psi_{\vec{s}}\in \Sigma_{\mathcal{MA}}$ whenever the corresponding prediction of the scenario satisfies $P(\vec{a}_j=\vec{v}_j|\vec{a}_l=\vec{v}_l,\vec{s})=1$ or equivalently $P(\vec{a}_j=\neg \vec{v}_j, \vec{a}_l=\vec{v}_l|\vec{s})=0$, where $\neg$ denotes negation.
    \end{itemize}
    The setting unlabelled version of a statement $\phi_{\vec{s}}\in \Sigma_{\mathcal{MA}}$ corresponds to the same statement but with the subscript $\vec{s}$ removed. 
\end{definition}

\subsection{Compatibility of measurements and agents}

In \Cref{def:language_v2} of the statements of the theory, atomic inferences can relate any two subsets of measurement outcomes $\vec{a}_j$ and $\vec{a}_l$. However, in the FR scenario, such inferences are only made between compatible (or jointly measurable) pairs of measurements. Therefore, to define the appropriate generalisation of an FR-style multi-agent paradox, we must first define the compatibility of agents through their measurements. For this, we note that although different measurements $\mathcal{M}_i$ and $\mathcal{M}_j$ in a multi-agent setup generally can act on different subsets $\mathcal{S}_i$ and $\mathcal{S}_j$ of systems (that can include the labs of other agents), we can trivially extend all measurements to measurements on the total set $\mathcal{S}\cup \mathcal{L}$ of all systems and labs by appending the identity transformation on the remaining systems. In the following definition, we consider the measurements of a multi-agent setup as acting on the same space through such a trivial extension.

In order to meaningfully speak about measurement compatibility of a subset of measurements in a multi-agent setup, and thereby meaningfully speak about measurement contexts and contextuality,
we need to create an additional bridge. This is because in a multi-agent setup, all $N$ measurements are actually performed in every given realization or experimental round of the protocol specified by the setup. On the other hand, when testing compatibility or joint measurability of a set of measurements in the usual sense, only the measurements in the given set are applied and their probabilities are simulated through a single joint measurement \Cref{def: meas_scenario}. Therefore, the following definition of compatibility in a multi-agent setup includes an additional step before invoking the usual definition of joint measurability/compatibility of measurements. We state the definition and then illustrate it with concrete and intuitive examples.

\begin{sloppypar}

\begin{definition}[Compatibility of measurements in a multi-agent setup]
\label{def:compatibility_v2}
Consider a multi-agent setup $\mathcal{MA}$ involving $N$ agents $\{A_i\}_{i=1}^N$ and a corresponding set of $N$ measurements $\{\mathcal{M}_i\}_{i=1}^N$. Let $\rho$ denote the joint initial state of all systems and memories $\mathcal{S}\cup \mathcal{L}$ in the setup. A subset $\{\mathcal{M}_{j_1},...,\mathcal{M}_{j_p}\}$ of the $N$ measurements is
said to be compatible in the setup $\mathcal{MA}$ if the following conditions hold. For each $\mathcal{M}_{j_k}$, there exists a corresponding measurement $\mathcal{M}'_{j_k}$ defined on $\mathcal{S}\cup \mathcal{L}$ and having the same outcomes $a_{j_k}$ as the original measurement such that 

\begin{itemize}
    \item The default prediction $P(a_{j_1},...,a_{j_p}|\vec{s})$ of the original setup is equivalent to the probability $P(a_{j_1},...,a_{j_p}|\rho, \{\mathcal{M}'_{j_k}\}_{k=1}^p)$ of applying the primed measurements to $\rho$, $$P(a_{j_1},...,a_{j_p}|\vec{s})=P(a_{j_1},...,a_{j_p}|\rho, \{\mathcal{M}'_{j_k}\}_{k=1}^p),$$
    and the marginals $P(a_{j_k}|\vec{s})=\sum_{a_{j_l},l\neq k} P(a_{j_1},...,a_{j_p}|\vec{s})$ of the former equal the corresponding marginals $P(a_{j_k}|\rho, \{\mathcal{M}'_{j_k}\}_{k=1}^p)=\sum_{a_{j_l},l\neq k} P(a_{j_1},...,a_{j_p}|\rho, \{\mathcal{M}'_{j_k}\}_{k=1}^p)$ of the latter.

      \item The primed set of measurements $\{\mathcal{M}'_{j_k}\}_{k=1}^p$ are jointly measurable or compatible (see Definition~\ref{def:meas_compatibility}).
\end{itemize}

If a subset of measurements $\{\mathcal{M}_{j_1},...,\mathcal{M}_{j_p}\}$ is compatible, we will say that the corresponding subset $\{A_{j_1},...,A_{j_p}\}$ of agents performing those measurements form a compatible set of agents. 

\end{definition}
\end{sloppypar}
Note that even though a multi-agent setup specifies a particular input state, the concept of compatibility of measurements and agents in the setup is defined relative to all possible input states as joint measurability applies to all initial states (\Cref{def:meas_compatibility}).

\paragraph{Examples of compatible and incompatible measurements in a multi-agent setup.}  In quantum theory, one generally associates joint measurability to the compatibility of the measurement basis, in particular regarding two measurements $\mathcal{M}_1$ and $\mathcal{M}_2$ of a system in the same basis as being compatible. This is only meaningful when $\mathcal{M}_1$ and $\mathcal{M}_2$ are the only measurements being applied on the systems. In a multi-agent setup where two agents perform these measurements, in the presence of further agents who may measure in other bases, we can no longer regard these measurements as being compatible as other intervening measurements can disturb the system.

Consider a simple quantum multi-agent setup illustrated in \Cref{fig:agent-compatibilityA} with an initial state $\ket{\psi_{RS}}$ of two qubits $R$ and $S$. The set-up involves four agents Alice, Bob, Charlie and Debbie and we take their memories $A$, $B$, $C$ and $D$ to be initialied in the $\ket{0}$ state. Alice and Debbie both measure a system $R$ in the $\{\ket{0},\ket{1}\}$ basis but at different times. At a time between Alice and Debbie's measurement, Charlie measures $R$ in the complementary, Hadamard basis $\{\ket{+},\ket{-}\}$. Additionally, we can have another agent Bob who measures a system $S$ entangled with $R$. All agents store the measurement outcome in their respective memories. In quantum theory considering projective measurements and taking the memory systems to be initialised to $\ket{0}$, setting $s_i=0$ for a measurement $\mathcal{M}_i$ corresponds to a unitary evolution $U_{\mathcal{M}_i}$ which is a CNOT (from the system to the memory) in the basis of the measurement while setting $s_i=1$ corresponds to additionally applying the measurement projectors after this CNOT (e.g., projectors of $\ket{00}$, $\ket{11}$ on the system and memory for a computational basis measurement) \cite{Vilasini2022}.\footnote{Note that applying a projective measurement e.g., $\{\ket{0},\ket{1}\}$ on the system alone is operationally equivalent in terms of probabilities, to applying the unitary model $U_{\mathcal{M}_i}$ of the measurement (in this case, a usual CNOT from system to memory) and then the projective measurement $\{\ket{00},\ket{11}\}$ on the system and memory. This corresponds to the information-preserving property of the quantum memory update $U_{\mathcal{M}_i}$ \cite{VNdR, Ormrod2023}.} 

Consider the default prediction (\cref{def: default_predictions}) for Alice and Debbie, which is explicitly given as $P(ad|\vec{s}):=P(ad|s_A=s_D=1,s_C=s_B=0)$. Now let $\mathcal{M}'_D:=\mathcal{M}_D\circ U_{\mathcal{M}_C}$ which denotes the sequential composition of the unitary description of Charlie's measurement (since $s_C=0$) followed by Debbie's projective measurement (since $s_D=1$). This defines a new binary outcome projective measurement, w.l.o.g. we can view this as having the same outcome $d$ as $\mathcal{M}_D$. Since Bob acts on a distinct subsystem than all other agents, the non-signalling property of quantum theory allows us to ignore Bob when computing predictions for other agents. Therefore the default prediction $P(ad|\vec{s})$ can equivalently be obtained by applying only $\mathcal{M}_A$ and $\mathcal{M}'_D$ to the initial state $\ket{\psi_{RS}}$. However, the unitary $ U_{\mathcal{M}_C}$ now transforms the basis of Debbie's measurement such that $\mathcal{M}_A$ and $\mathcal{M}'_D$ are not jointly measurable in the usual sense, according to \Cref{def:meas_compatibility}. This highlights that two agents measuring in the same basis is insufficient for compatibility between them in a multi-agent setup, and is formalised through the first condition of our \Cref{def:compatibility_v2}.

In comparison, consider another setup illustrated in \Cref{fig:agent-compatibilityB}. Here Ursula acts jointly on Alice's system $R$ and memory $A$ after Alice's measurement in the computational basis $\{\ket{0}_R,\ket{1}_R\}$. Here there are no intermediate agents, applying our definition, it is easy to see that if Ursula measures in the $\{\ket{00}_{AR},\ket{11}_{AR}\}$ basis, she would be compatible with Alice (physically, this would correspond to Ursula simply asking Alice what outcome she obtained) but if she measures in the Bell basis which is particular includes the elements $\frac{1}{\sqrt{2}}(\ket{00}\pm \ket{11})_{AR}$, then she would not be compatible with Alice (as Ursula would be a superagent who is ``Hadamarding'' Alice's brain, in colloquial terms). In both the examples of \Cref{fig:agent-compatibility}, it is easy to verify that Bob is compatible with all other agents according to our definition.

We note that in any multi-agent setup within a theory that has a well-defined compositional structure, one can generalise the idea of the first example discussed above, which involves ``grouping together'' a measurement with the setting $s_i=0$ descriptions of intermediate measurements, to construct the primed measurements needed in the first condition of our \Cref{def:compatibility_v2}. The primed measurements act in the same order (here $A$ before $D$) and have the same outcomes as the original measurements, and yield equivalent probabilities. While the ability to make this construction is not relevant to our main results, this highlights that the first condition of our \Cref{def:compatibility_v2} can generally be satisfied in multi-agent scenarios in a large class of theories, therefore (in)compatibility of measurements in a multi-agent setup according to our definition \Cref{def:compatibility_v2}, does generally indicate (in)compatibility or joint (non-)measurability in the usual sense of this concept, \Cref{def:meas_compatibility}.

\begin{figure}
    \centering
    \begin{subfigure}{0.49\textwidth}
        \centering
        \includegraphics[width=1\linewidth]{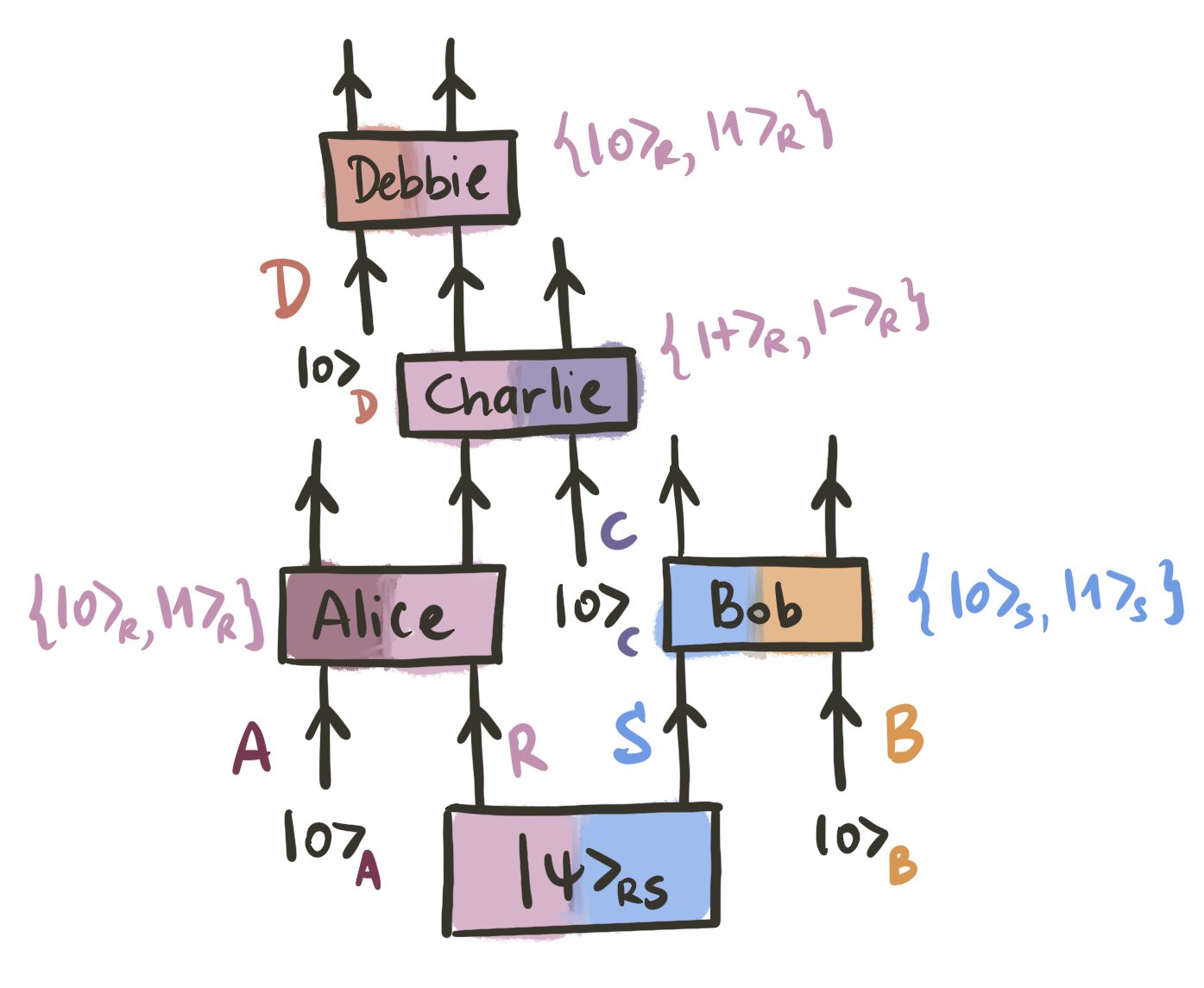}
        \caption{{\bf Compatibility in sequential measurement setups.} Two qubits $R$ and $S$ are prepared in the initial state $\ket{\psi_{RS}}$. Alice and Debbie both measure a system $R$ in the $\{\ket{0},\ket{1}\}$ basis but at different times; between their measurements, Charlie measures $R$ in the complementary, Hadamard basis $\{\ket{+},\ket{-}\}$. Additionally, we can have another agent Bob who acts on a system $S$ entangled with $R$. Note that Alice and Debbie cannot be regarded as performing compatible measurements, as Charlie's intervening measurement induces disturbance. This means that two agents measuring in the same basis is insufficient for compatibility between them in a multi-agent setup, as formalised through the first condition of our \Cref{def:compatibility_v2}.}
        \label{fig:agent-compatibilityA}
    \end{subfigure}
    \
    \begin{subfigure}{0.49\textwidth}
        \centering
        \includegraphics[width=1\linewidth]{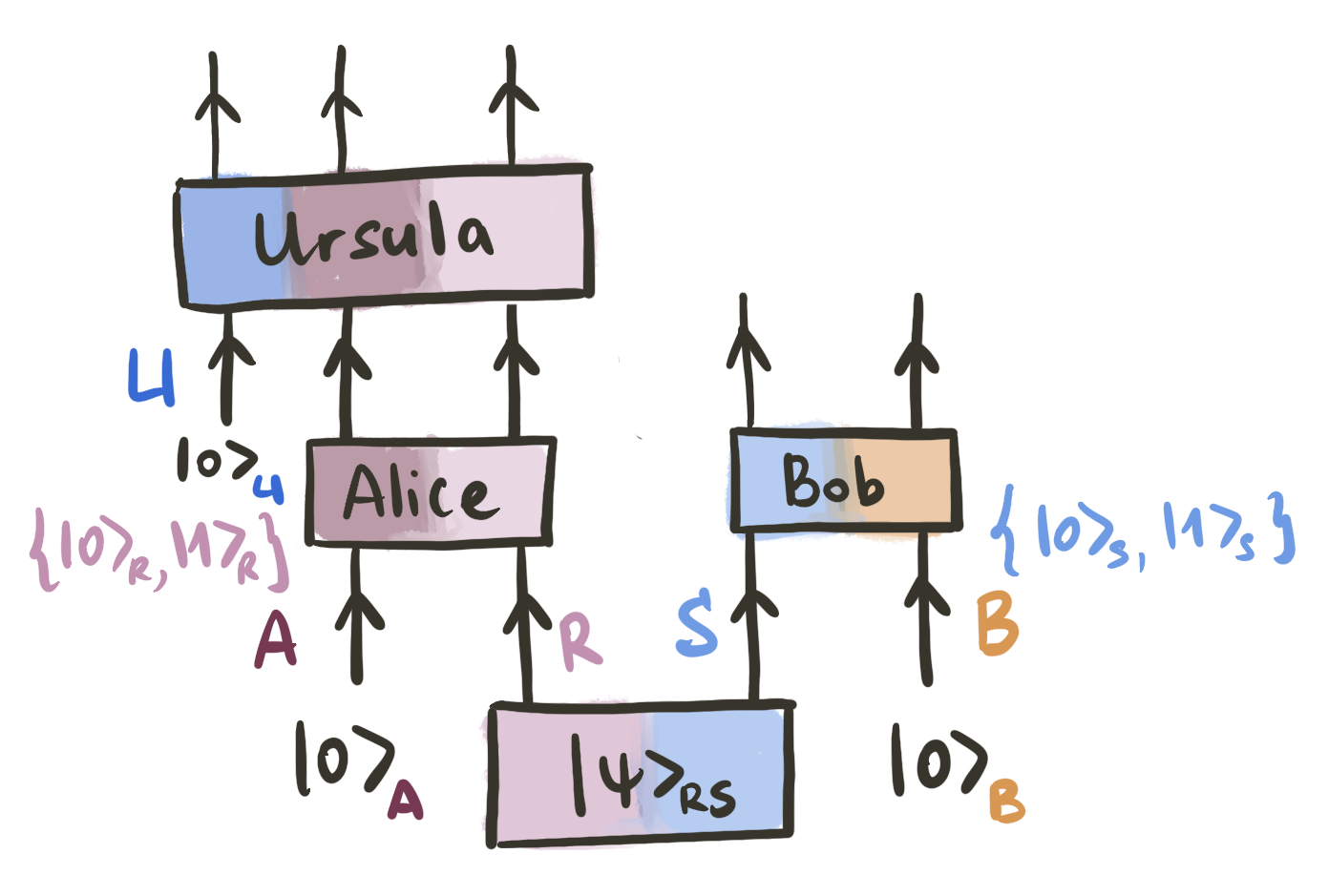}
        \caption{{\bf Compatibility and superagents.} In this example, Ursula, posing as a `superagent', acts jointly on Alice's system $R$ and memory $A$ after Alice's measurement in the computational basis $\{\ket{0}_R,\ket{1}_R\}$. If she measures in the $\{\ket{00}_{AR},\ket{11}_{AR}\}$ basis, they would be compatible; if she chooses to measure in the Bell basis, they are not. Note that Bob is compatible with both Alice and Ursula, independent of their measurement choices, as he acts on a different part of the system.}
        \label{fig:agent-compatibilityB}
    \end{subfigure}
    \caption{{\bf Compatibility of measurements/agents in a multi-agent setup.} We explain \cref{def:compatibility_v2} of compatibility in two cases through quantum examples: when simple sequential measurements on a system are made, and when superagents measuring other agents' memories are involved. All memory systems $A$, $B$, $C$, $D$ and $U$ are taken to be initialised in the $\ket{0}$ state.}
    \label{fig:agent-compatibility}
\end{figure}

\subsection{Knowledge axioms}

The knowledge of agents is governed by a set of axioms, or rules, which describe how it is operated~\cite{LogicStanford}\footnote{For a quick recap, see~\cite{NL2018}.}. We first discuss these axioms, at an abstract level of generic statements and then instantiate these concepts in a physical theory by linking them to the predictions of the theory for a given multi-agent setup.

We state the axioms relevant to how agents reason and pass on their knowledge. This can be formalised by associating a knowledge operator $K_i$ to each agent $A_i$, such that $K_i\phi$ denotes that agent $A_i$ knows the statement $\phi$ (see \cite{Kripke2007} for a formal definition in terms of Kripke structures of modal logic). Then we have the distribution axiom~\cite{Kripke2007,LogicStanford} which allows agents combine statement which contain inferences:
\begin{axiom}[Distribution axiom.]
\label{axiom:distribution}
If an agent $A_i$ knows a statement $\phi$ and that a statement $\psi$ follows from the statement $\phi$, then the agent can conclude that $\psi$ holds:
$$(K_i\phi\wedge K_i(\phi\Rightarrow\psi))\Rightarrow K_i(\phi\wedge \phi\Rightarrow\psi) \Rightarrow K_i\psi.$$
\end{axiom}
The distributivity of knowledge is an essential property which allows, given a set of facts and a set of inferences, to make conclusions about the world. Weakening, or excluding this rule would lead to agents restricted, or not being able to make their conclusions, and in turn, restrict our analysis of the theory they use. 

In logical puzzles, as well as in knowledge transfer scenarios, agents are not always combining their knowledge directly -- sometimes they reason from the viewpoint of each other, for example, ``agent $A_i$ knows that agent $A_j$ knows that ...''. Then to use the distribution axiom like above, we require an additional step of inheriting or trusting the knowledge of another agent as well as their own,
\begin{gather*}
    K_i (\phi \Rightarrow \psi) \wedge K_i K_j (\psi \Rightarrow \chi) \ \Longrightarrow \ K_i (\phi \Rightarrow \psi) \wedge K_i (\phi \Rightarrow \chi) \ \Longrightarrow \ K_i (\phi \Rightarrow \chi).
\end{gather*} 

\begin{figure}
    \centering
    \begin{subfigure}{0.45\textwidth}
    \centering
    \includegraphics[scale=0.3]{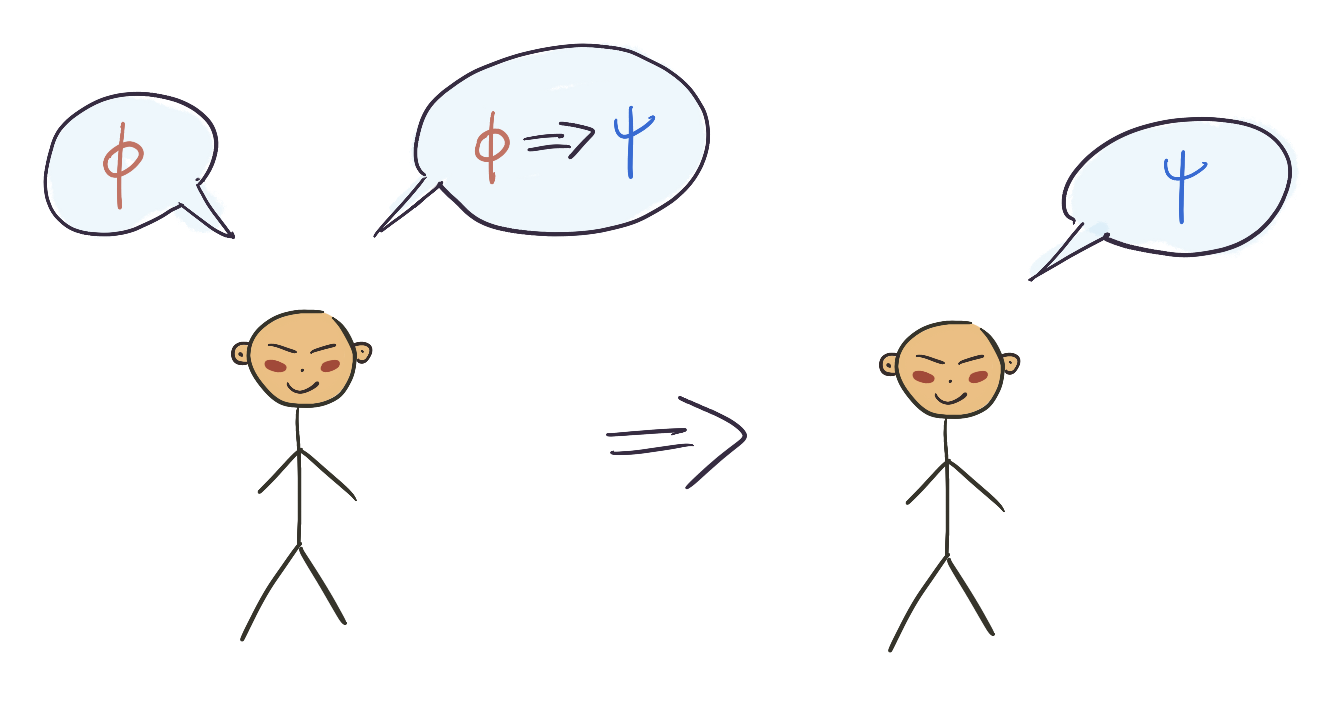}
    \caption{{ \bf Distribution axiom.} This basic logical axiom allows agents to use inferences in their reasoning; if $B$ knows $\phi$ and also knows that $\psi$ follows from $\phi$, he can conclude $\psi$, where $\psi$ and $\phi$ are some statements within the language of the theory, in the case of this paper, statements about outcomes.}
        \label{fig:distribution}
    \end{subfigure}
    \
    \begin{subfigure}{0.45\textwidth}
    \centering
    \includegraphics[scale=0.25]{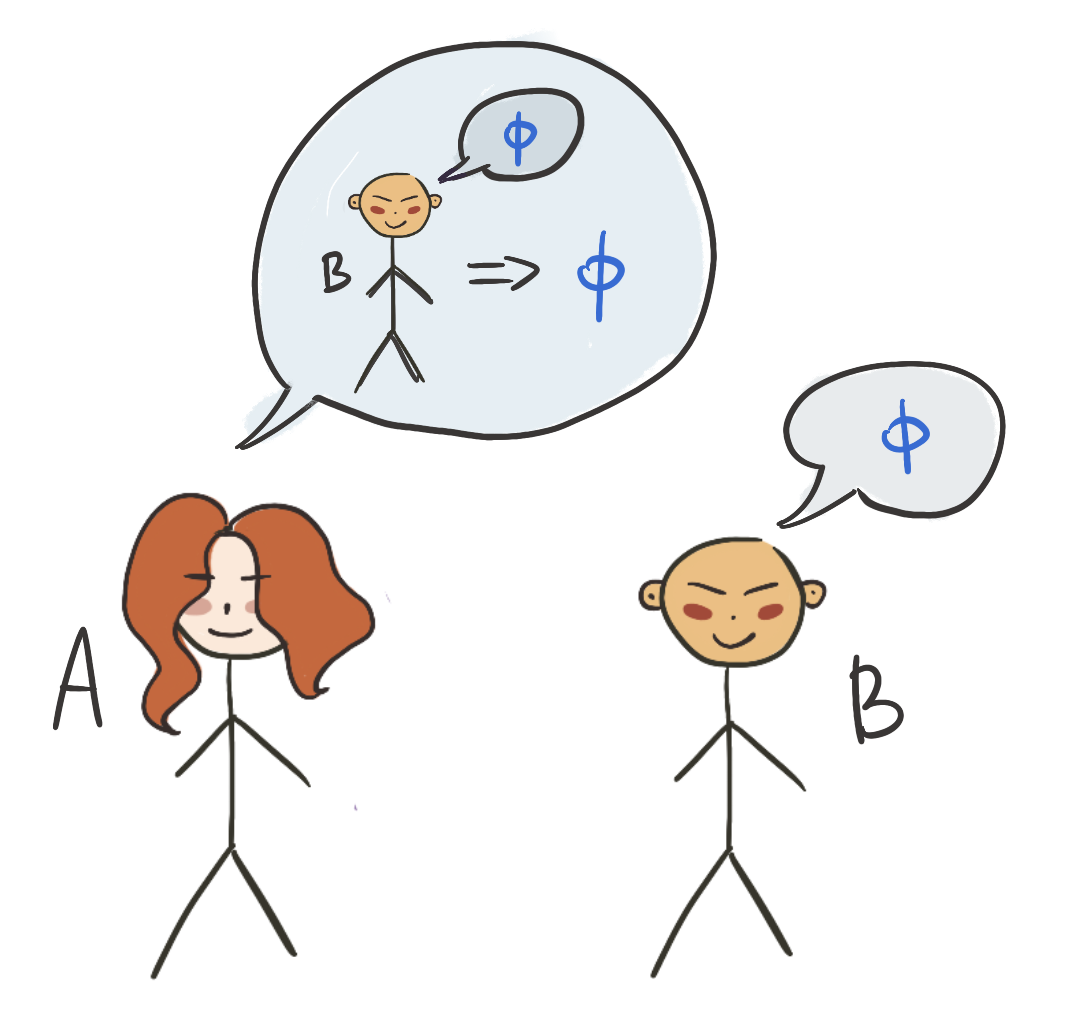}
    \caption{{ \bf Trust relation.} If agent $A$ trusts agent $B$, this means that she can treat $B$'s knowledge as her own.}
        \label{fig:trust}
    \end{subfigure}
    \caption{{ \bf Knowledge axioms.} Illustrating Axiom~\ref{axiom:distribution} and Definition~\ref{def:trust}.}
\end{figure}

\begin{definition}[Trust]
\label{def:trust}
    We say that an agent $A_i$ trusts the knowledge of an agent $A_j$ (and denote it by $A_j \leadsto A_i$) if and only if $$K_i K_j \phi \implies K_i \phi,$$ for all $\phi$. 
\end{definition}

The condition for trust needs to be instantiated in a given theory and scenario. In the next subsection, we will instantiate this axiom through the concept of compatibility between agents in a multi-agent setup and define a multi-agent paradox in a physical theory.

\subsection{Defining a multi-agent paradox in a theory}

We are ready to formally define our main object of interest in this paper -- a Wigner's Friend type multi-agent paradox. The assumptions below are a generalisation of the quantum theory specific assumptions of the Frauchiger-Renner paper \cite{Frauchiger2018,NL2018}, and a recent paper \cite{Vilasini2022} that proves the necessity of a ``setting-independence'' assumption for deriving apparent Frauchiger-Renner type paradoxes in quantum theory.

\begin{definition}[Wigner's Friend type multi-agent paradox in a physical theory]
\label{def: paradox}
If a theory $\mathbb{T}$ includes a multi-agent setup $\mathcal{M}\mathcal{A}$
where it is impossible to satisfy the following four assumptions simultaneously, then we say that the theory entails a Wigner's Friend type multi-agent paradox,
\begin{enumerate}
    \item (\textit{\textbf{common knowledge}}) all conclusions made by agents are based on the common theory $\mathbb{T}$ they use and the same multi-agent setup $\mathcal{MA}$. The theory and setup are taken to be common knowledge.
    
    \item (\textit{\textbf{reasoning about compatible agents}}) 
Two agents $A$ and $B$ can only apply the trust rule (\Cref{def:trust}) to statements $\phi_{\vec{s}}$ of a multi-agent set-up, i.e., $K_AK_B(\phi_{\vec{s}})\implies K_A\phi_{\vec{s}}$ when $A$ and $B$ are compatible. Additionally, agents can only reason using statements $\phi_{\vec{s}}$ (either atomic outcome statement or atomic inference) of a multi-agent set-up (\Cref{def:language_v2}) where the set of all outcomes appearing in $\phi_{\vec{s}}$ correspond to those of a set of compatible measurements.

    \item (\textit{\textbf{setting-independence}}) Setting labels (modeling a generalised notion of Heisenberg cuts, \cref{sec:meas_model}) can be ignored, every statement $\phi_{\vec{s}}\in \Sigma_{\mathcal{MA}}$ can be replaced with its setting-unlabelled version $\phi$.

    \item (\textit{\textbf{non-contradictory outcomes}}) No agent can arrive at contradictory conclusions about the value of any set of outcomes in the setup
i.e., there exists no agent $A$ such that  $$K_A(\vec{a}_j=\vec{v}_j \land \vec{a}_j=\neg \vec{v}_j)$$ for some value assignment $\vec{v}_j$ for a set $\vec{a}_j$ of outcomes in the scenario, where $\neg$ denotes negation.
   
\end{enumerate}
\end{definition}

When reasoning using the assumptions of the above definition, compatible agents can fully ignore the settings they label their statements with. To show this, consider two compatible agents Alice $A$ and Bob $B$, who have knowledge of statements $\phi_{\vec{s}_1}$ and $\chi_{\vec{s}_2}$ derived under two different setting vectors $\vec{s}_1$ and $\vec{s}_2$ respectively. Then if $A$ makes use of $B$'s knowledge, she obtains
\begin{align*}
    & K_A \phi_{\vec{s}_1} \wedge K_A K_B \chi_{\vec{s}_2} \\
    \xrightarrow{\text{trust}} &K_A \phi_{\vec{s}_1} \wedge K_A \chi_{\vec{s}_2} \\
    \xrightarrow{\text{distrib. axiom}} &K_A \left(\phi_{\vec{s}_1} \wedge \chi_{\vec{s}_2}\right) \\
    \xrightarrow{\text{setting-indep. }} &K_A (\phi \wedge \chi).
\end{align*}
Notice that this conclusion is identical to beginning with the setting independent statements $\phi$ and $\chi$ and applying the trust and distributive axioms.
Therefore, for simplicity, we can ignore the setting labels in the explanations and analysis concerning multi-agent paradoxes by imposing the setting-independence assumption first. Below, we provide an overview of the Wigner's Friend type multi-agent paradox arising in the Frauchiger-Renner set-up \cite{Frauchiger2018}, as an example.

\begin{figure}[h!]
    \centering
    \includegraphics[width=0.8\linewidth]{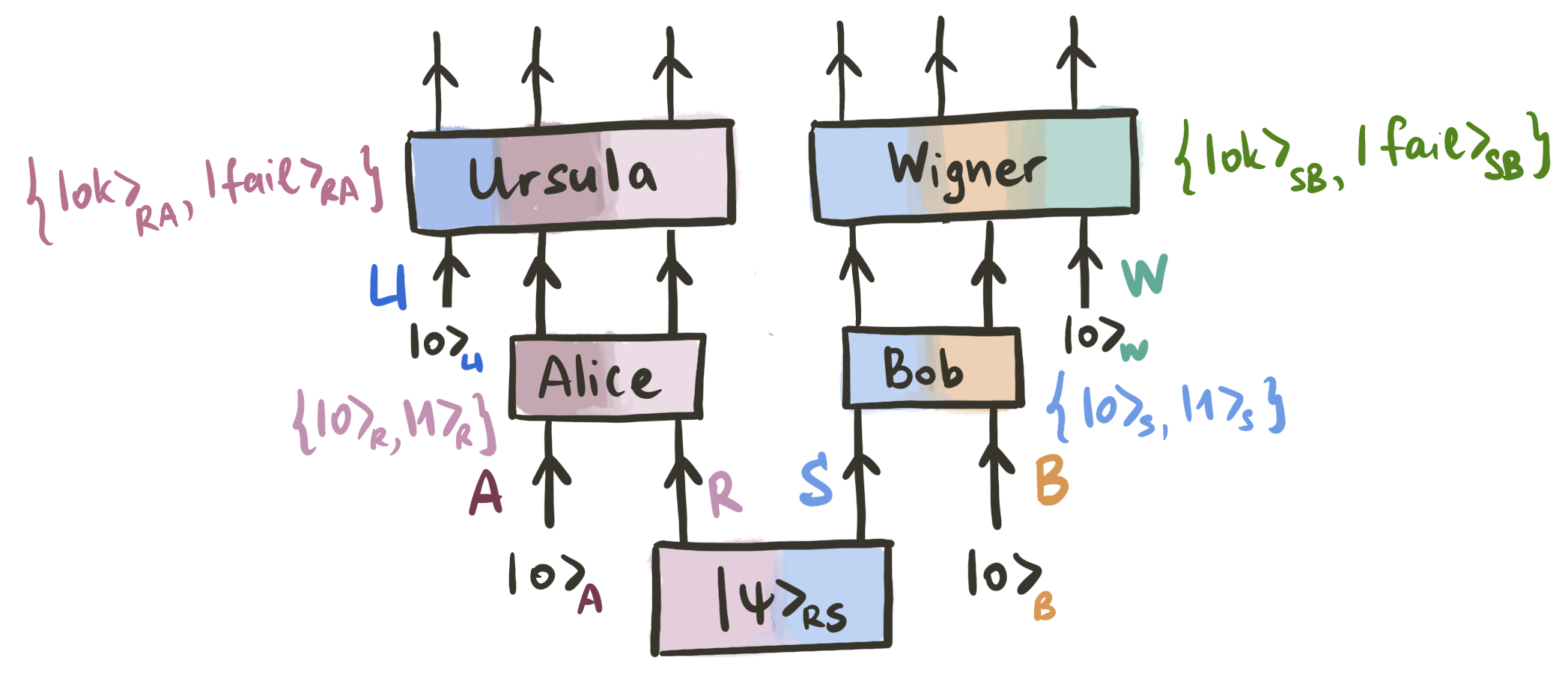}
    \caption{The Frauchiger-Renner setup \cite{Frauchiger2018}  $\mathcal{M}\mathcal{A}_{FR}$ as an example of a multi-agent paradox.}
    \label{fig:fr-entanglement}
\end{figure}
The multi-agent setup $\mathcal{M}\mathcal{A}_{FR}$ consists of: a set of agents $\mathcal A_{FR} = \{\text{Alice, Bob, Ursula, Wigner}\}$; a set of systems $\mathcal S_{FR} = \{R,S\}$; a set of memories, one for each agent $\mathcal L_{FR} = \{ A,B,U,W \}$ (each initialised in the state $\ket{0}$); an initial joint state of all systems $\mathcal S_{FR}$,  $\ket{\psi}_{RS} = (\ket{00}_{RS}+\ket{10}_{RS}+\ket{11}_{RS})/\sqrt{3}$, and
     a set of measurements conducted by the agents $\mathcal M_{FR} = \{\mathcal M_A,\mathcal M_B,\mathcal M_U,\mathcal M_W\}$, given by the bases
    \begin{align*}
\mathcal M_A: \{\ket{0}_R&,\ket{1}_R\},\\
\mathcal M_B : \{\ket{0}_S&,\ket{1}_S\},\\
\mathcal M_U: \{\ket{ok}_{RA}=(\ket{00}_{RA}-\ket{11}_{RA})/\sqrt{2}&,\ket{fail}_{RA}=(\ket{00}_{RA}+\ket{11}_{RA})/\sqrt{2}\},\\
\mathcal M_W: \{\ket{ok}_{SB}=(\ket{00}_{SB}-\ket{11}_{SB})/\sqrt{2}&,\ket{fail}_{SB}=(\ket{00}_{SB}+\ket{11}_{SB})/\sqrt{2}\}.
    \end{align*} 
    For each measurement $i\in\{A,B\}$, the setting $s_i=0$ description is a unitary corresponding to the CNOT gate on the respective system and memory, with the system as control. Then it is easy to see that, for example, when computing probabilities of Ursula’s measurement $\mathcal M_U$ acting on $R$ and $A$ while using the setting $s_A=0$ for Alice, this can be equivalently reduced to a simple $X$-basis measurement $\mathcal{M}'_U:=\{(\ket{0}_{R}-\ket{1}_{R})/\sqrt{2},(\ket{0}_{R}+\ket{1}_{R})/\sqrt{2}\}$ on the system $R$ alone. Thus the default prediction $P(u,w|\vec{s})$ (which uses $s_A=s_B=0$ by \cref{def: default_predictions}) is equivalent to $P(u,w|\mathcal{M}’_U,\mathcal{M}’_W,\ket{\psi}_{RS})$ where the primed measurements are $X$-basis measurements directly on $R$ and $S$ and are compatible (as required by \cref{def:compatibility_v2}). Similarly, the default predictions for the outcome pairs $(a,b)$, $(a,w)$ and $(u,b)$ can be written in terms of measurements (either $Z$ or $X$ basis) acting directly on the state $\ket{\psi}_{RS}$. Dropping the setting labels via the setting-independence assumption and the conditioning on the measurements (as these are clear from context), the probabilities of the scenario can be expressed as $P(u,w)$, $P(a,w)$, $P(a,b)$ and $P(u,b)$, and yield the following (setting-unlabelled) inferences according to Definition 10. These probabilities (as well as the state and the primed measurements) are in fact isomorphic to those of Hardy's paradox \cite{Hardy1993}.
    \begin{align*}
  \text{ Alice}:    P(w=fail|a=1)=1 \ &\Rightarrow \ K_A(a=1 \Rightarrow \ w=fail);\\
\text{Bob}: P(a=1|b=1)=1 \ &\Rightarrow \ K_B(b=1 \Rightarrow \ a=1);\\
 \text{Ursula}:  P(b=1|u=ok)=1 \ &\Rightarrow \ K_U(u=ok \Rightarrow \ b=1).
    \end{align*}
    Finally, we post-select on the outcomes $u=w=ok$ which occur with probability $1/12$ in this scenario, and in such a run $K_W(u=ok \wedge w=ok)$. The pairs of agents performing compatible measurements in the setup are: $(Alice,Bob)$, $(Bob,Ursula)$, $(Alice,Wigner)$ and $(Ursula,Wigner)$. These pairs of agents can then use the Distribution and Trust Axioms to combine the above inferences in the given experimental run. For instance, $K_W(u=ok \wedge w=ok)$, and $K_WK_U(u=ok \Rightarrow \ b=1)$ can be combined to give $K_W(u=ok \wedge w=ok \Rightarrow \ b=1)$, and by combining all the above statements in this manner, the agents reach a contradiction between the assumptions of Definition 13, i.e., a multi-agent paradox.

\section{Multi-agent paradoxes and contextuality in general theories}
\label{sec:paradox}

In this section, we present our main result linking multi-agent paradoxes and contextuality, which implies that that no logically non-contextual physical theory can lead to such paradoxes. Then, we illustrate an example within quantum theory, by constructing a multi-agent paradox based on a genuine contextuality scenario in quantum theory, which establishes that non-locality of a theory is not a necessary condition for the theory to admit multi-agent paradoxes in the sense of \Cref{def: paradox}. 

\subsection{Multi-agent paradoxes are proofs of logical contextuality}
\label{sec: paradox_implies_contextual}

The following theorem captures our main result, a proof of the theorem can be found in \Cref{appendix: proof_main}.

\begin{restatable}[Multi-agent paradoxes imply logical contextuality]{theorem}{ParadoxImpliesContextuality}
\label{theorem: main}

Any theory $\mathbb{T}$ which admits a multi-agent paradox (\Cref{def: paradox}) also admits a logically contextual empirical model (\Cref{def: logical_contextuality}) i.e., multi-agent paradoxes in a theory prove the logical contextuality of the theory.
\end{restatable}

\paragraph{Contrasting Wigner's original paradox with FR-type paradoxes.}
The {\bf reasoning about compatible agents} assumption in \Cref{def: paradox} is key to our main result and the connection between multi-agent paradoxes and contextuality. However, this restriction on reasoning is not fundamental; similar paradoxes can arise even when incompatible agents reason about each other. A simple example is Wigner's original thought experiment in quantum theory. Here, an agent (the Friend) measures a system $S$ in the state $\frac{1}{\sqrt{2}}(\ket{0}+\ket{1})_{S}$ in the computational basis, recording the outcome in memory $F$ (initialised to $\ket{0}_F$). Wigner then measures the Friend's lab, modeled by the joint system $SF$, in the Bell basis, using states $\ket{ok/fail}_{SF}:=\frac{1}{\sqrt{2}}(\ket{00}\mp \ket{11})_{SF}$.

The Friend's measurement, when described using the setting $s_F=0$, is a unitary implementing a CNOT gate, with the system as control and memory as target (just as in the FR case, \Cref{fig:fr-entanglement}). In this scenario one would obtain the probabilities, $P(w=ok|s_F=0)=0$, but $P(w=ok|s_F=1)>0$ (see \cite{Vilasini2022} for further details), which would allow the statements $(w=fail)_{s_F=0}$ and $(w=ok)_{s_F=1}$ to be made according to the language of the theory (\Cref{def:language_v2}). Imposing {\bf setting-independence} reduces these statements to $w=fail$ and $w=ok$, leading to a violation of the {\bf non-contradictory outcomes assumption}, thereby suggesting a paradox. However, since Wigner and the Friend are not compatible agents in this set-up\footnote{We note that the probabilities  $P(w=ok|s_F=0)=0$ and $P(w=ok|s_F=1)>0$ also have $s_W=1$ in both cases (although not written above, in the interest of brevity) as they refer to Wigner's classical outcome, and the latter probability is equivalent to $\sum_{f} P(w=ok,f|s_W=s_F=1)$. But this is a joint probability of the Wigner and the Friends' outcomes computed in the order in which they act, these agents perform incompatible measurements and using this probability in the reasoning would violate the {\bf reasoning about compatible agents} assumption.}, this does not constitute a paradox under \Cref{def: paradox}. Additionally, there are no non-trivial measurement contexts in this example, preventing any link to contextuality. This provides a clear contrast between Wigner’s original paradox and the one by FR, only the latter is possible when restricting agents to only reason about other compatible agents.

We leave open the question of characterizing the properties a theory must have to generate paradoxes where the compatibility assumption in \Cref{def: paradox} is relaxed. Results from \cite{Vilasini2022} suggest that in quantum theory, setting-dependence or dependence on Heisenberg cuts is necessary for all such paradoxes. Specifically, in Wigner's example, the predictions for $w$ show setting-dependence, as the probability of $w$ depends on the choice of setting $s_F$, as seen in $P(w=ok|s_F=0)=0$ and $P(w=ok|s_F=1)>0$. 

\paragraph{Link between contextuality and setting (or Heisenberg cut) dependence.}
As noted above, \cite{Vilasini2022} shows that for any logical paradox in quantum theory, even where incompatible agents reason about each other, setting-dependence is necessary. Any such paradox can only arise when these settings are ignored by imposing the setting-independence assumption of \Cref{def: paradox}, in a scenario where the predictions are in fact sensitive to the choice of settings. On the other hand, for paradoxes restricted to reasoning between compatible agents, we have shown that logical contextuality is necessary. This raises questions about the relationship between settings and contexts, or setting-(in)dependence and logical (non)-contextuality.

Since logical contextuality requires measurement compatibility, a meaningful comparison is possible only in paradoxes where the compatibility assumption holds. The default prediction $P(a_{j_1},\ldots,a_{j_p}|\vec{s})$ for every subset $\{a_{j_1},\ldots,a_{j_p}\}$ of outcomes assigns setting $s_i=1$ for each $i\in\{j_1,...,j_p\}$, and $s_i=0$ otherwise. When these outcomes result from compatible measurements, then the set $\{a_{j_1},\ldots,a_{j_p}\}$ specifies a measurement context, and this context information can be read-off from the $s_i=1$ entries of the setting vector. Thus, in multi-agent paradoxes where statements refer only to outcomes of compatible agents, each distinct setting vector uniquely specifies a measurement context. Dependence on settings therefore translates into dependence on contexts, leading to contextuality.

For example, suppose $C_1=\{a_1,a_2\}$ and $C_2=\{a_2,a_3\}$ are two sets of outcomes of compatible measurements, representing two contexts that overlap on $a_2$. The default predictions for the two subsets are $P(a_1,a_2|s_1=s_2=1,s_3=0)$ and $P(a_2,a_3|s_1=0,s_2=s_3=1)$. The marginals for $a_2$ would not match if the prediction for $a_2$ is setting-dependent, i.e., 
\begin{align*}
\begin{split}
&\sum_{a_1}  P(a_1,a_2|s_1=s_2=1,s_3=0)=    P(a_2|s_1=s_2=1,s_3=0)\\
\neq &P(a_2|s_1=0,s_2=s_3=1)=\sum_{a_3}P(a_2,a_3|s_1=0,s_2=s_3=1).
\end{split}
\end{align*}

 The setting choices in these cases reveal the distinct contexts $C_1$ and $C_2$ under which the predictions for $a_2$ were derived, and we would therefore have that $P(a_2|C_1)\neq P(a_2|C_2)$ in this scenario. Thus, setting-dependence translates directly into context-dependence. In this paper, we have focused on logical contextuality, where value assignments for outcomes are context-dependent (i.e., no consistent global, context independent assignment). More generally, contextuality entails that the probabilities of outcomes are context-dependent, as shown in this example.

Physically, this correspondence suggests that Wigner's Friend scenarios involving multi-agent paradoxes, as defined by \Cref{def: paradox}, translate different measurement contexts into different choices of Heisenberg cuts (which determine whether measurements are treated as producing classical records or only as evolutions of systems of the theory), and consequently context-dependence (or contextuality) to Heisenberg-cut-dependence of predictions.

\subsection{Wigner's Friend type paradox from a genuine contextuality scenario}
\label{sec:kcbs}

In this section, we describe the construction of a Wigner's Friend type multi-agent paradox based on a genuine contextuality scenario (distinct from Bell non-locality) in quantum theory, as opposed to FR's paradox which involves Bell non-locality, a special form of contextuality. Further details, explicit instantiation of the theory-independent assumptions of \Cref{def: paradox} and demonstration of a contradiction between these assumptions in this set-up are given in \Cref{appendix: kcbs}. The particular example, the states and measurements involved in the set-up we propose here are very similar to the example constructed in \cite{Walleghem2024-FR}, as noted in the introduction. However, there are slight differences. For instance, here we derive the paradox by instantiating the theory-independent assumptions and definition of paradox given in \Cref{def: paradox}, while \cite{Walleghem2024-FR} derived the paradox from quantum-theory specific assumptions, such as those referring to commutation of quantum measurements. The two are however closely related as our general assumption on reasoning about compatible agents is instantiated through commuting quantum measurements in this example. 
Another crucial aspect where our analysis of this example differs, is in explicitly accounting for the setting (or Heisenberg cut) independence assumption, noting that this is a necessary assumption for deriving apparent Wigner's Friend paradoxes in quantum theory \cite{Vilasini2022}. 

In his seminal paper~\cite{Hardy1993}, Hardy showed that the Bell non-locality of quantum theory can be established through a simple logical argument, without the need for inequalities. As has been previously pointed out, and we have discussed earlier, this proof is the basis of Hardy's paradox and is also for the original Frauchiger-Renner paradox which is constructed using a state and measurements that are isomorphic to those of Hardy's paradox (\Cref{fig:fr-entanglement}). Therefore, in order to construct FR-like multi-agent paradoxes using genuine contextuality scenarios (such as temporally sequential measurements on a single system as opposed to measurements on space-like separated subsystems), it is natural to consider a Hardy-like, logical proof of quantum contextuality in such settings. Such a proof has been provided in \cite{Cabello2013} using the KCBS contextuality scenario, that corresponds to a $n=5$-cycle i.e., a contextuality scenario with 5 measurements $X_1,...,X_5$ such that the only pairs of compatible measurements are $X_i$ and $X_{i+1 (\text{mod } 5)}$ for each $i\in \{1,2,3,4,5\}$. This is established using the following qutrit state $\ket{\psi}$ and binary valued measurements $X_1,...,X_5$ on it, where the projector corresponding to the ``1'' outcome of $X_i$ is given by the vector $\ket{v_i}$, defined as follows. The ``0'' outcome simply corresponds to its complement i.e., $\mathds{1}-\ket{v_i}\bra{v_i}$.
\begin{subequations}
\begin{equation}
\label{eq: KCBSstate}
    \ket{\psi}=\frac{1}{\sqrt{3}}(\ket{0}+\ket{1}+\ket{2}),
\end{equation} 
\begin{equation}
\label{eq: KCBSmmt1}
    \ket{v_1}=\frac{1}{\sqrt{3}}(\ket{0}-\ket{1}+\ket{2}),
\end{equation}
\begin{equation}
\label{eq: KCBSmmt2}
    \ket{v_2}=\frac{1}{\sqrt{2}}(\ket{0}+\ket{1}),
\end{equation}
\begin{equation}
\label{eq: KCBSmmt3}
    \ket{v_3}=\ket{2},
\end{equation}
\begin{equation}
\label{eq: KCBSmmt4}
    \ket{v_4}=\ket{0},
\end{equation}
\begin{equation}
\label{eq: KCBSmmt5}
    \ket{v_5}=\frac{1}{\sqrt{2}}(\ket{1}+\ket{2}).
\end{equation}
\end{subequations}
Note that $\ket{v_i}$ and $\ket{v_{i+1 (\text{mod }5) }}$ are orthogonal for each $i$ and hence the measurements defined by them are compatible (and jointly measurable) as required. Using the above state and measurements, one can readily obtain the following probabilities for the corresponding outcomes $a_i$. 
\begin{subequations}
\begin{equation}
\label{eq: KCBSpr1}
    P(a_2=1,a_1=1|\psi)=0,
\end{equation}
\begin{equation}
\label{eq: KCBSpr2}
    P(a_3=0,a_2=0|\psi)=0,
\end{equation}
\begin{equation}
\label{eq: KCBSpr3}
    P(a_4=1,a_3=1|\psi)=0,
\end{equation}
\begin{equation}
\label{eq: KCBSpr4}
    P(a_5=0,a_4=0|\psi)=0,
\end{equation}
\begin{equation}
\label{eq: KCBSpr5}
    P(a_5=0,a_1=1|\psi)>0.
\end{equation}
\end{subequations}
The proof of quantum contextuality proceeds by showing (through a logical argument) that in any non-contextual hidden variable model, the first four of the above equations would necessarily imply that $P(a_5=0,a_1=1|\psi)=0$ which contradicts the quantum observation embodied in the last equation.

Now, in an experimental test of such a contextuality scenario, each measurement context i.e., set of compatible measurements would be measured on one copy of the initial state among an ensemble of identical preparations. For the current example this can be achieved by preparing an ensemble $\ket{\psi}\otimes\ket{\psi}\otimes\ket{\psi}\otimes...$ and measuring one of the 5 measurement contexts (chosen freely) $\{(X_1,X_2),(X_2,X_3),(X_3,X_4),(X_4,X_5),(X_5,X_1)\}$ on each copy of the state to estimate the joint probability of the two outcomes in each context. For a FR-like multi-agent paradox however, we would require a way of implementing all 5 measurements on a single copy of $\psi$ by allowing agents to manipulate the memories of other agents using quantum operations to effectively ``undo'' their measurements.

Below, we present such a set-up consisting of 5 agents $\{A_i\}_{i=1}^5$, 5 memory systems, one per agent $\{M_i\}_{i=1}^5$ and a single additional system $S$ which they measure and store the resulting outcome in their respective memories. The 5 agents will each apply the 5 measurements of \cite{Cabello2013} for the KCBS contextuality scenario described above, on the state (also that of the KCBS scenario) of a single system $S$. Thus the resulting multi-agent setup will have default setting-conditioned predictions equivalent to  Equations~\eqref{eq: KCBSpr1}-\eqref{eq: KCBSpr5}, and allow agents to arrive the following paradoxical chain when reasoning using the assumptions of \Cref{def: paradox} after post-selecting on the outcomes $a_5=0,a_1=1$. 
\begin{equation}
\label{eq: 5chain}
    a_1=1\land a_5=0\Rightarrow a_4=1 \Rightarrow a_3=0\Rightarrow a_2=1\Rightarrow a_1=0.
\end{equation}
This chain corresponds to a Liar cycle of length 5. The schematic and circuit diagram for our KCBS-based multi-agent setup are given in Figures~\ref{fig:KCBS_scheme} and~\ref{fig:KCBS_circuit} respectively. In the following, we provide an overview of the proposed multi-agent set-up and the main steps of the associated protocol. Further details of the protocol along with an explicit demonstration of a multi-agent paradox (according to \Cref{def: paradox}) in this set-up are given in \Cref{appendix: kcbs}.
  
\begin{figure}[!ht]
     \centering
     \includegraphics[scale=0.75]{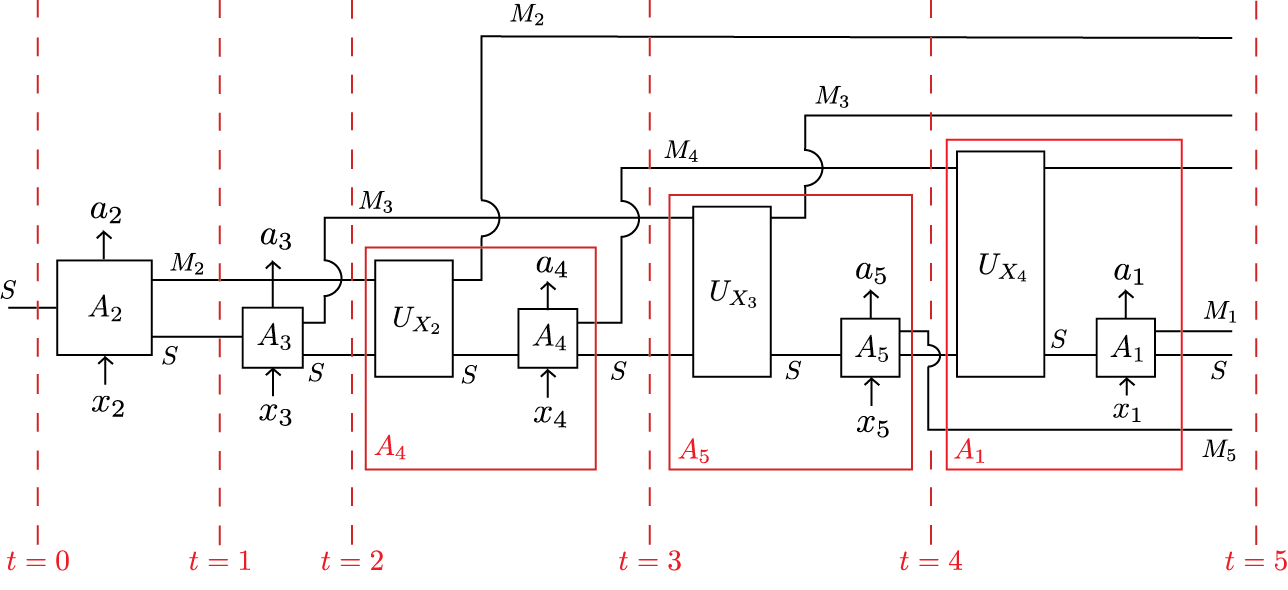}
     \caption{{\bf Circuit schematic for our protocol for an Wigner's Friend type multi-agent paradox using the KCBS contextuality scenario.}}
     \label{fig:KCBS_scheme}
\end{figure}
 
\begin{figure}
    \centering
    \begin{subfigure}{0.9\textwidth}
    \centering
  \includegraphics[scale=1.0]{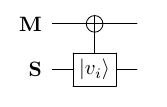}
    \caption{{ \bf  CNOT gate in the observable basis.} Denotes CNOT gate, controlled in the basis of $X_i$: $CNOT = \ket{v_i}\bra{v_i}_S\otimes(\sigma_x)_M + (\id_S - \ket{v_i}\bra{v_i}_S) \otimes \id_M$. The memory qubit $M$ is flipped if the control qubit is in the state $\ket{v_i}_S$ (defined in equations~\eqref{eq: KCBSmmt1}-\eqref{eq: KCBSmmt5}), and stays the same otherwise.}
        \label{fig:cnot-kcbs}
    \end{subfigure}
    \\
    \begin{subfigure}{0.9\textwidth}
    \centering
   \includegraphics[scale=1.0]{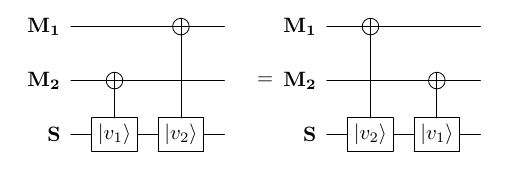}
    \caption{{ \bf  Circuit reduction rule.} Two memory updates commute iff they correspond to orthogonal measurements $\inprod{v_1}{v_2} =0$.}
        \label{fig:cnot-commute}
    \end{subfigure}
    \\
    \begin{subfigure}{0.9\textwidth}
    \includegraphics[scale=0.9]{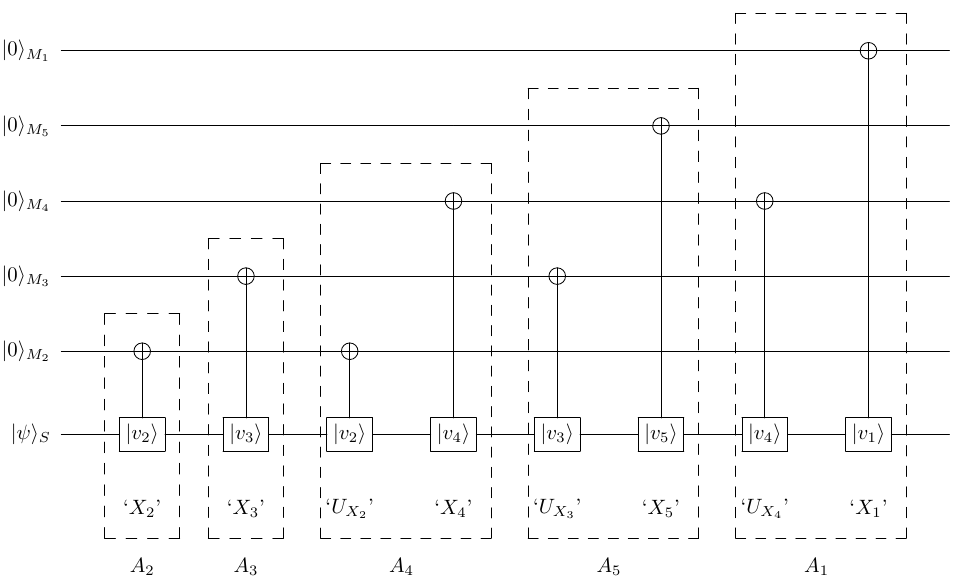}
    \caption{{ \bf Unitary circuit associated to our Wigner's Friend type paradox based on the KCBS contextuality scenario.} This circuit represents all the projective measurements of the set-up in their unitary description (i.e., setting 0 description), along with the additional unitaries. The operations $U_{X_2}$, $U_{X_3}$ and $U_{X_4}$ are designed to ``undo'' the measurements $X_2$, $X_3$ and $X_4$ respectively. }
        \label{fig:fr-kcbs}
    \end{subfigure}
    \caption{{\bf Constructing a  Wigner's Friend type multi-agent paradox paradox from KCBS contextuality scenario.}}
    \label{fig:KCBS_circuit}
\end{figure} 

\paragraph{Main information-processing steps of the setup: }

\begin{itemize}
    \item[$t=0$] Agent $A_2$ performs the measurement $X_2$ on the system $S$. 
    \item[$t=1$] Agent $A_3$ performs the measurement $X_3$ on the system $S$.
    \item[$t=2$] Agent $A_4$ performs a measurement on the joint system $SM_2$ that is operationally identical to the measurement $X_4$ on $S$. $A_4$'s measurement is decomposed as a unitary $U_{X_2}$ acting on $SM_2$ followed by a projective measurement on $S$.
    \item[$t=3$] Agent $A_5$ performs a measurement on the joint system $SM_3$ that is operationally identical to the measurement $X_5$ on $S$. $A_5$'s measurement is decomposed as a unitary $U_{X_3}$ acting on $SM_3$ followed by a projective measurement on $S$.
    \item[$t=4$] Agent $A_1$ performs a measurement on the joint system $SM_4$ that is operationally identical to the measurement $X_1$ on $S$. $A_1$'s measurement is decomposed as a unitary $U_{X_4}$ acting on $SM_4$ followed by a projective measurement on $S$.
    \item[$t=5$] Agents $A_5$ and $A_1$ (whose memories have not been tampered with) announce their outcomes. If the post-selection condition $a_1=1\land a_5=0$ is met, they halt the experiment and if not, the protocol is repeated until the post-selection succeeds (which it will with non-zero probability).
\end{itemize}

One can see that the above set-up now involves 5 measurements, all of which act on $S$ alone, along with 3 unitaries $U_{X_2}$, $U_{X_3}$ and $U_{X_4}$ acting on $SM_2$, $SM_3$ and $SM_4$ respectively. The 5 measurements on $S$ are precisely those of the KCBS scenario, given by equations~\eqref{eq: KCBSmmt1}-\eqref{eq: KCBSmmt5}, while the unitaries are specified further in \Cref{fig:KCBS_circuit}. We defer further details on the instantiation of settings and derivation of the multi-agent paradox in this set-up to \Cref{appendix: kcbs}.

\section{Structure and characterisations of multi-agent paradoxes}
\label{sec:character}
In this section, we present results on various properties of multi-agent paradoxes in physical theories, staring by deriving properties that hold in any theory, and then presenting properties which are applicable specifically when the theory in question is quantum (noting examples of violations in beyond-quantum theories).

\subsection{General theories}
\label{sec: structure_general}

\paragraph{Multi-agent paradoxes and liar cycles.} The following lemma on the structure of multi-agent paradoxes immediately follows from the proof of \Cref{theorem: main} (in particular from \Cref{eq: proof_liarcycle} established in this proof, along with the known fact that every liar cycle of statements corresponds to a cyclic reference graph, as discussed in \Cref{sec:background}). Note that for the purposes of the reference relation graph construction and simplicity, we label the nodes of the graph by atomic statements about outcomes (which can refer to a set of outcomes $\vec{a}_j$ produced by a set of agents, \Cref{def:language_v2}); two such vertices are related by a reference relation (a directed edge) if and only if there is an atomic inference connecting these two atomic statements and the agents producing them are connected by a trust relation. Alternatively, if we are additionally considering a scenario where each atomic statement involves the outcome of exactly one agent, we can also associate the vertices of the graph to agents directly.

\begin{restatable}[Structure of a multi-agent paradox]{lemma}{ma-paradox}
\label{lemma:multi-agent-paradox}
In every multi-agent paradox (\Cref{def: paradox}), the associated contradiction can be reduced to a half Liar's cycle type chain of statements between compatible sets of outcomes i.e., the following structure, and this implies a cyclic reference graph. 
\begin{equation}
\label{eq: structure_paradox}
     \vec{a}_j=\vec{v}_j  \Rightarrow \vec{a}_{l_1}=\vec{v}_{l_1}\Rightarrow \vec{a}_{l_2}=\vec{v}_{l_2}\Rightarrow\dots \Rightarrow   \vec{a}_{l_q}=\vec{v}_{l_q} \Rightarrow  \vec{a}_j=\neg \vec{v}_j,  
\end{equation}
where $ \vec{a}_j$, $ \vec{a}_{l_1}$,...,$\vec{a}_{l_q}$ are arbitrary disjoint subsets of outcomes,and $ \vec{a}_j\cup  \vec{a}_{l_1}$ and well as $ \vec{a}_{l_k}\cup  \vec{a}_{l_{k+1}}$ for $k\in \{1,...,p-1\}$ are all compatible sets of outcomes. 
\end{restatable}

Notice that \Cref{eq: structure_paradox} is a generalisation of a half-Liar's cycle chain (i.e., one of the chains of \Cref{eq: liar}), the sets of outcomes $\vec{a}_j,\vec{a}_{l_1},...,\vec{a}_{l_q}$ play the role of the statements $s_1,...,s_n$ and the value assignments $\vec{v}_k$ generalise the binary values. For this to be a full Liar's cycle, the following chain would also need to be derivable from the same multi-agent set-up, which as we will see in \Cref{thm:symmetricB}, is not always possible.

\begin{equation}
     \vec{a}_j=\neg \vec{v}_j  \Rightarrow \vec{a}_{l_1}=\neg \vec{v}_{l_1}\Rightarrow \vec{a}_{l_2}=\neg \vec{v}_{l_2}\Rightarrow\dots \Rightarrow   \vec{a}_{l_q}=\neg \vec{v}_{l_q} \Rightarrow  \vec{a}_j= \vec{v}_j. 
\end{equation}
\begin{figure}[t]
\centering
\includegraphics[scale=0.35]{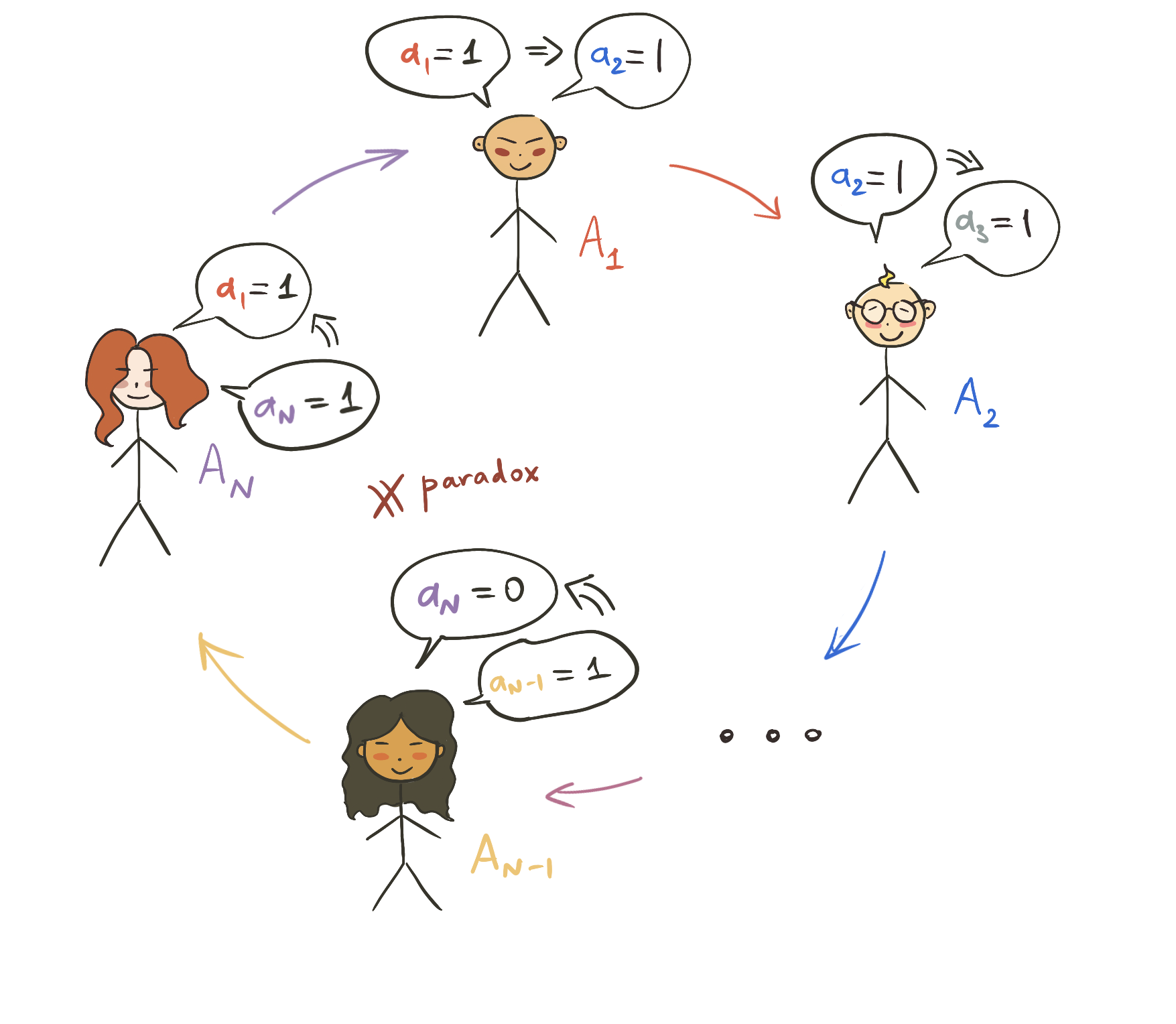}
\caption{{ \bf  The minimal structure of the physical multi-agent paradox.} Every Wigner's Friend type multi-agent paradox in a physical theory (\Cref{def: paradox}) can be reduced to a half Liar's cycle type chain of statements (\Cref{lemma:multi-agent-paradox}).}
\label{fig:epistemic-quantum}
\end{figure}
\paragraph{Absence of paradoxes based on Yablo type chains.}

As we have seen in \Cref{subsec:epistemic-general}, the Yablo’s paradox requires an infinite chain of statements. In this work, we have considered statements associated with agents’ measurement outcomes and only scenarios with a finite number of agents and outcomes. Thus it would seem that we cannot have Yablo type paradoxes by construction. However, one may ask whether we could have a finite paradox by building on a finite Yablo-type chain and combining it with other statements (which could introduce cycles in the reference graph, that the original Yablo graph of \cref{fig:yablo} does not have). The following theorem shows that this is not possible in any theory.

\begin{restatable}[No paradox building on finite Yablo chains]{theorem}{yabloB}
\label{thm: yabloB}
Consider a multi-agent set-up with a finite number $N$ of agents. Then no theory can contain a multi-agent paradox (\Cref{def: paradox}) in such a set-up where the set of statements involved in deriving the paradox includes a Yablo type chain (\Cref{eq: Yablo_statements}) for the $N$ agents. 
\end{restatable}

There is scope to generalise these considerations to certain infinite settings. One may then consider whether the original Yablo's paradox involving an acyclic but infinite reference graph (without adding any additional statements to induce cyclicity or direct self-reference) can be realised in an appropriate extension of multi-agent setups to infinite agents. Although we do not formalise and study this here, we discuss it in further detail and present a conjecture. 
One model of an infinite agent scenario would be through a protocol that extends in time (indefinitely) with a new agent $A_i$ acting at each new (discrete) time step $t_i$, and taking the limit of infinitely many time steps. This could be done by considering a sequence of multi-agent set-ups $\mathcal{MA}_1,…,\mathcal{MA}_k,….$ where the subscript denotes the number of agents in the set-up and each $\mathcal{MA}_i$ is obtained from the previous $\mathcal{MA}_{i-1}$ by simply adding an agent $A_i$ acting at a time $t_i>t_{i-1}$, and where $\mathcal{MA}_i$ otherwise includes the same description for all agents $A_1,…,A_{i-1}$ as $\mathcal{MA}_{i-1}$. This gives a sequence of outcomes $\{a_i\}_i$ that can be described as random variables each associated with a time step $t_i$, which would constitute a stochastic process. 

We conjecture that Yablo-type paradoxes involving an infinite acyclic reference graph would not occur in multi-agent setups involving such infinitely extended chain of agents, in any causally well-behaved physical theory. Let us further elaborate on the motivation for this conjecture. Suppose a Yablo-like infinite chain were possible, where the statements $s_i$ of the Yablo chain are instantiated through outcomes $a_i$ using the same ordering between the two i.e., for all $i\geq 1$
\begin{align}
\begin{split}
    &a_{i}=1 \implies a_{j}=0 \quad \forall j>i,\\
&a_{i}=0 \implies \exists j>i,\quad a_{j}=1.
\end{split}
\end{align}

Now if we impose minimal consistency conditions on the theory that ensures  that we can recover the predictions of $\mathcal{MA}_1,…,\mathcal{MA}_{i-1}$ from $\mathcal{MA}_i$ by ``ignoring’’ the latest agents’ action, then applying the arguments of \Cref{thm: yabloB} to every  $\mathcal{MA}_{i}$ for a finite $i$ we would conclude that the predictions of each such set-up is explained by a single joint distribution on all outcomes $a_1,…,a_i$ of that set-up where the marginals of the distribution recover the predictions of the previous set-ups in the sequence. For instance, this type of consistency condition is true in quantum theory and can be instantiated in general theories within frameworks such as process theories, operational probabilistic theories and generalised probabilistic theories through a causality condition \cite{coecke2016categoricalquantummechanicsi,Coecke_Kissinger_2017, Chiribella2010, dAriano2017, Henson_2014}, formalised by invoking a generalised notion of tracing out systems. Then applying the Kolmogorov extension theorem for classical stochastic processes \cite{Kolmogorov1933} and its generalisation to quantum and generalised probabilistic theories with well-defined causal order \cite{Milz_2020}, will allow to infer the existence of a well-defined probability measure on the infinite sequence of random variables (outcomes of measurements on classical or non-classical systems) that recovers all marginals on the finite sub-sequences. Finally, the existence of a well-defined global probability distribution over the outcomes of all the agents would imply the absence of any multi-agent paradox in this infinite case. Generalising our framework to model such infinite set-ups, extending our main theorem, \Cref{theorem: main} to this case and formally proving the above conjecture are left for future work.

\paragraph{No paradox with deterministic probabilities.}

Consider a chain of statements about $N$ disjoint sets of measurement outcomes $\{\vec{a}_i\}_{i=1}^N$, starting with a value assignment for $\vec{a}_1$ and ending with one for $\vec{a}_N$. Can a multi-agent paradox occur in a theory, involving such a chain and where the joint probability associated with the end points $\vec{a}_1$ and $\vec{a}_N$ is deterministic? So far none of the examples we know of in quantum theory or general theories have this property, and the following theorem establishes that this is generally impossible.

\begin{restatable}[No paradox with deterministic probabilities]{theorem}{postselect}
\label{theorem: postsel}
No theory admits a multi-agent paradox associated with the following chain of statements,
\begin{gather*}
\vec{a}_1=\vec{v}_1  \Rightarrow \vec{a}_2=\vec{v}_2 \Rightarrow \dots \Rightarrow   \vec{a}_N=\vec{v}_N,
\end{gather*}
where $\{\vec{a}_i\}_{i=1}^N$ are disjoint sets of outcomes with the sets of measurements associated with each $\vec{a}_i$ and $\vec{a}_{j}$ being compatible for $j=i+1 \text{ mod } N$, $\{\vec{v}_i\}_{i=1}^N$ are some value assignments for those outcomes, $N\geq 2$ and the end points of the chain have deterministic probability $P(\vec{a}_1=\vec{v}_1,\vec{a}_N=\neg \vec{v}_N)=1$.
\end{restatable}

\paragraph{Chain negation.}
Can we find multiple distinct chains leading to a paradox in a single multi-agent set-up? If we are given at least one, we can construct another one by simply considering its logical negation. This statement is justified by a following theorem, which is true in any theory not just the quantum case, as the proof follows purely from the rules of conditional probabilities.
\begin{restatable}[Negating an atomic inference]{theorem}{negation}
\label{claim:negating-inference}
Given an inference 
\begin{gather*}
\vec{a}=\vec{v}_a\Rightarrow \vec{b}=\vec{v}_b
\end{gather*}
where the measurements associated with the sets $\vec{a}$ and $\vec{b}$ of outcomes are compatible, the negation of the inference also holds:
\begin{gather*}
\neg (\vec{a}=\vec{v}_a \Rightarrow \vec{b}=\vec{v}_b) \ := \ (\vec{b}=\neg \vec{v}_b \Rightarrow \vec{a}=\neg \vec{v}_a)
\end{gather*}
\end{restatable}

The next corollary follows immediately from recursive application of the above theorem to an inference chain.
\begin{corollary}[Negating an inference chain]
    Given an inference chain where each implication relates outcomes of subsets of compatible measurements,
    \begin{gather*}
\dots \Rightarrow \vec{a}=\vec{v}_a\Rightarrow \vec{b}=\vec{v}_b \Rightarrow \vec{c}=\vec{v}_c\Rightarrow  \dots,
\end{gather*}
The negation of the inference chain also holds
\begin{gather*}
\dots \Rightarrow \vec{c}=\neg \vec{v}_c\Rightarrow \vec{b}=\neg \vec{v}_b \Rightarrow \vec{a}=\neg\vec{v}_a \Rightarrow \dots,
\end{gather*}
\end{corollary}

While the above theorem applies to general theories, it is useful to illustrate it with the well-known quantum example of FR where we had the following chain.
\begin{gather*}
    u = ok \Rightarrow b = 1 \Rightarrow a = 1 \Rightarrow w = fail.
\end{gather*}
We can negate this chain to obtain another chain, starting with the agent who was at the end point of the original chain:
\begin{gather*}
    w = ok \Rightarrow a = 0 \Rightarrow b = 0 \Rightarrow u = fail.
\end{gather*}
This also leads to a paradox if we post-select on (thought-)experimental runs with $u=w=ok$, as in the original argument of FR.

\paragraph{Link to extremal vertices in $n$-cycle scenarios.} An $n$-cycle scenario is a measurement scenario with a set $\{X_i\}_{i=1}^n$ of $n$ variables and where the contexts consist of each pair of adjacent variables i.e.,
\begin{equation}
  \mathcal{C}_{n-cyc}:=\{(X_1,X_2),(X_2,X_3),...,(X_{n-1},X_n),(X_n,X_1)\}=\{C_1,...,C_n\}.  
\end{equation}
An empirical model (\cref{def: emp_model}) would then specify for each $i$, a probability $e_{C_i}:=P(x_i,x_{i+1})$. Here we assume binary outcomes $x_i$ for each $X_i$ taking values in $\{0,1\}$. We can consider the following measure of correlations for each context $C_i$, that captures the expectation value
\begin{align}
\begin{split}
        <X_i,X_{i+1}>:= &P(x_i=x_{i+1})-P(x_i\neq x_{i+1})\\
        =&P(x_i=0,x_{i+1}=0)+P(x_i=1,x_{i+1}=1)\\-&P(x_i=0,x_{i+1}=1)-P(x_i=1,x_{i+1}=0)
\end{split}
\end{align}

Following \cite{Araujo_2013}, consider the function 
\begin{equation}
\label{eq: ncyc_omega}
  \Omega=\sum_{i=1}^n\gamma_i  <X_i,X_{i+1}>,\qquad   \gamma_i\in\{+1,-1\} \quad \forall i\in\{1,...,n\},
\end{equation}
where the number of $\gamma_i=-1$ is odd.
Different theories place different bounds on $\Omega$ (depending on the $\gamma_i$). A bound on $\Omega$ which holds in all classical theories (specifically, non-contextual hidden variable or NCHV theories, see \cite{Araujo_2013}) defines a contextuality inequality for the $n$-cycle scenario. Such contextuality inequalities for the $n$-cycle are derived in \cite{Araujo_2013}. For even $n$, the $n$-cycle can be understood as a bipartite Bell non-locality scenario with $n/2$ measurements per party and these would correspond to Bell inequalities in that case.  Notice that the maximum possible value of $\Omega$ for both the even and odd $n$ scenarios is $n$, which occurs when $\gamma_i=<X_i,X_{i+1}>$ for each $i$, these define the extremal vertices of the no-disturbance polytope (analogous to the non-signalling polytope in the case of Bell non-locality). For instance, the CHSH Bell scenario corresponds to an $n=4$ cycle, and the PR box is one such extremal vertex. Interestingly, the multi-agent paradox constructed using PR boxes in \cite{VNdR} is possible without any post-selection, in contrast to all known examples of quantum multi-agent paradoxes constructed based on an $n$-cycle scenario (e.g., Frauchiger-Renner's set-up which is a $4$-cycle as well as the KCBS or $5$-cycle based paradox we introduced in \Cref{sec:kcbs}, along with the similar one proposed in \cite{Walleghem2024-FR}).

We now define a sub-class of the multi-agent paradoxes allowed by \Cref{def: paradox}, which are possible without post-selection in a binary outcome $n$-cycle scenario and link them to such extremal vertices in general.

\begin{definition}[Post-selection free $n$-cycle paradox]
\label{def: paradox_no_ps}
Consider a multi-agent set-up with $n$ agents, each agent $A_i$ is associated with a measurement outcome $x_i$ taking binary values $\{0,1\}$, and only adjacent pairs of agents $(A_1,A_2)$, $(A_2,A_3)$,...,$(A_{n-1},A_n)$, $(A_n,A_1)$ are compatible (\Cref{def:compatibility_v2}). A post-selection free $n$-cycle paradox is any multi-agent paradox in such a set-up where agents reasoning using the assumptions of \Cref{def: paradox} obtain the following two chains for some $i\in \{1,...,n\}$.
\begin{align}
\label{eq: no_ps_paradox}
    \begin{split}
        &x_i=0\Rightarrow x_{i+1}=v_{i+1} \Rightarrow x_{i+2}=v_{i+2} ...\Rightarrow x_{i-1}=v_{i-1}\Rightarrow x_{i}=1\\
        &x_i=1\Rightarrow x_{i+1}=v'_{i+1} \Rightarrow x_{i+2}=v'_{i+2} ...\Rightarrow x_{i-1}=v'_{i-1}\Rightarrow x_{i}=0,
    \end{split}
\end{align}
    where for each $j\neq i\in \{1,...,n\}$, $v_j\neq v'_j$ and $v_j,v'_j\in\{0,1\}$ and the addition in the above equation is to be taken modulo $n$.
\end{definition}

\begin{restatable}[Post-selection free $n$-cycle paradox and extremal vertices]{theorem}{ncycleParadox}
\label{theorem: ncycle_paradox}
A physical theory $\mathbb{T}$ admits a post-selection free $n$-cycle paradox if and only if the theory contains an extremal vertex of the non-disturbance polytope of the associated $n$-cycle scenario.
\end{restatable}

A proof of this theorem can be found in \Cref{appendix:epistemic-quantum}. 

\subsection{Quantum theory}
\label{sec: structure_quantum}
In quantum theory, we will take the measurement of each agent in a multi-agent set-up (\Cref{def: MAsetup}) to be a projective measurement. This is without loss of generality as a set-up where agents measure their systems using general POVM measurements, can, as a consequence of the Naimark dilation, always be transformed into a set-up where they use rank-1 projective measurements by expanding the set of systems. Thus, in the quantum case we use the same definition of multi-agent set-ups as was proposed in \cite{Vilasini2022}, which uses projective measurements, whose generality is also further justified in that paper.
Moreover, the \textit{compatibility} of agents will be instantiated in terms of commutativity of their measurement projectors. We will continue to apply \Cref{def: paradox} of a multi-agent paradox, in which case the structure of the paradox is given by~\Cref{lemma:multi-agent-paradox}. The proofs of the results presented in this section can be found in Appendix~\ref{appendix:epistemic-quantum}.

\paragraph{Chain reduction for triples.} Given a chain of inferences, possibly witnessing a multi-agent paradox, can one reduce it to a more minimal chain (in certain cases) that still witnesses an inconsistency?  It is possible to identify a rule which governs cases when we can safely remove some statements from the inference chain.

\begin{restatable}[Chain reduction]{theorem}{reduction}
\label{theorem:reduction}
Given a chain of inferences of the form
\begin{gather*}
    \dots \Rightarrow \vec{c}=\vec{v}_c \Rightarrow \vec{b}=\vec{v}_b \Rightarrow \vec{a}=\vec{v}_a \ \Rightarrow \dots,
\end{gather*}
where the outcome sets $\vec{a}$, $\vec{b}$ and $\vec{c}$ are associated with projective quantum measurements, where the measurements 
 projectors pairwise commute $[\pi_\vec{i}^{\vec{v}_i},\pi_\vec{j}^{\vec{v}_j}]=0$ for all value assignments $\vec{v}_i$, $\vec{v}_j$ with $i,j=\{a,b,c\}$, then the inference chain can be reduced to
\begin{gather*}
\dots \Rightarrow \vec{c}=\vec{v}_c \Rightarrow \vec{a}=\vec{v}_a \ \Rightarrow \dots
\end{gather*}
\end{restatable}

We note that the above chain reduction property is also a consequence of the fact that quantum theory admits no contextual correlations in the $n=3$ cycle and needs a minimum of $n=4$ to exhibit contextuality (here the three sets of measurements associated with $\vec{a}$, $\vec{b}$ and $\vec{c}$ form a $3$-cycle scenario). However, there do exist post-quantum GPTs that admit contextual correlations in $n=3$-cycle scenarios and would thus violate this property.

Now let us consider an arbitrary chain of statements which leads to a multi-agent paradox (cf. \Cref{def: paradox}) in quantum theory with projective measurements. Following Theorem~\ref{theorem:reduction}, we can take a look at all consecutive triples of statements, and exclude the middle statement for every case where all three statements in the triple pairwise commute. After this reduction procedure is carried out, we are left with the chain where for each triple the first and the last statements correspond to measurement projectors which don't commute. In the following for simplicity (and without any loss of generality) we represent each value assignment $\vec{v}_i$ by $1$.
\begin{gather*}
    \vec{u}=1 \Rightarrow \vec{a}_1=1 \Rightarrow \dots \Rightarrow \vec{a}_N=1 \Rightarrow \vec{w}=1 \\ \text{with} \ [\pi_{\vec{a}_i}^1,\pi_{\vec{a}_{i+2}}^1]\neq 0 \ \forall i\in\{1,\dots,N-2\}, \ [\pi_\vec{u}^1,\pi_{\vec{a}_2}^1]\neq 0 \ \text{and} \ [\pi_{\vec{a}_{N-1}}^1,\pi_\vec{w}^1]\neq 0.
\end{gather*}
Note that in the above, since every inference can only connect compatible agents, the projectors for every adjacent set of measurement outcomes in the above chain commute. The FR paradox scenario constitutes a maximally reduced chain of four agents.

\paragraph{Symmetric chains: always consistent.}
So far, we have considered inference chains with one way implications, here we consider chains related by equivalences and obtain the following reduction rule in quantum theory. This is turn implies the impossibility of quantum multi-agent paradoxes involving only symmetric inferences,

\begin{restatable}[Symmetric chain reduction]{theorem}{symmetricB}
\label{thm:symmetricB}
Consider an inference chain
\begin{gather*}
   \vec{a}_1= \vec{v}_1\Leftrightarrow \vec{a}_2= \vec{v}_2 \Leftrightarrow \dots \Leftrightarrow \vec{a}_N= \vec{v}_N,
\end{gather*}
where the measurement projectors for each adjacent pair of outcomes sets commute i.e., $[\pi_{\vec{a}_i}^{\vec{v}_i},\pi_{\vec{a}_j}^{\vec{v}_j}]=0$ whenever $j=i+1\text{ mod }N$ for $i\in\{1,...,N\}$ and for all value assignments $\vec{v}_i$, $\vec{v}_j$. The remaining sets of measurements need not necessarily commute. Then the chain can be reduced to the symmetric inference between its end points
\begin{gather*}
   \vec{a}_1= \vec{v}_1\Leftrightarrow \vec{a}_N= \vec{v}_N.
\end{gather*}
\end{restatable}

Note that the above reduction does not hold in general if the implications go only in one direction. In \Cref{theorem:reduction} the one-way chain was reduced by assuming an additional commutation relation between the end-points of each triple, e.g., that the projectors of $\vec{a}_1$ and $\vec{a}_3$ also commute, which is not required here.

This means that such symmetric chains are always globally consistent in quantum theory, meaning that whenever local consistency is satisfied in individual equivalences, global consistency in chain as a whole  will be satisfied as well. This implies the following corollary which shows that quantum theory does not admit any multi-agent paradoxes arising from symmetric inference chains.
\begin{corollary}[No symmetric quantum paradoxes]
    \label{corollary:no_sym_paradox}

A multi-agent paradox based on the following inference chain along with a non-zero probability $P(\vec{a}_1= \vec{v}_1,\vec{a}_N= \neg \vec{v}_N)>0$ for the chain's end-points
\begin{gather*}
   \vec{a}_1= \vec{v}_1\Leftrightarrow \vec{a}_2= \vec{v}_2 \Leftrightarrow \dots \Leftrightarrow \vec{a}_N= \vec{v}_N,
\end{gather*}
cannot arise in quantum theory. Here the measurement projectors for each adjacent pair of outcomes sets commute i.e., $[\pi_{\vec{a}_i}^{\vec{v}_i},\pi_{\vec{a}_j}^{\vec{v}_j}]=0$ whenever $j=i+1\text{ mod }N$, for $i\in\{1,...,N\}$ and for all value assignments $\vec{v}_i$, $\vec{v}_j$
\end{corollary}

Again, however, this property is specific to quantum theory; it does not hold in PR-box world: there the analogue of the FR scenario leading to a paradox can in fact be formulated using equivalences instead of inferences~\cite{VNdR}.

\paragraph{No post-selection free $n$-cycle quantum paradoxes.} In \Cref{theorem: ncycle_paradox} we showed that any theory admits a particular class of multi-agent paradoxes, which we call post-selection free $n$-cycle paradoxes (\Cref{def: paradox_no_ps}) if and only if the theory contains an extremal vertex of the associated $n$-cycle non-disturbance polytope. Such vertices achieve $\Omega=n$ but the Tsirelson-type bounds for $n$-cycle scenarios derived in
\cite{Araujo_2013} show that the maximum $\Omega$ achievable in quantum theory is strictly less than $n$. This immediately yields the following corollary of \Cref{theorem: ncycle_paradox}.

\begin{corollary}[No post-selection free quantum paradoxes]
\label{corollary: postselection_quantum}
    Quantum theory does not admit any post-selection free $n$-cycle paradoxes.
\end{corollary}

Another easy way to understand this result is to take a closer look at the two chains in~\Cref{eq: no_ps_paradox}: negating the second statement chain according to~\Cref{claim:negating-inference}, and combining it with the first one, we obtain an equivalence chain
\begin{align}
    x_i=0 \Leftrightarrow x_{i+1}=v_{i+1} \Leftrightarrow \dots \Leftrightarrow x_{i-1}=v_{i-1} \Leftrightarrow x_i=1,
\end{align}
which cannot exist according to~\Cref{thm:symmetricB}.

This explains the necessity of post-selection in the FR paradox as well as the KCBS paradox we introduced in \Cref{sec:kcbs}, which are both constructed in quantum theory and are linked to $n=4$ and $n=5$ cycles respectively. Whether an analogous result ruling out quantum paradoxes without post-selection holds for non-binary outcomes in $n$-cycle scenarios and in more general scenarios beyond the $n$-cycle are left as open questions for future work.

\section{Conclusions and outlook}
\label{sec:conclusions}

In this work, we explored a specific class of multi-agent paradoxes that arise when agents reason within a physical theory using a defined set of assumptions, as outlined in \Cref{def: paradox}. These assumptions generalize the quantum-theory-specific assumptions underlying the Frauchiger-Renner paradox \cite{Frauchiger2018} in a theory-independent manner, while also incorporating implicit assumptions related to Heisenberg cuts, shown to be necessary for FR-type quantum paradoxes \cite{Vilasini2022}. In \Cref{theorem: main}, we demonstrated that the existence of such a paradox implies the logical contextuality of the theory, a distinctly non-classical feature. Additionally, we characterized these paradoxes within quantum theory, highlighting their differences from analogous paradoxes in post-quantum theories identified in previous work \cite{VNdR}.

This study represents a step toward a unified understanding of existing no-go theorems in quantum foundations \cite{Frauchiger2018,Brukner2017, Brukner2018,Bong2020, Haddara2022, Ormrod2022, Ormrod2023}, shifting the focus from contradictions to a clear articulation of the assumptions underlying their derivations. By formalizing these assumptions in a general manner, our work lays the groundwork for systematically classifying different classes of multi-agent paradoxes in Wigner's Friend-like scenarios. These paradoxes can be generated by varying the assumptions that lead to contradictions, thereby providing a method to certify other non-classical resources within the theory. We discuss several concrete possibilities for such investigations in the following sections.

\paragraph{Modifying compatibility of agents/measurements.} One way to generalize \Cref{def: paradox} is to drop the compatibility assumption, allowing incompatible agents to reason about each other. As discussed in \Cref{sec: paradox_implies_contextual}, this broader definition includes Wigner's original quantum thought experiment which involves two agents. In quantum theory, as shown in \cite{Vilasini2022}, Heisenberg-cut dependence (formalized as setting-dependence) is necessary for such paradoxes, which is closely linked to measurement incompatibility. However, it remains an open question which property of the theory, its state, or its measurement structure is necessary for setting-dependence in general. This also depends on specific conditions related to the evolution dictated by the two settings; for instance, in quantum theory, $s_i=0$ corresponds to modeling the $i^{\text{th}}$ measurement unitarily, and in the paradox found in \cite{VNdR} in box-world, a post-quantum theory, this evolution corresponds to an information-preserving memory update of box-world.

The restriction to reasoning about compatible agents in our \Cref{def: paradox} corresponds to defining contexts in terms of jointly measurable or compatible sets of measurements (\Cref{def:meas_compatibility}) in the definition of a measurement scenario and logical contextuality (\Cref{subsec:contextuality}). For compatible sets, the measurement order is irrelevant due to joint measurability. To study multi-agent paradoxes without this restriction, it would be interesting to consider an alternative definition of a measurement scenario by taking the order of measurements into account, thus taking the set of variables $X$ in a measurement scenario to be ordered, and allowing subsets of potentially incompatible variables in $\mathcal{C}$\footnote{The ordering information can be used to ensure that there is no ambiguity in considering incompatible sets of measurements, where the order in which the measurements are performed generally matters.}. Given such an ordered scenario, one could explore \textit{ordered empirical models}, as well as ordered compatibility notions for agents and measurements. Whether this approach can lead to a meaningful notion of non-classicality, analogous to contextuality in the unordered case, is an intriguing question for future work, which could also benefit from comparisons with other extensions of the sheaf-theoretic approach to contextuality that account for causal order \cite{Gogioso2023topology, Gogioso2023geometry}.

\paragraph{Modifying language of the theory and its logic.}
In this work, we focused on reasoning with statements of a specific form (\Cref{def:language_v2}) and combination rules such as the trust and distributive rules, which generalize the reasoning in the quantum FR paradox. For the class of paradoxes defined through this reasoning conditions, we showed that they always admit a half Liar cycle form (\Cref{lemma:multi-agent-paradox}) and indicate logical contextuality of the underlying theory (\Cref{theorem: main}). 
However it is not necessary that all logically contextual scenarios lead to paradoxes of this specific type. The multi-agent paradoxes based on Hardy's and KCBS scenarios of quantum theory (i.e., the FR paradox and our example in \Cref{sec:kcbs}, and the similar one in \cite{Walleghem2024-FR}) and the PR box scenario of box-world (the paradox of \cite{VNdR}) are examples of such Liar's cycle type paradoxes involving binary outcomes, which are linked to the $n$-cycle contextuality scenario. In such scenarios, quantum paradoxes necessarily involve post-selection (\Cref{corollary: postselection_quantum}). On the other hand, the GHZ scenario in quantum theory exhibits strong logical contextuality that can be certified without post-selection, it is not based on an $n$-cycle and does not appear to involve a half Liar cycle form as in \Cref{lemma:multi-agent-paradox}. However, the GHZ scenario has been used to construct a Wigner's Friend like paradox \cite{Walleghem2024_GHZ} under a different set of quantum theory-dependent assumptions. One future direction building on our framework would be to adapt the types of statements and combination rules allowed in the reasoning to explore  different classes of Wigner's Friend type paradoxes in a theory-independent manner, and subsequently use this to certify and characterise other forms of contextuality.

\paragraph{Shifting the setting variable: language, label, or agent's property.} 
In this paper we have considered setting as a label for the statement produced by an agent about theirs or other agent's outcome. There are however few other ways to include it in the formal description, which mathematically yield equivalent results (in the context of this paper), but can be ascribed different interpretations. Formally, including setting as a label for a proposition means that the truth function proving a truth value of the statements in the language, is (at least) a two-variable function, taking as arguments the statement itself and the setting it was produced in. If such labelled statements are combined in a setting-dependent scenario, the labels have to compared separately.
Alternatively, one can consider incorporating settings in the language of the theory itself by logically adding it to the statement (as done in \cite{Vilasini2022})
\begin{align*}
    \phi_s \ \rightarrow \ \phi\wedge s.
\end{align*}
In this case, when statements are combined the consistency of settings is ensured simply by the rules of propositional logic. 
Finally, one can consider promoting the setting from labelling the proposition to labelling the knowledge operator for the agent producing the statement,
\begin{align*}
    K_A \phi_s \ \rightarrow \ K_A^s \phi.
\end{align*}
This way of expressing the setting variable can highlight the scenarios where it is attributed to the agents producing the statements, and not the statements themselves. In all three cases, formalising setting-independence leads to amending different elements of the logical structure: language or distribution rules for proposition or knowledge operators.

\paragraph{Absoluteness of events and non-logical forms of contextuality.} 
We focused on Wigner’s Friend arguments that highlight inconsistencies in agents’ reasoning under certain assumptions, which includes the Frauchiger-Renner's argument. Another class of no-go arguments, including Brukner’s work \cite{Brukner2017, Brukner2018} and the Local Friendliness (LF) theorem \cite{Bong2020}, demonstrates the failure of Absoluteness of Observed Events (AoE) under additional assumptions. AoE posits a single well-defined joint probability distribution over all outcomes in a multi-agent scenario, recovering correct marginal distributions for compatible measurements.

While we linked logical non-contextuality to the absence of Wigner’s Friend type paradoxes (generalizing FR arguments), our methods suggest broader connections between non-contextual theories and those respecting AoE. In non-contextual theories, empirical models admit globally consistent joint distributions matching observed marginals. In the proof of \Cref{theorem: main}, we mapped multi-agent Wigner’s Friend setups to empirical models with setting-conditioned predictions corresponding to the model's probabilities. This suggests that non-contextuality in the model would lift to objective setting-independent joint distributions for all agents in the Wigner's Friend scenario, thereby implying an absolute notion of observed measurement events for all agents. We leave a full formalization and proof of this observation for future work.

In this direction, recent work \cite{Ormrod2023} studies the measurement problem from the perspective of the AoE assumption, linking this to Bell non-locality as well as information preserving evolutions \cite{VNdR} of agents' labs in a theory. Building on our work and \cite{Ormrod2023}, it would be interesting to address whether contextuality, rather than Bell non-locality is the more general necessary property for having a measurement problem in a theory. Furthermore, another recent and relevant work \cite{Walleghem2024} links LF-type arguments to non-contextuality for scenarios with parties acting on distinct subsystems. Applying our approach to study AoE and its links to contextuality as suggested above, could enable such considerations to be extended to multi-agent setups with sequential measurements on single systems, potentially under different sets of assumptions. Our framework also provides a basis for generalizing connections between paradox-based and AoE-based Wigner’s Friend arguments, as explored in \cite{Walleghem2024} (Theorem 4) for a class of set-ups in quantum theory, to arbitrary theories and multi-agent setups.

\paragraph{Limit of infinite time and infinite measurements.} As discussed in \Cref{sec: structure_general}, there is scope to generalise the present framework and results to consider protocols extended indefinitely in time, which would involve an infinite sequence of measurements. In the literature on classical semantic paradoxes (see also \Cref{sec:background}), while all paradoxes based on finitely many statements necessarily have a cyclic reference graph structure ~\cite{Jongeling2002, Rabern2013, Beringer2017} (analogous to our \Cref{lemma:multi-agent-paradox} for multi-agent paradoxes), new types of semantic paradoxes based on acyclic reference graphs, such as Yablo's paradox \cite{Yablo1993} become possible in the infinite-statement case. In the paragraph following \Cref{thm: yabloB}, we have presented a conjecture that we do not expect multi-agent paradoxes with a Yablo-type structure to be possible in a class of infinite-measurement extensions of multi-agent set-ups to physical theories with a well-defined causal order. Proving this conjecture will potentially involve an interface with studies on quantum and non-classical generalisations of stochastic processes and generalisations of the Kolmogorov extension theorem \cite{Milz_2020}.

\paragraph{Understanding the structure of quantum theory through multi-agent paradoxes.} A central question in the foundations of physics is understanding why nature follows quantum theory, often studied by identifying physical principles that might single out quantum behavior among other probabilistic theories. Our work highlights similarities between quantum theory and other theories like boxworld, which admit multi-agent paradoxes showing that logical contextuality is a necessary in both cases. Going further, we have identified key properties that hold for quantum multi-agent paradoxes but can fail in other theories such as boxworld, such as the necessity of post-selection for free $n$-cycle paradoxes and impossibility of paradoxes through symmetric reasoning chains. We believe our work, along with the extensions discussed above, could inform a research program to address the foundational question of "why quantum theory" through the lens of multi-agent paradoxes in Wigner's Friend type scenarios.

\paragraph{Multi-agent paradoxes with causal assumptions and general causal scenarios.} 
Our focus has been on multi-agent scenarios where operations occur in an acyclic order consistent with the time order, and agents communicate only from earlier to later moments in time. By studying the structure and connectivity of the agents' communication channels, one can specify an operational causal structure for such a scenario. By also accounting for assumptions on the causal structure of the multi-agent setup, such as the independence of sources shared between different sets of parties, one can consider whether Wigner's Friend type paradoxes based on other non-classical resources such as network non-locality (see \cite{Tavakoli_2022} for a review) can be constructed, and conversely used to certify said resources.

Alternatively, one might challenge the very assumption of an acyclic operational causal structure. Quantum theory allows for scenarios where agents' operations do not occur in a fixed acyclic order \cite{Hardy2005,Oreshkov2012, Chiribella2013}, as seen in the \emph{quantum switch} \cite{Chiribella2013}, where Alice acts before or after Bob depending coherently on a control qubit. More broadly, both classical and quantum theories permit consistent causal loops or cyclic causal structures that avoid paradoxes like the grandfather paradox \cite{Deutsch1991,Lloyd_2011,Baumeler_2016, Bongers_2021}. It would be intriguing to extend our formalism for multi-agent reasoning and paradoxes to these more general causal scenarios without an acyclic structure, both within and beyond quantum theory. This could reveal novel types of Wigner's Friend type arguments by combining logical and causal features of non-classical theories, as well as explore the interplay between Wigner's Friend type and causal paradoxes.

From a physical standpoint, it is also essential to consider the role of relativistic causality in spacetime, which is crucial for integrating quantum theory with general relativity. The relationship between relativistic causality (from a spatio-temporal context) and cyclic or indefinite causality (formulated in an abstract information-theoretic context) has been studied in \cite{Portmann2017,VilasiniRennerPRA, VilasiniRennerPRL}. It was shown that physical realizations of any indefinite causal order process in a background spacetime, and in particular performed table-top experiments of the quantum switch (e.g., \cite{Procopio2015, Rubino2017}), will necessarily admit a fine-grained explanation in terms of a definite acyclic causal order process and involve the non-localization of systems in spacetime.\footnote{We note that non-localization of systems in spacetime as per \cite{VilasiniRennerPRA, VilasiniRennerPRL} is distinct from the notion of time-delocalized subsystems introduced in \cite{Oreshkov2019} to argue for the implementation of indefinite causal order in quantum switch experiments. In particular, the former is consistent with the existence of a fine-grained definite causal order explanation for such experiments and also features in physical protocols in spacetime \cite{Portmann2017} which do not involve any indefinite causal order (as also clarified in \cite{Vilasini2025}).} Extending such studies into Wigner's Friend scenarios in general theories, by modeling agents as physical systems of the theory whose actions must comply with relativistic principles in spacetime (even at a fine-grained level), remains a relevant avenue for future research. 


\paragraph{Acknowledgements.}
We thank L\'{i}dia del Rio for interesting discussions.
VV's research has been supported by an ETH Postdoctoral Fellowship. VV acknowledges support from ETH Zurich Quantum Center, the Swiss National Science Foundation via project No.\ 200021\_188541 and the QuantERA programme via project No.\ 20QT21\_187724.
NN acknowledges support from the Swiss
National Science Foundation through 
SNSF project No.\ $200020\_165843$ and through the 
the National Centre of
Competence in Research \emph{Quantum Science and Technology}
(QSIT). 


\bibliographystyle{unsrtnat}


\addcontentsline{toc}{section}{\sc{Appendix}}

\section{A multi-agent paradox in a genuine contextuality scenario: further details}
\label{appendix: kcbs}

In this appendix, we give further details of the Wigner's Friend scenario captured by the multi-agent set-up of \cref{sec:kcbs} based on the $5$-cycle KCBS contextuality scenario. Here, we explicitly show how the theory-independent assumptions of \cref{def: paradox} can be instantiated in quantum theory to demonstrate a multi-agent paradox in that set-up. The set-up consists of 5 agents $\{A_i\}_{i=1}^5$ and a single system $S$, along with a memory $\{M_i\}_{i=1}^5$ for each agent. The schematic depicting the protocol of the set-up and the unitary circuit description of the set-up (corresponding to modelling each measurement as a pure unitary, using setting $s_i=0$) are illustrated in \cref{fig:KCBS_scheme} and \cref{fig:KCBS_circuit} respectively. We first describe how the settings (capturing the idea of Heisenberg cuts, \cref{sec:meas_model}) are instantiated in this case, and analyse the details of the setup before using that analysis to demonstrate the reasoning steps leading to the paradox according to \cref{def: paradox}.

\paragraph{Instantiating the settings.} In our general framework, the setting labels $s_i\in\{1,0\}$ were treated abstractly as distinguishing whether a measurement $\mathcal{M}_i$ is modeled as producing classical records or as a pure transformation of a theory. We instantiate these in quantum theory according to \cite{Vilasini2022} in this particular example. As in the FR set-up (\cref{fig:fr-entanglement}), in this quantum set-up as well, the setting $s_i=0$ case for $A_i$'s measurement will correspond to the unitary description of the measurement: taking the memory $M_i$ to be initialised to $\ket{0}_{M_i}$ this will be a CNOT gate on the system measured by $A_i$ (control) and $M_i$ (target) in the basis of $A_i$'s measurement. For example, if $A_i$ measures a qubit state $\alpha\ket{0}_S+\beta\ket{1}_S$ in the $\{\ket{0}_S,\ket{1}_S\}$ and stores the outcome in their memory, this unitary is a CNOT which would yield the final state $\alpha\ket{0}_S\ket{0}_{M_i}+\beta\ket{1}_S\ket{1}_{M_i}$ in the $s_i=0$ model of the measurement. The setting $s_i=1$ case will correspond to applying the measurement projectors, which will only be used when reasoning about the resulting probabilities via the Born rule. Notice that the probabilities for $\{\ket{0}_S,\ket{1}_S\}$ in the state $\alpha\ket{0}_S+\beta\ket{1}_S$ are identical to the probabilities for $\{\ket{0}_S\ket{0}_{M_i},\ket{1}_S\ket{1}_{M_i}\}$ in the state $\alpha\ket{0}_S\ket{0}_{M_i}+\beta\ket{1}_S\ket{1}_{M_i}$, for all amplitudes $\alpha$ and $\beta$. Thus the setting $s_i=1$ case can be seen as first applying the unitary description of $s_i=0$ and then the joint projectors on $SM_i$ specified by the basis $\{\ket{0}_S\ket{0}_{M_i},\ket{1}_S\ket{1}_{M_i}\}$. 

This will be useful to keep in mind in our analysis. Our example will only involve pure states, and it will be sufficient to consider the states at each time step obtained via the purely unitary, $s_i=0$ description of each measurement (see \cref{fig:fr-kcbs} for the corresponding unitary circuit), and the default predictions of the scenario which involve using $s_i=1$ will be evident simply by expressing the unitarily evolved state in the basis of the relevant measurement projectors. The idea in the following is to reduce the relevant default predictions of this multi-agent set-up to the probabilities obtained in the KCBS model, equations~\eqref{eq: KCBSpr1}-\eqref{eq: KCBSpr5}, through this instantiation of the settings and thereby show a contradiction between the assumptions of \cref{def: paradox}. Note that the KCBS probabilities of equations~\eqref{eq: KCBSpr1}-\eqref{eq: KCBSpr5} imply the following for conditional probabilities.
\begin{subequations}
\begin{equation}
\label{eq: KCBSpr1_cond}
    P(a_1=0|a_2=1,\psi)=1,
\end{equation}
\begin{equation}
\label{eq: KCBSpr2_cond}
    P(a_2=1|a_3=0,\psi)=1,
\end{equation}
\begin{equation}
\label{eq: KCBSpr3_cond}
    P(a_3=0|a_4=1,\psi)=1,
\end{equation}
\begin{equation}
\label{eq: KCBSpr4_cond}
    P(a_4=1|a_5=0,\psi)=1,
\end{equation}
\begin{equation}
\label{eq: KCBSpr5_cond}
    P(a_5=0,a_1=1|\psi)>0.
\end{equation}
\end{subequations}

We now describe how the same conditional probabilities are reproduced via default setting-conditioned predictions of the multi-agent set-up described in \cref{sec:kcbs}.

\paragraph{Details of the setup and its predictions: }

We now describe each step of the above setup in detail, characterising the states and measurements involved that enable the agents to reason as in Equation~\eqref{eq: 5chain}. Firstly note that the 5 measurements of equations~\eqref{eq: KCBSmmt1}-\eqref{eq: KCBSmmt5} can be implemented using projectors corresponding to the following basis elements, by making a particular basis choice for the subspace corresponding to the ``0'' outcome in each case. Since we are considering a qutrit state, we will have 3 basis vectors and since the measurements are binary valued with the ``1'' outcome corresponding to the basis element $\ket{v_i}$, the complement space  $\mathds{1}-\ket{v_i}\bra{v_i}$ is two-dimensional. We denote this space as span$\{\ket{\overline{v}_i^1}, \ket{\overline{v}_i^2}\}$ where  $\ket{\overline{v}_i^1}$ and $\ket{\overline{v}_i^2}$ are orthogonal to each other and to $\ket{v_i}$ and both correspond to the ``0'' outcome of measurement $X_i$. 
\begin{subequations}
\begin{equation}
\label{eq: basis1}
    \{\ket{v_1},\ket{\overline{v}_1^1},\ket{\overline{v}_1^2}\}:= \{\frac{1}{\sqrt{3}}(\ket{0}-\ket{1}+\ket{2}), \frac{1}{\sqrt{2}}(\ket{1}+\ket{2}),\frac{1}{\sqrt{6}}(-2\ket{0}-\ket{1}+\ket{2})\},
\end{equation}
\begin{equation}
\label{eq: basis2}
    \{\ket{v_2},\ket{\overline{v}_2^1},\ket{\overline{v}_2^2}\}:= \{\frac{1}{\sqrt{2}}(\ket{0}+\ket{1}),\frac{1}{\sqrt{3}}(\ket{0}-\ket{1}+\ket{2}),\frac{1}{\sqrt{6}}(\ket{0}-\ket{1}-2\ket{2})\},
\end{equation}
\begin{equation}
\label{eq: basis3}
    \{\ket{v_3},\ket{\overline{v}_3^1},\ket{\overline{v}_3^2}\}:= \{\ket{2},\frac{1}{\sqrt{2}}(\ket{0}+\ket{1}),\frac{1}{\sqrt{2}}(\ket{0}-\ket{1})\},
\end{equation}
\begin{equation}
\label{eq: basis4}
    \{\ket{v_4},\ket{\overline{v}_4^1},\ket{\overline{v}_4^2}\}:= \{\ket{0},\ket{1},\ket{2}\},
\end{equation}
\begin{equation}
\label{eq: basis5}
    \{\ket{v_5},\ket{\overline{v}_5^1},\ket{\overline{v}_5^2}\}:= \{\frac{1}{\sqrt{2}}(\ket{1}+\ket{2}),\frac{1}{\sqrt{2}}(\ket{1}-\ket{2}),\ket{0}\}.
\end{equation}
\end{subequations}
Now at time $t=0$, $A_2$ measures the state $\ket{\psi^{t=0}}_S:=\ket{\psi}_S$ (Equation~\eqref{eq: KCBSstate}), which when expressed in the basis of $X_2$ (Equation~\eqref{eq: basis2}) is
\begin{equation}
\label{eq: t=0psi}
    \ket{\psi^{t=0}}_S=\frac{\sqrt{2}}{\sqrt{3}}\ket{v_2}_S+\frac{1}{3}\ket{\overline{v}_2^1}_S-\frac{\sqrt{2}}{3}\ket{\overline{v}_2^2}_S.
\end{equation}

Using the von-Neumann picture of a measurement as a unitary evolution on the initial state $\ket{\psi^{t=0}}_S\otimes\ket{0}_{M_2}$ (which constitutes the description of $A_2$'s measurement with setting $s_2=0$, with the memory initialised to $\ket{0}_{M_2}$), the state of the system $S$ and $A_2$'s memory $M_2$ after the measurement is
\begin{equation}
\label{eq: t=1psi}
\begin{split}
       \ket{\psi^{t=1}}_{SM_2}&=\frac{\sqrt{2}}{\sqrt{3}}\ket{v_2}_S\ket{1}_{M_2}+\frac{1}{3}\ket{\overline{v}_2^1}_S\ket{0}_{M_2}-\frac{\sqrt{2}}{3}\ket{\overline{v}_2^2}_S\ket{0}_{M_2}\\
    &=\frac{\sqrt{2}}{\sqrt{3}}\ket{\overline{v}_3^1}_S\ket{1}_{M_2}+\frac{1}{\sqrt{3}}\ket{v_3}_S\ket{0}_{M_2}.
\end{split}
\end{equation}

{\it Reproducing the KCBS probability \cref{eq: KCBSpr1_cond}: } 
Consider what happens when $A_2$ sees the outcome $a_2=1$. Given the knowledge of the steps of the protocol, the state evolution at each time step, from $A_2$'s perspective (upon seeing $a_2=1$) would be as follows, where we use the bold superscript $\mathbf{A_2}$ to denote that the following states are calculated from the perspective of $A_2$. The first step involves projecting the state in \cref{eq: t=1psi} using $\ket{\overline{v}_3^1}_S\ket{1}_{M_2}$ (hence using setting $s_2=1$ as this is a classical record for $A_2$). It does not impact the argument whether we use the sub-normalised or normalised version of the projected state and we drop the normalisation for brevity. Then we evolve this projected state using the purely unitary description of $A_3$, $A_4$, $A_5$'s measurements (i.e., settings $s_3=s_4=s_5=0$) and we will consider what outcome $A_1$ who measures last, would obtain (this uses setting $s_1=1$ as we will refer to their classical outcome). This allows to compute the default prediction $P(a_1=0|a_2=1,s_1=s_2=1,s_3=s_4=s_5=0)$. When $A_2$ observes the outcome $a_2=1$ at time $t=1$, they ascribe the state
    \begin{subequations}
    \begin{equation}
        \label{eq: t=1psiA2}
       \ket{\psi^{t=1}}_{SM_2}^{\mathbf{A_2}}=\ket{v_2}_S\ket{1}_{M_2}=\ket{\overline{v}_3^1}_S\ket{1}_{M_2},
    \end{equation}

    After $A_3$'s measurement (modelled unitarily, with $s_3=0$), this yields
        \begin{equation}
        \label{eq: t=2psiA2}
       \ket{\psi^{t=2}}_{SM_2M_3}^{\mathbf{A_2}}=\ket{\overline{v}_3^1}_S\ket{1}_{M_2}\ket{0}_{M_3},
    \end{equation}

    $A_4$'s unitary $U_{X_2}$ undoes $X_2$'s memory update, yielding
    \begin{equation}
        \label{eq: t=2psi2A2}
        \begin{split}
    \ket{\psi^{t=2}_{U_{X_2}}}_{SM_3M_2}^{\mathbf{A_2}}&=\ket{\overline{v}_3^1}_S\ket{0}_{M_3}\ket{0}_{M_2}\\
    &=\big(\frac{1}{\sqrt{2}}\ket{v_4}\ket{0}_{M_3}+\frac{1}{\sqrt{2}}\ket{\overline{v}_4^2}\ket{0}_{M_3}\big)\ket{0}_{M_2},
        \end{split}
    \end{equation}

    $A_4$'s measurement modeled as a unitary evolution (i.e., setting $s_4=0$) then gives the following, after which we apply $A_5$'s unitary $U_{X_3}$, $A_5$'s measurement modeled unitarily (i.e., setting $s_5=0$), and $A_1$'s
 unitary $U_{X_4}$  to obtain the subsequent equations for the state after each step
 
 \begin{equation}
 \label{eq: t=3psiA2}
     \ket{\psi^{t=3}_{U_{X_2}}}_{SM_3M_4M_2}^{\mathbf{A_2}}=\big(\frac{1}{\sqrt{2}}\ket{v_4}\ket{0}_{M_3}\ket{1}_{M_4}+\frac{1}{\sqrt{2}}\ket{\overline{v}_4^1}\ket{0}_{M_3}\ket{0}_{M_4}\big)\ket{0}_{M_2}
 \end{equation}
 \begin{equation}
   \label{eq: t=3psi2A2}  
   \begin{split}
  \ket{\psi^{t=3}_{U_{X_3}U_{X_2}}}_{SM_4M_3M_2}^{\mathbf{A_2}}&=\big(\frac{1}{\sqrt{2}}\ket{v_4}\ket{1}_{M_4}+\frac{1}{\sqrt{2}}\ket{\overline{v}_4^1}\ket{0}_{M_4}\big)\ket{0}_{M_3}\ket{0}_{M_2}\\
  &=\big(\frac{1}{\sqrt{2}}\ket{\overline{v}_5^2}\ket{1}_{M_4}+\frac{1}{2}\ket{v_5}\ket{0}_{M_4}+\frac{1}{2}\ket{\overline{v}_5^1}\ket{0}_{M_4}\big)\ket{0}_{M_3}\ket{0}_{M_2},
   \end{split}
 \end{equation}
 \begin{equation}
 \label{eq: t=4psiA2} 
 \ket{\psi^{t=4}_{U_{X_3}U_{X_2}}}_{SM_4M_5M_3M_2}^{\mathbf{A_2}}= \big(\frac{1}{\sqrt{2}}\ket{\overline{v}_5^2}\ket{1}_{M_4}\ket{0}_{M_5}+\frac{1}{2}\ket{v_5}\ket{0}_{M_4}\ket{1}_{M_5}+\frac{1}{2}\ket{\overline{v}_5^1}\ket{0}_{M_4}\ket{0}_{M_5}\big)\ket{0}_{M_3}\ket{0}_{M_2}   
 \end{equation}
 \begin{equation}
     \label{eq: t=4psi2A2} 
     \begin{split}
         \ket{\psi^{t=4}_{U_{X_4}U_{X_3}U_{X_2}}}_{SM_5M_4M_3M_2}^{\mathbf{A_2}}&=\big(\frac{1}{\sqrt{2}}\ket{\overline{v}_5^2}\ket{0}_{M_5}+\frac{1}{2}\ket{v_5}\ket{1}_{M_5}+\frac{1}{2}\ket{\overline{v}_5^1}\ket{0}_{M_5}\big)\ket{0}_{M_4}\ket{0}_{M_3}\ket{0}_{M_2}\\
         &=\big(-\frac{\sqrt{3}}{2}\ket{\overline{v}_1^2}\ket{0}_{M_5}+\frac{1}{2}\ket{\overline{v}_1^1}\ket{1}_{M_5}\big)\ket{0}_{M_4}\ket{0}_{M_3}\ket{0}_{M_2},
     \end{split}
 \end{equation}
    \end{subequations}

This state has zero overlap with $A_1$'s measurement projector $\ket{v_1}_S$ which corresponds to $a_1=1$, this is also explicitly seen when one calculates $A_1$'s post-measurement state applying the unitary description of $A_1$'s measurement (as below) and seeing that the projector $\ket{v_1}_S\ket{1}_{M_1}$ (associated with $s_1=1$) annihilates this state.
 \begin{equation}
   \label{eq: t=5psiA2} 
  \ket{\psi^{t=5}_{U_{X_4}U_{X_3}U_{X_2}}}_{SM_5M_1M_4M_3M_2}^{\mathbf{A_2}}=\big(-\frac{\sqrt{3}}{2}\ket{\overline{v}_1^2}\ket{0}_{M_5}\ket{0}_{M_1}+\frac{1}{2}\ket{\overline{v}_1^1}\ket{1}_{M_5}\ket{0}_{M_1}\big)\ket{0}_{M_4}\ket{0}_{M_3}\ket{0}_{M_2}.   
 \end{equation}
 
 Therefore $A_2$ would obtain the following default prediction which mirrors \cref{eq: KCBSpr1_cond}.
 \begin{equation}
     \label{eq: KCBS_prediction1}
     P(a_1=0|a_2=1,s_1=s_2=1,s_3=s_4=s_5=0)=1
 \end{equation}

{\it Reproducing the KCBS probability \cref{eq: KCBSpr2_cond}: } For this, we need to consider $A_3$'s measurement. First modelling $A_3$'s measurement on $S$ as a unitary evolution on $SM_3$ with $M_3$ initialised to $\ket{0}$ gives the following state on $SM_2M_3$ at the next time step, $t=2$.
\begin{equation}
    \label{eq: t=2psi}
\begin{split}
       \ket{\psi^{t=2}}_{SM_2M_3}
    &=\frac{\sqrt{2}}{\sqrt{3}}\ket{\overline{v}_3^1}_S\ket{1}_{M_2}\ket{0}_{M_3}+\frac{1}{\sqrt{3}}\ket{v_3}_S\ket{0}_{M_2}\ket{1}_{M_3}\\
    &=\frac{\sqrt{2}}{\sqrt{3}}\ket{v_2}_S\ket{1}_{M_2}\ket{0}_{M_3}+\frac{1}{3}\ket{\overline{v}_2^1}_S\ket{0}_{M_2}\ket{1}_{M_3}-\frac{\sqrt{2}}{3}\ket{\overline{v}_2^2}_S\ket{0}_{M_2}\ket{1}_{M_3}.
\end{split}
\end{equation}

Now, to consider the case where $A_3$ observes $a_3=0$, we must apply the projector $\ket{v_2}_{S}\ket{0}_{M_3}$ (thus the setting $s_3=1$), but this has zero-overlap with terms where $\ket{0}_{M_2}$ i.e. $A_2$ observed $a_2=0$. Since the remaining parties act after $A_2$ and $A_3$, the joint probability of $a_2$ and $a_3$ will be independent of the settings $s_4$, $s_5$ and $s_1$ (follows from Theorem of \cite{Vilasini2022}) and we can w.l.o.g. take these to be 0. This yields the following default prediction analogous to \cref{eq: KCBSpr2_cond}.

\begin{equation}
    \label{eq: KCBS_prediction2}
    P(a_2=1|a_3=0,s_1=0,s_2=s_3=1,s_4=s_5=0)=1
\end{equation}

{\it Reproducing the KCBS probability \cref{eq: KCBSpr3_cond}: } 
This concerns $A_3$ and $A_4$'s outcomes. 
In the given set-up, $A_4$ jointly operates upon $SM_2$ as follows: they first perform a unitary $U_{X_2}$ on $SM_2$ that ``undoes'' the unitary corresponding to $A_2$'s measurement (which was incompatible to $A_4$'s measurement) i.e., after the unitary, the state on $M_2$ (obtained using the $s_2=0$ unitary description) gets reset to $\ket{0}$, allowing $A_4$ to measure $X_4$ on $S$, as though the only measurement preceding it was the compatible measurement of $X_3$. The unitary $U_{X_2}:=\ket{v_2}\bra{v_2}_S\otimes(\sigma_X)_{M_2}+(\mathds{1}-\ket{v_2}\bra{v_2})_S\otimes\mathds{1}_{M_2}$ is CNOT in the basis of $X_2$, with $S$ as the control and $M_2$ as the target, where $\sigma_X$ denotes the Pauli $X$ operator that implements a qubit bit flip. Then the joint state of $SM_2M_3$ after the application of $U_{X_2}$ is denoted as $\ket{\psi^{t=2}_{U_{X_2}}}_{SMM_3M_2}$ and given as
\begin{equation}
\label{eq: t=2psi2}
\begin{split}
       \ket{\psi^{t=2}_{U_{X_2}}}_{SM_3M_2}&=\big(\frac{\sqrt{2}}{\sqrt{3}}\ket{\overline{v}_3^1}_S\ket{0}_{M_3}+\frac{1}{\sqrt{3}}\ket{v_3}_S\ket{1}_{M_3}\big)\ket{0}_{M_2}\\
   &=\big(\frac{1}{\sqrt{3}}\ket{v_4}_S\ket{0}_{M_3}+\frac{1}{\sqrt{3}}\ket{\overline{v}_4^1}_S\ket{0}_{M_3}+\frac{1}{\sqrt{3}}\ket{\overline{v}_4^2}_S\ket{1}_{M_3}\big)\ket{0}_{M_2}
\end{split}
\end{equation}
Then, after $A_4$'s measurement of $X_4$ on $S$, described in unitary picture, the state of $SM_3M_4M_2$ would be
\begin{equation}
    \label{eq: t=3psi}
    \ket{\psi^{t=3}_{U_{X_2}}}_{SM_3M_4M_2}=\big(\frac{1}{\sqrt{3}}\ket{v_4}_S\ket{0}_{M_3}\ket{1}_{M_4}+\frac{1}{\sqrt{3}}\ket{\overline{v}_4^1}_S\ket{0}_{M_3}\ket{0}_{M_4}+\frac{1}{\sqrt{3}}\ket{\overline{v}_4^2}_S\ket{1}_{M_3}\ket{0}_{M_4}\big)\ket{0}_{M_2}.
\end{equation}

To consider the case where $A_4$ observes $a_4=1$, we apply the projector $\ket{v_4}_S\ket{1}_{M_4}$ (which corresponds to setting $s_4=1$), but this has zero overlap with $\ket{1}_{M_3}$, and yields the following default prediction replicating \cref{eq: KCBSpr3_cond}. Here, recall that we chose $s_2=0$ for $A_2$'s measurement while $A_5$ and $A_1$ act later allowing their setting to be taken as $a_1=s_5=0$ as before (w.l.o.g as the prediction is independent of this setting).
\begin{equation}
\label{eq: KCBS_prediction3}
    P(a_3=0|a_4=1,s_1=s_2=0, s_3=s_4=1,s_5=0)=1
\end{equation}

{\it Reproducing the KCBS probability \cref{eq: KCBSpr4_cond}: } This involves $A_4$ and $A_5$'s outcomes. 
Similarly to the previous step, $A_5$'s measurement on $SM_3$ consists of a unitary $U_{X_3}:=\ket{v_3}\bra{v_3}_S\otimes(\sigma_X)_{M_3}+(\mathds{1}-\ket{v_3}\bra{v_3})_S\otimes\mathds{1}_{M_3}$ (a CNOT in the $X_3$ basis that ``undoes'' $A_3$'s incompatible measurement), followed by a measurement of $X_5$ on $S$. The state $\ket{\psi^{t=3}_{U_{X_3}U_{X_2}}}_{SM_4M_3M_2}$ after $A_5$ applies $U_{X_3}$ is
\begin{equation}
\label{eq: t=3psi2}
\begin{split}
       \ket{\psi^{t=3}_{U_{X_3}U_{X_2}}}_{SM_4M_3M_2}&=\big(\frac{1}{\sqrt{3}}\ket{v_4}_S\ket{1}_{M_4}+\frac{1}{\sqrt{3}}\ket{\overline{v}_4^1}_S\ket{0}_{M_4}+\frac{1}{\sqrt{3}}\ket{\overline{v}_4^2}_S\ket{0}_{M_4}\big)\ket{0}_{M_3}\ket{0}_{M_2}\\
   &=\big(\frac{1}{\sqrt{3}}\ket{\overline{v}_5^1}_S\ket{1}_{M_4}+\frac{\sqrt{2}}{\sqrt{3}}\ket{v_5}_S\ket{0}_{M_4}\big)\ket{0}_{M_3}\ket{0}_{M_2},
\end{split}
\end{equation}
and the state of $SM_4M_5M_3M_2$ after $A_5$'s measurement (applying the unitary model of the measurement, or setting $s_5=0$) is
\begin{equation}
\label{eq: t=4psi}
    \ket{\psi^{t=4}_{U_{X_3}U_{X_2}}}_{SM_4M_5M_3M_2}=\big(\frac{1}{\sqrt{3}}\ket{\overline{v}_5^1}_S\ket{1}_{M_4}\ket{0}_{M_5}+\frac{\sqrt{2}}{\sqrt{3}}\ket{v_5}_S\ket{0}_{M_4}\ket{1}_{M_5}\big)\ket{0}_{M_3}\ket{0}_{M_2}.
\end{equation}

Reasoning exactly as in the previous steps, we see that projecting on $A_5$ seeing the outcome $a_5=0$ (which will involve setting $s_5=1$) we have zero-overlap with the projector for $a_4=0$ (involves $s_4=1$) and we have used the unitary description for the remaining measurements, which gives the following default prediction which reproduces \cref{eq: KCBSpr4_cond}.
\begin{equation}
\label{eq: KCBS_prediction4}
    P(a_4=1|a_5=0,s_1=s_2=s_3=0,s_4=s_5=1).
\end{equation}

{\it Reproducing the KCBS probability \cref{eq: KCBSpr5_cond}: } 
Finally $A_1$, performs the unitary $U_{X_4}:=\ket{v_4}\bra{v_4}_S\otimes(\sigma_X)_{M_4}+(\mathds{1}-\ket{v_4}\bra{v_4})_S\otimes\mathds{1}_{M_4}$ (a CNOT in the $X_4$ basis that ``undoes'' $A_4$'s incompatible measurement, when it is described using $s_4=0$) after which the state is
\begin{equation}
    \label{eq: t=4psi2}
    \begin{split}
    \ket{\psi^{t=4}_{U_{X_4}U_{X_3}U_{X_2}}}_{SM_5M_4M_3M_2}&=\big(\frac{1}{\sqrt{3}}\ket{\overline{v}_5^1}_S\ket{0}_{M_5}+\frac{\sqrt{2}}{\sqrt{3}}\ket{v_5}_S\ket{1}_{M_5}\big)\ket{0}_{M_4}\ket{0}_{M_3}\ket{0}_{M_2}\\
    &=\big(\frac{1}{3}\ket{v_1}_S\ket{0}_{M_5}-\frac{\sqrt{2}}{3}\ket{\overline{v}_1^2}_S\ket{0}_{M_5}+\frac{\sqrt{2}}{\sqrt{3}}\ket{\overline{v}_1^1}_S\ket{1}_{M_5}\big)\ket{0}_{M_4}\ket{0}_{M_3}\ket{0}_{M_2},
    \end{split}
\end{equation}
and then $A_1$ measures $X_1$ on $S$ giving the following final state by modelling $A_1$'s measurement unitarily.
\begin{align}
    \label{eq: t5psi}
  \begin{split}      &\ket{\psi^{t=5}_{U_{X_4}U_{X_3}U_{X_2}}}_{SM_5M_1M_4M_3M_2}\\=&\big(\frac{1}{3}\ket{v_1}_S\ket{0}_{M_5}\ket{1}_{M_1}-\frac{\sqrt{2}}{3}\ket{\overline{v}_1^2}_S\ket{0}_{M_5}\ket{0}_{M_1}+\frac{\sqrt{2}}{\sqrt{3}}\ket{\overline{v}_1^1}_S\ket{1}_{M_5}\ket{0}_{M_1}\big)\ket{0}_{M_4}\ket{0}_{M_3}\ket{0}_{M_2}.
        \end{split}  
\end{align}

To compute the joint probability for $a_5=0$ and $a_1=1$ we must apply a projector on $\ket{0}_{M_5}\ket{1}_{M_1}$ (this involves settings $s_1=s_5=1$). This has non-zero overlap with the above state and yields the default prediction below, which mirrors \cref{eq: KCBSpr5_cond}.
\begin{equation}
    \label{eq: KCBS_prediction5}
    P(a_5=0,a_1=1|s_1=1,s_2=s_3=s_4=0,s_5=1)>0.
\end{equation}

\paragraph{Reasoning: } 
Since for every pair of adjacent agents, neither agent measures the others' memory and both perform mutually compatible measurements on the same system, the agents 
can reason using the predictions of equations \eqref{eq: KCBS_prediction1} to \eqref{eq: KCBS_prediction5}
while being compatible with the {\bf reasoning about compatible agents} assumption. Moreover, we take the setup and quantum theory to be in the {\bf common knowledge} of all agents (this was used in deriving the above-mentioned predictions), and further imposing the {\bf setting-independence} assumption allows us to ignore all setting-labels in the statements derived from these default predictions (i.e., drop the $\vec{s}$ subscript in the statements given by applying \Cref{def:language_v2} to equations \eqref{eq: KCBS_prediction1} to \eqref{eq: KCBS_prediction5}). 
When the post-selection $a_1=1\land a_5=0$ succeeds (which it will at some point, according to \eqref{eq: KCBS_prediction5}), Agent $A_5$ reasons as follows:

\begin{enumerate}
   \item \emph{$A_5$ reasons about $A_4$'s outcome.} $A_5$ considers the state $\ket{\psi^{t=3}_{U_{X_3}U_{X_2}}}$ just before measuring $X_5$ on $S$ at $t=3$ (and after applying the unitary $U_{X_3}$). $A_5$ then reasons using the second line of Eq.~\eqref{eq: t=3psi2} that having observed the outcome $a_5=0$, the state of the system must be $\ket{\overline{v}_5^1}_S$ in which case, the state of $M_4$ must be $\ket{1}_{M_4}$. This gives the prediction of \eqref{eq: KCBS_prediction4}, which allows $A_5$ to conclude that $A_4$ must have observed $a_4=1$ obtaining the following statement, after using {\bf setting independence}. 
   $$ K_5 (a_5=0 \implies a_4=1).$$
   \item \emph{$A_5$ reasons about $A_4$'s reasoning.} $A_5$ reasons that $A_4$ would have seen the pre-measurement state $\ket{\psi^{t=2}_{U_{X_2}}}$ just before measuring $X_4$ on $S$ at $t=2$. Then, $A_4$ would have reasoned using the second line of Eq.~\eqref{eq: t=2psi2} that having observed $a_4=1$, the state of the system is $\ket{v_4}_S$ in which case, the state of $M_3$ is $\ket{0}_{M_3}$. That is $A_4$ would establish the prediction of \cref{eq: KCBS_prediction3} and $A_4$ would therefore conclude that $A_3$ must have observed $a_3=0$, and obtain the following statement after ignoring settings
    $$ K_5K_4 (a_4=1 \implies a_3=0).$$
    \item \emph{$A_5$ reasons about $A_4$'s reasoning about $A_3$'s reasoning.} $A_5$ reasons that $A_4$ would have reasoned that $A_3$
    would have seen the pre-measurement state $\ket{\psi^{t=1}}$ just before measuring $X_3$ on $S$ at $t=1$. Then, $A_3$ would have reasoned using the second line of Eq.~\eqref{eq: t=1psi} that having observed $a_3=0$, the state of the system is $\ket{\overline{v}_3^1}_S$ in which case, the state of $M_2$ is $\ket{1}_{M_2}$, and $A_3$ would therefore conclude using \cref{eq: KCBS_prediction2} that $A_2$ must have observed $a_2=1$, obtain the following after dropping the settings. 
    $$ K_5K_4K_3 (a_3=0 \implies a_2=1).$$
    \item \emph{$A_5$ reasons about $A_4$'s reasoning about $A_3$'s reasoning about $A_2$'s reasoning.} $A_5$ reasons that $A_4$ would have reasoned that $A_3$ would have concluded that $A_2$ must have obtained the outcome $a_2=1$ and upon observing this outcome, $A_2$ can reason about what $A_1$ would see, following the steps leading to the default prediction computed in \cref{eq: KCBS_prediction1}. 
This prediction used along with the {\bf setting-independence} assumption of \cref{def: paradox}, would allow $A_2$ to conclude that $a_2=1\Rightarrow a_1=0$ and we have,   
       $$ K_5K_4K_3K_2 (a_2=1 \implies a_1=0).$$
 \item \emph{$A_5$ applies trust relations.}  $A_5$ can then combine the statements obtained in the above items, by applying the trust relations (as entailed by the {\bf reasoning about compatible agents} assumption of \Cref{def: paradox}). As mentioned before, each pair of adjacent parties in the 5-cycle trust each other since they perform mutually compatible measurements on the same system (and do not tamper with each other's memories). The trust relations are applied to each statement as shown in the following example,
    \begin{align*}
         K_5 K_4 K_3 K_2 &(a_2=1 \implies a_1=0) \\
         \implies K_5 K_4 K_3  &(a_2=1 \implies a_1=0) \qquad \because A_2 \leadsto A_3\\
         \implies K_5 K_4  &(a_2=1 \implies a_1=0) \qquad \because A_3 \leadsto A_4\\
         \implies K_5 &(a_2=1 \implies a_1=0), \qquad \because A_4 \leadsto A_5
    \end{align*}
    and similarly for the other statements, so that we obtain 
    \begin{align*}
         & K_5  \big[ (a_5=0 \implies a_4=1) \wedge (a_4=1 \implies a_3=0) \wedge  (a_3=0 \implies a_2=1) \wedge (a_2=1 \implies a_1=0) \big]
    \end{align*} 
    Therefore in any run of the protocol where the post-selection of $a_1=1\land a_5=0$ succeeds, $K_5$ can reason about other agents' as above, using the first three assumptions of \cref{def: paradox} to obtain the paradoxical statement
        $$K_5 (a_1=1\land a_5=0 \implies a_1=0).$$
\end{enumerate}
In other words, assuming the first three assumptions of \Cref{def: paradox}, we have arrived at a contradiction to the fourth assumption, which reveals a multi-agent paradox according to that definition.

\section{Proof of \Cref{theorem: main}}
\label{appendix: proof_main}
\ParadoxImpliesContextuality*
\begin{sloppypar}

\begin{proof}

Consider an $N$-agent multi-agent setup $\mathcal{MA}$ in a theory $\mathbb{T}$. We now consider agents reasoning in this theory using all the assumptions of \Cref{def: paradox} and show that if they arrive at a contradiction, then there exists a logically contextual empirical model in the theory. The first assumption about {\bf common knowledge} allows us to take $\mathbb{T}$ and the setup $\mathcal{MA}$ to be in the common knowledge of the agents $\{A_i\}_{i=1}^N$ of the scenario.

Using the given multi-agent setup $\mathcal{MA}$ we can construct a measurement scenario (\Cref{def: meas_scenario}) and an empirical model for that scenario (\Cref{def: emp_model}), as we detail below.

Recall that a measurement scenario is specified by a set $X$ of variables, a set $\mathcal{C}$ of contexts which are families of jointly measurable subsets of variables in $X$, a set of outcomes $O$ (which assign possible values $O_x$ for every $x\in X$) and an initial state $\rho$. An empirical model for the scenario then specifies, a probability distribution $P_C$ for every context $C\in \mathcal{C}$ of the scenario.

For our given multi-agent setup $\mathcal{MA}$, $X$ is given by the set of all outcome variables $X:=\{a_i\}_{i=1}^N$, one variable $a_i$ for each agent $A_i$ of the setup. Then the contexts for such an $X$ must be $\mathcal{C}\subseteq \text{Powerset}(\{a_i\}_{i=1}^N)$ by definition. More precisely $\mathcal{C}$ is constructed by including, for every compatible set of measurements $\{\mathcal{M}_{j_k}\}_{k=1}^{k=p}$ in the multi-agent setup, the corresponding outcome set $C:=\{a_{j_k}\}_{k=1}^{k=p}$ in $\mathcal{C}$. 
The outcomes $O_x$ for each $x\in X$ are given by the outcome sets $\mathtt{O}_i$ assigned to each measurement outcome $a_i$ of the setup (\Cref{def: MAsetup}). The multi-agent set-up specifies, by definition, an initial state of the system and memories, we take this to be $\rho$. We therefore have a measurement scenario given by $(\{a_i\}_{i=1}^N,\mathcal{C},\{\mathtt{O}_i\}_{i=1}^N, \rho)$.

By \Cref{def:compatibility_v2}, for each set $\{\mathcal{M}_{j_k}\}_{k=1}^{k=p}$ of compatible measurements in the setup, the associated default prediction $P(a_{j_1},...,a_{j_p}|\vec{s})$ can be equivalently computed by applying an operationally equivalent set of measurements (the primed measurements) on the state $\rho$, $P(a_{j_1},...,a_{j_p}|\rho,\{\mathcal{M}'_{j_k}\}_{k=1}^p)$. 
Observe that the primed  measurements are in one-to-one correspondence with the elements of chosen context $C$, and moreover the definition of compatibility ensures that the primed measurements are jointly measurable, thus $C$ defines a valid context. Since we have associated to each context $C:=\{a_{j_k}\}_{k=1}^{k=p}$, a distribution   $e_C:=P(a_{j_1},...,a_{j_p}|\rho,\{\mathcal{M}'_{j_k}\}_{k=1}^p)$, we have now specified an empirical model for the measurement scenario $(\{a_i\}_{i=1}^N,\mathcal{C},\{\mathtt{O}_i\}_{i=1}^N, \rho)$.

The above construction uses the default prediction $P(a_{j_1},...,a_{j_p}|\vec{s})$ associated with the given subset of measurements. Recall that the setting vector $\vec{s}$ in such a prediction involves the setting values $s_{i}=1$ if $i\in\{j_1,...,j_p\}$ and $s_i=0$ otherwise, where $\{j_1,...,j_p\}$ is specified by the subset of measurements under consideration. Therefore the setting values are distinct for the different contexts $C$ that we consider.

The {\bf setting-independence} assumption allows agents in the setup to ignore the specification of these setting values in all the reasoning. Since the details of the setup such as the initial state $\rho$, the measurements $\mathcal{M}_i$ (and therefore the measurements $\mathcal{M}'_i$ constructed from this) are all in the common knowledge of the agents, we also do not condition on these explicitly when referring to statements in the setup. Thus having imposed the {\bf common knowledge} and {\bf setting-independence} assumptions of \Cref{def: paradox}, agents will only be reasoning using statements of the form $\vec{a}_j=\vec{v}_j$ or $\vec{a}_l=\vec{v}_l\Rightarrow \vec{a}_j=\vec{v}_j$, which are the allowed form of statements according to \Cref{def:language_v2} but without setting labels.

Moreover, imposing the {\bf reasoning about compatible agents} assumption, the statements $\vec{a}_j=\vec{v}_j$ and $\vec{a}_l=\vec{v}_l\Rightarrow \vec{a}_j=\vec{v}_j$ can each be only composed of outcomes of compatible measurements. Moreover, applying \Cref{def:language_v2} then tells us that $\vec{a}_j=\vec{v}_j$ is only possible when at least one context $C\supseteq \vec{a}_j$ supports (through the probability $e_C$) the value assignment $\vec{v}_j$ and $\vec{a}_l=\vec{v}_l\Rightarrow \vec{a}_j=\vec{v}_j$ only when at least one context $C\supseteq \vec{a}_j\cup \vec{a}_l$ has a distribution $e_C$ which has zero probability of $ \vec{a}_j=\neg \vec{v}_j \land \vec{a}_l=\vec{v}_l$.

Now suppose that agents reason using the first three assumptions of \Cref{def: paradox} to arrive  at a violation of the last assumption, {\bf non-contradictory outcomes} of \Cref{def: paradox}, i.e., there exists an agent $A$ and a set $\vec{a}_j$ of outcomes associated with a compatible set of measurements such that $K_A(\vec{a}_j=\vec{v}_j\land \vec{a}_j=\neg \vec{v}_j)$. Without loss of generality, suppose that in a given run of the experiment, $\vec{a}_j=\vec{v}_j$ is observed. From what we have established above, it follows that agents reasoning using the first three assumptions can only arrive at $\vec{a}_j=\neg \vec{v}_j$ in this run of the experiment if there exists an agent $A$ who obtains a set of statements of the following form,

\begin{align}
\begin{split}
\label{eq: proof_liarcycle}
    \vec{a}_j&=\vec{v}_j \\
    \land \quad  \vec{a}_j=\vec{v}_j   &\Rightarrow \vec{a}_{l_1}=\vec{v}_{l_1}\\
    \land \quad \vec{a}_{l_1}=\vec{v}_{l_1} &\Rightarrow \vec{a}_{l_2}=\vec{v}_{l_2}\\
    &.\\
   & .\\
    &.\\   
\land\quad  \vec{a}_{l_q}=\vec{v}_{l_q} &\Rightarrow  \vec{a}_j=\neg \vec{a}_j,    
\end{split}
\end{align}

where each pair of sets of outcomes related by an implication $\Rightarrow$ are associated with compatible measurement sets.

Note that in arriving at such a chain, we have already used the trust rule (entailed by the {\bf reasoning about compatible agents} assumption) which allows one agent to inherit the knowledge of another compatible agent e.g., if $K_A(\vec{a}_j=\vec{v}_j \Rightarrow \vec{a}_{l_1}=\vec{v}_{l_1})$ and $K_AK_B(\vec{a}_{l_1}=\vec{v}_{l_1} \Rightarrow \vec{a}_{l_2}=\vec{v}_{l_2})$ where $A$ and $B$ are compatible with each other and all agents whose outcomes appear in the associated statements, then the later implies  $K_A(\vec{a}_{l_1}=\vec{v}_{l_1} \Rightarrow \vec{a}_{l_2}=\vec{v}_{l_2})$ by the trust rule and together with the former yields $K_A(  \vec{a}_j=\vec{v}_j \Rightarrow \vec{a}_{l_1}=\vec{v}_{l_1} \land \vec{a}_{l_1}=\vec{v}_{l_1} \Rightarrow \vec{a}_{l_2}=\vec{v}_{l_2})$. 
This means that the empirical model specifies a non-zero probability for $  \vec{a}_j=\vec{v}_j$, but a zero probability for $\vec{a}_j=\vec{v}_j  \land \vec{a}_{l_1}=\neg \vec{v}_{l_1}$, a zero probability for $\vec{a}_{l_1}=\vec{v}_{l_1}  \land \vec{a}_{l_2}=\neg \vec{v}_{l_2}$,...., and a zero probability for $\vec{a}_{l_q}=\vec{v}_{l_q}  \land \vec{a}_{j}= \vec{v}_{j}$.

Let us assume logical non-contextuality, and prove a contradiction. Logical non-contextuality implies that the value assignment $\vec{a}_j=\vec{v}_j$ must belong to  compatible family of value assignments, which in-particular means that the assignment must respect the support of the distributions in the empirical model, e.g., a zero probability for $\vec{a}_j=\vec{v}_j  \land \vec{a}_{l_1}=\neg \vec{v}_{l_1}$ implies that given the value assignment $\vec{a}_j=\vec{v}_j $ (which we have taken w.l.o.g. to be the observed outcome), we must have the value assignment $ \vec{a}_{l_1}= \vec{v}_{l_1}$ according to the support of the empirical model. Similarly,  a zero probability for $\vec{a}_{l_1}=\vec{v}_{l_1}  \land \vec{a}_{l_2}=\neg \vec{v}_{l_2}$ implies that given the value assignment $\vec{a}_{l_1}=\vec{v}_{l_1} $ (which we have established), we must have the value assignment $ \vec{a}_{l_2}= \vec{v}_{l_2}$ according to the support of the empirical model. Applying this logic to each step of the chain, we finally have that given $\vec{a}_{l_q}=\vec{v}_{l_q}$, we must have $\vec{a}_j=\neg \vec{a}_j$, which contradicts the initial assignment $\vec{a}_j=\vec{a}_j$. This shows that $\vec{a}_j=\vec{a}_j$ cannot belong to a compatible  family of value assignments, and certifies the logical contextuality of the constructed empirical model. This completes the proof.

\end{proof}
    
\end{sloppypar}

\section{Proofs of Section~\ref{sec:character}}
\label{appendix:epistemic-quantum}

\subsection{Proofs for \cref{sec: structure_general}}

\yabloB*
\begin{proof}
A Yablo type chain for the $N$ agents would be one where for some permutation $\pi$ of $\{1,…,N\}$, the following holds for every $1\leq i <N$ 
\begin{align}
\begin{split}
    &a_{\pi(i)}=1 \implies a_{\pi(j)}=0 \quad \forall j>i,\\
&a_{\pi(i)}=0 \implies \exists j>i,\quad a_{\pi(j)}=1.
\end{split}
\end{align}

In particular, $a_{\pi(1)}=1\implies a_{\pi(2)}=…=a_{\pi(N)}=0$. Suppose such a statement arose under any setting choice $\vec{s}$ in the reasoning for multi-agent paradox i.e., $\phi_{\vec{s}}=\Big(a_{\pi(1)}=1\implies a_{\pi(2)}=…=a_{\pi(N)}=0\Big)_{\vec{s}}$ is a statement of the set-up (\cref{def:language_v2}) that is used in the reasoning in accordance with the assumptions of \cref{def: paradox}. Then the {\bf reasoning about compatible agents} assumption would imply that the set $\{a_{\pi(1)},…,a_{\pi(N)}\}$ of outcomes correspond to a compatible set of measurements. Since $\pi$ is simply a permutation, this is in fact the set of all outcomes $\{a_1,..,a_N\}$ of the set-up implying that all $N$ measurements must be mutually compatible or jointly measurable (\cref{def:compatibility_v2}). Therefore the predictions of the protocol are explained by a single well-defined probability distribution on $\{a_1,..,a_N\}$ (independent of settings) which recovers the marginals of all the subsets of measurements. In other words, the associated model is non-contextual (hence logically non-contextual) and cannot admit any multi-agent paradox, according to \cref{theorem: main}.

\end{proof}

\postselect*

\begin{proof}
The setting-independence assumption in the definition of a multi-agent paradox allows us to ignore all the setting labels in the probabilities and statements, as also explained in the paragraph following \cref{def: paradox}. Thus the relevant probabilities will only refer to measurement outcomes of agents. Keeping this in mind, we proceed as follows. We have $P(\vec{a}_1=\vec{v}_1)=P(\vec{a}_1=\vec{v}_1,\vec{a}_N=\neg \vec{v}_N)+P(\vec{a}_1=\vec{v}_1,\vec{a}_N= \vec{v}_N)$, the probabilities on the right are well-defined since $\vec{a}_1$ and $\vec{a}_N$ are jointly measurable. Using $P(\vec{a}_1=\vec{v}_1,\vec{a}_N=\neg \vec{v}_N)=1$ in this along with the non-negativity of all probabilities gives $P(\vec{a}_1=\vec{v}_1)=1$ and $P(\vec{a}_1=\vec{v}_1,\vec{a}_N= \vec{v}_N)=0$. 

Next, the implication $\vec{a}_1=\vec{v}_1\Rightarrow \vec{a}_2=\vec{v}_2$ gives $P(\vec{a}_2=\vec{v}_2|\vec{a}_1=\vec{v}_1)=\frac{P(\vec{a}_2=\vec{v}_2,\vec{a}_1=\vec{v}_1)}{P(\vec{a}_1=\vec{v}_1)}=1$. Using this with $P(\vec{a}_1=\vec{v}_1)=1$, we get $P(\vec{a}_2=\vec{v}_2,\vec{a}_1=\vec{v}_1)=1$. Since $P(\vec{a}_2=\vec{v}_2)=P(\vec{a}_2=\vec{v}_2,\vec{a}_1=\vec{v}_1)+P(\vec{a}_2=\vec{v}_2,\vec{a}_1=\neg \vec{v}_1)$, we must have $P(\vec{a}_2=\vec{v}_2)=1$ and $P(\vec{a}_2=\vec{v}_2,\vec{a}_1=\neg \vec{v}_1)=0$. Repeating the same argument, we get $P(\vec{a}_i=\vec{v}_i)=1$ for all $i\in\{1,…,N\}$.

In particular, this implies $P(\vec{a}_N=\neg\vec{v}_N)=0=P(\vec{a}_N=\neg\vec{v}_N,\vec{a}_1=\vec{v}_1)+P(\vec{a}_N=\neg\vec{v}_N,\vec{a}_1=\neg \vec{v}_1)$ which in turn implies  $P(\vec{a}_1=\vec{v}_1,\vec{a}_N=\neg\vec{v}_N)=0$. This contradicts our initial assumption that  $P(\vec{a}_1=\vec{v}_1,\vec{a}_N=\neg\vec{v}_N)=1$. Therefore such a chain cannot exist in any theory. 
\end{proof}

\negation*

\begin{proof}

As noted in the previous proof, in analysing the structure of multi-agent paradoxes according to \cref{def: paradox}, we can ignore setting labels in probabilities and statements, without loss of generality, to simplify the analysis. Let us consider the joint probability distribution $P(\vec{a},\vec{b})$ over measurement outcomes $\vec{a}$ and $\vec{b}$. From the inference $\vec{a}=\vec{v}_a \Rightarrow \vec{b} = \vec{v}_b$ we have:
\begin{gather*}
    P(\vec{b} = \vec{v}_b | \vec{a} = \vec{v}_a) = \frac{ P(\vec{a} = \vec{v}_a, \vec{b} = \vec{v}_b)}{ P(\vec{a} = \vec{v}_a) } = 1 \Rightarrow \\
    P(\vec{a} = \vec{v}_a, \vec{b} = \vec{v}_b) = P(\vec{a} = \vec{v}_a)=1 \ \text{and} \ P(\vec{a} = \vec{v}_a, \vec{b} = \neg \vec{v}_b) = 0.
\end{gather*}
Then the conditional probability $P(\vec{a}=\neg\vec{v}_a|\vec{b}=\neg\vec{v}_b)$ is given as
\begin{align*}
 P(\vec{a}=\neg\vec{v}_a|\vec{b}=\neg\vec{v}_b) &= \frac{P(\vec{a}=\neg\vec{v}_a,\vec{b}=\neg\vec{v}_b)}{P(\vec{b}=\neg\vec{v}_b)} \\ &= \frac{1 - P(\vec{a}=\neg\vec{v}_a, \vec{b}=\vec{v}_b) - P(\vec{a}=\vec{v}_a, \vec{b}=\neg\vec{v}_b) - P(\vec{a}=\vec{v}_a, \vec{b}=\vec{v}_b)}{P(\vec{b}=\neg\vec{v}_b)} \\ &= \frac{1 - P(\vec{a}=\neg\vec{v}_a, \vec{b}=\vec{v}_b) - P(\vec{a}=\vec{v}_a, \vec{b}=\vec{v}_b)}{P(\vec{b}=\neg\vec{v}_b)} \\ &= \frac{1 - P(\vec{b}=\vec{v}_b) }{P(\vec{b}=\neg\vec{v}_b)} = 1.
\end{align*}

From this unit conditional probability we obtain the negated inference $\vec{b}=\neg\vec{v}_b \Rightarrow \vec{a}=\neg\vec{v}_a$, which proves the claim.
\end{proof}

\ncycleParadox*

\begin{proof} We carry out the two directions of the proof separately.

\paragraph{Paradox implies extremal vertex} Suppose that a theory $\mathbb{T}$ admits a post-selection free $n$-cycle paradox. Then it must lead to the two reasoning chains given in \cref{eq: no_ps_paradox}. Consider the first implications of the two chains: $x_i=0\Rightarrow x_{i+1}=v_{i+1}$ and $x_i=1\Rightarrow x_{i+1}=v'_{i+1}$. 

Using \cref{def:language_v2}, the former implies $P(x_i=0,x_{i+1}\neq v_{i+1})=0$ and the latter implies $P(x_i=1,x_{i+1}\neq v'_{i+1})=0$. Here we have used the fact that the assumptions of \cref{def: paradox} allow us to translate all arguments relative to the setting-condition predictions appearing in \cref{def:language_v2} equivalently at the level of the probabilities $e_{C_i}:=P(x_i,x_{i+1})$ of the associated empirical model (as explicitly shown in the proof of \cref{theorem: main}). Since $v_{i+1}\neq v'_{i+1}$ and $v_{i+1},v'_{i+1}\in\{0,1\}$ by \Cref{def: paradox_no_ps}, we have the following cases
\begin{itemize}
    \item If $v_{i+1}=0$, then we must have $P(x_i=0,x_{i+1}=1)=0$ and $P(x_i=1,x_{i+1}=0)=0$. This implies $P(x_i=x_{i+1})=1$ and hence $<X_i,X_{i+1}>=1$. In this case, choose $\gamma_i=+1$.
    \item If $v_{i+1}=1$, then we must have $P(x_i=0,x_{i+1}=0)=0$ and $P(x_i=1,x_{i+1}=1)=0$. This implies $P(x_i\neq x_{i+1})=1$ and hence $<X_i,X_{i+1}>=-1$. In this case, choose $\gamma_i=-1$.
\end{itemize}
 By repeating the above argument for each implication in the two chains, we obtain $\{\gamma_i\in\{+1,-1\}\}_{i=1}^n$ and an $\Omega$ defined as in \Cref{eq: ncyc_omega} relative to this choice where $\gamma_i=<X_i,X_{i+1}>$ for every $i\in\{1,...,n\}$. This is precisely the condition on correlations that defines an extremal vertex of the $n$-cycle polytope. Moreover, since each chain starts with a value of $x_i$ and arrives at an opposite value of $x_i$, this is only possible if the value was flipped an odd number of times as we go along the implications of the chain. From the above construction, this implies that the number of $\gamma_i=-1$ is odd as required in the definition of $\Omega$. Hence we have shown that the theory $\mathbb{T}$ contains such an extremal vertex of the non-disturbance polytope of the given $n$-cycle scenario.

\paragraph{Extremal vertex implies paradox} Suppose $\mathbb{T}$ contains an extremal vertex of an $n$-cycle scenario i.e., for some $\{\gamma_i\in\{+1,-1\}\}_{i=1}^n$, the theory admits correlations between the $n$-measurements $X_1,...,X_n$ that achieve $\gamma_i=<X_i,X_{i+1}>$ for every $i\in\{1,...,n\}$ and hence $\Omega=n$ for the associated function $\Omega$ as per \Cref{eq: ncyc_omega}. Then select any $i\in\{1,...,n\}$ and consider the following two possible cases
\begin{itemize}
    \item If $\gamma_i=<X_i,X_{i+1}>=+1$, this implies $P(x_i=x_{i+1})=1$. Then we have $x_i=0\Rightarrow x_{i+1}=0$ and $x_i=1\Rightarrow x_{i+1}=1$. Use these as the first implication in the first and second chains of \cref{eq: no_ps_paradox} respectively in the corresponding multi-agent paradox. 
    \item If $\gamma_i=<X_i,X_{i+1}>=-1$, this implies $P(x_i\neq x_{i+1})=1$. Then we have $x_i=0\Rightarrow x_{i+1}=1$ and $x_i=1\Rightarrow x_{i+1}=0$. Use these as the first implication in the first and second chains of \cref{eq: no_ps_paradox} respectively in the corresponding multi-agent paradox. 
\end{itemize}
For the remaining implications of the chain, follow the same procedure. Since the number of $\gamma_i=-1$ is odd, it is guaranteed through this construction that both chains end up with the opposite value assignment to $x_i$ than the starting value. This establishes the existence of a post-selection free multi-agent paradox in the theory $\mathbb{T}$.

\end{proof}

\subsection{Proofs of \cref{sec: structure_quantum}}
In the following proofs concerning quantum multi-agent paradoxes, we will simplify the analysis by applying the {\bf reasoning about compatible agents} and {\bf setting-independence assumptions} involved in multi-agent paradoxes (\cref{def: paradox}). Specifically, when computing a default prediction such as $P(a=1,b=1|\vec{s})$ where $a$ and $b$ are outcomes of compatible measurements, \cref{def:compatibility_v2} ensures that we can compute this equivalently through (primed) measurements acting directly on the initial state $\rho$ of the multi-agent set-up. In the quantum case, we have modeled these as projective measurements w.l.o.g. as explained in \cref{sec: structure_quantum}. Setting-independence allows to ignore the setting labels $\vec{s}$ in the probabilities and associated statements. Hence we can express such default predictions as $P(a=1,b=1)=\tr(\pi_a^1\pi_b^1\rho)$ after ignoring the settings, where the projectors $\pi_a^1$ and $\pi_b^1$ corresponding to the outcomes $a=1$ and $b=1$ are those of the primed measurements of \cref{def:compatibility_v2}, which are by construction associated with the same outcomes as the original ones.

\reduction*

\begin{proof}
In this proof, for brevity and clarity of notation, we will represent each value assignment $\vec{v}_i$ for $i\in\{a,b,c\}$ as the value 1. This is without loss of generality as 1 is an arbitrary label, just as the initial value assignments $\vec{v}_i$ are. In this case the set of values given by the negation $\neg \vec{v}_i$ will simply be denoted as 0.

Let us denote projectors corresponding to the three sets of measurements as $\{\pi_\vec{a}^0,\pi_\vec{a}^1\}$, $\{\pi_\vec{c}^0,\pi_\vec{c}^1\}$ and $\{\pi_\vec{b}^1,\pi_\vec{b}^0\}$ respectively, and consider the density matrix of the global joint state of the whole chain $\rho$. Then we can rewrite the conditional probabilities of the chain as:

\begin{gather*}
    P(\vec{b}=1|\vec{c}=1)=\frac{\tr(\pi_\vec{b}^1\pi_\vec{c}^1\rho)}{\tr(\pi_\vec{c}^1\rho)}=1 \Rightarrow \\ \tr(\pi_\vec{b}^1\pi_\vec{c}^1\rho) = tr(\pi_\vec{c}^1\rho) \ \text{and} \ \tr(\pi_\vec{b}^0\pi_\vec{c}^1\rho)=0 \\
    P(\vec{a}=1|\vec{b}=1)=\frac{\tr(\pi_\vec{a}^1\pi_\vec{b}^1\rho)}{\tr(\pi_\vec{b}^1\rho)}=1 \Rightarrow \\ \tr(\pi_\vec{a}^1\pi_\vec{b}^1\rho)=\tr(\pi_\vec{b}^1\rho) \ \text{and} \ \tr(\pi_\vec{a}^0\pi_\vec{b}^1\rho)=0
\end{gather*}

Starting with $\tr(\pi_\vec{c}^1\rho)$ and $\tr(\pi_\vec{a}^0\pi_\vec{c}^1\rho)$, by inserting an identity, resolved in terms of a relevant set of projectors $\id=\pi_{\vec{a}}^0+\pi_{\vec{a}}^1=\pi_{\vec{b}}^0+\pi_{\vec{b}}^1$ and using the linearity of the trace, we obtain

\begin{gather*}
    \tr(\pi_\vec{a}^1\pi_\vec{c}^1\rho)=\tr(\pi_\vec{c}^1\rho)-\tr(\pi_\vec{a}^0\pi_\vec{c}^1\rho) \\
    \tr(\pi_\vec{a}^0\pi_\vec{c}^1\rho)=\tr(\pi_\vec{b}^1\pi_\vec{a}^0\pi_\vec{c}^1\rho) + \tr(\pi_\vec{b}^0\pi_\vec{a}^0\pi_\vec{c}^1\rho)
\end{gather*}
Since $\tr(\pi_\vec{b}^0\pi_\vec{a}^0\pi_\vec{c}^1\rho)=\tr(\pi_\vec{a}^0\pi_\vec{b}^0\pi_\vec{c}^1\rho)$ (from the given commutation relations) and $\tr(\pi_\vec{b}^0\pi_\vec{c}^1\rho)=0$, then $\tr(\pi_\vec{b}^0\pi_\vec{a}^0\pi_\vec{c}^1\rho)=0$ (using the same type of argument with resolving the identity and linearity of trace). Analogously from $\tr(\pi_\vec{a}^0\pi_\vec{b}^1\rho)=0$ it follows that $\tr(\pi_\vec{b}^1\pi_\vec{a}^0\pi_\vec{c}^1\rho)=0$, and hence $\tr(\pi_\vec{a}^0\pi_\vec{c}^1\rho)=0$. This implies $\tr(\pi_\vec{a}^1\pi_\vec{c}^1\rho)=\tr(\pi_\vec{c}^1\rho)$, and we can conclude
\begin{gather*}
P(\vec{a}=1|\vec{c}=1)=\frac{\tr(\pi_\vec{a}^1\pi_\vec{c}^1\rho)}{\tr(\pi_\vec{c}^1\rho)}=1 \Longleftrightarrow \ (\vec{c}=1 \Rightarrow \vec{a}=1).
\end{gather*}
\end{proof}

\symmetricB*

\begin{proof}
Throughout this proof, without loss of generality, we will denote $\vec{v}_i=1$ for all $i\in\{1,...,N\}$ as in earlier proofs. Thus the chain of equivalence under consideration becomes
\begin{gather*}
    \vec{a}_1=1\Leftrightarrow \vec{a}_2=1\Leftrightarrow\dots\Leftrightarrow \vec{a}_N=1.
\end{gather*}

Consider the first equivalence $\vec{a}_1=1\Leftrightarrow \vec{a}_2=1$. This implies that $P(\vec{a}_2=1|\vec{a}_1=1)=P(\vec{a}_1=1|\vec{a}_2=1)=1$. Expanding the conditional probabilities allows us to conclude that $P(\vec{a}_1=1,\vec{a}_2=1)=P(\vec{a}_1=1)=P(\vec{a}_2=1)$, which within quantum theory gives us $\tr(\pi_{\vec{a}_2}^{1}\pi_{\vec{a}_1}^{1}\rho)=\tr(\pi_{\vec{a}_1}^{1}\rho)=\tr(\pi_{\vec{a}_2}^{1}\rho)$.
Applying the same argument to every equivalence in the chain yields the following

\begin{align}
 \label{eq: sym_proof_1}
\begin{split}
      &\tr (\pi_{\vec{a}_1}^1 \rho) = \tr (\pi_{\vec{a}_2}^1 \rho) = \tr (\pi_{\vec{a}_3}^1 \rho)=\dots=\tr (\pi_{\vec{a}_N}^1 \rho)\\  
      = &\tr (\pi_{\vec{a}_1}^1 \pi_{\vec{a}_2}^1 \rho) = \tr (\pi_{\vec{a}_2}^1 \pi_{\vec{a}_3}^1 \rho)=\tr (\pi_{\vec{a}_3}^1 \pi_{\vec{a}_4}^1 \rho)=\dots=\tr (\pi_{\vec{a}_{N-1}}^1 \pi_{\vec{a}_N}^1 \rho). 
\end{split}
\end{align}

Additionally, $P(\vec{a}_2=1|\vec{a}_1=1)=P(\vec{a}_1=1|\vec{a}_2=1)=1$ implies $P(\vec{a}_1=1,\vec{a}_2=0)=P(\vec{a}_1=0,\vec{a}_2=1)=0$, which gives $\tr(\pi_{\vec{a}_2}^1(\id-\pi_{\vec{a}_1}^1))=\tr(\pi_{\vec{a}_1}^1(\id-\pi_{\vec{a}_2}^1))=0$. Applying the same argument to each equivalence in the chain gives us

\begin{gather}
\label{eq: sym_proof1}
    \tr(\pi_{\vec{a}_{i+1}}^1(\id-\pi_{\vec{a}_i}^1))=\tr(\pi_{\vec{a}_i}^1(\id-\pi_{\vec{a}_{i+1}}^1))=0 \quad \forall i\in\{1,...,N-1\}.
\end{gather}

Focusing on the triple $\vec{a}_1$, $\vec{a}_2$, $\vec{a}_3$, this implies
\begin{gather}
\label{eq: sym_proof_2}
      \tr (\pi_{\vec{a}_1}^1 \rho) = \tr (\pi_{\vec{a}_2}^1 \rho) = \tr (\pi_{\vec{a}_3}^1 \rho) = \tr (\pi_{\vec{a}_1}^1 \pi_{\vec{a}_2}^1 \rho) = \tr (\pi_{\vec{a}_2}^1 \pi_{\vec{a}_3}^1 \rho).
\end{gather}

We can rewrite
\begin{gather*}
    \tr (\pi_{\vec{a}_1}^1 \pi_{\vec{a}_3}^1 \rho) = \tr (\pi_{\vec{a}_1}^1 \pi_{\vec{a}_3}^1 \pi_{\vec{a}_2}^1 \rho) + \tr (\pi_{\vec{a}_1}^1 \pi_{\vec{a}_3}^1 (\id - \pi_{\vec{a}_2}^1) \rho),
\end{gather*} 

where $\tr (\pi_{\vec{a}_1}^1 \pi_{\vec{a}_3}^1 (\id - \pi_{\vec{a}_2}^1) \rho) = 0$ as $\tr (\pi_{\vec{a}_3}^1 (\id - \pi_{\vec{a}_2}^1) \rho) = \tr (\pi_{\vec{a}_1}^1 \pi_{\vec{a}_3}^1 (\id - \pi_{\vec{a}_2}^1) \rho)+\tr ((\id-\pi_{\vec{a}_1}^1) \pi_{\vec{a}_3}^1 (\id - \pi_{\vec{a}_2}^1) \rho)= 0$ (from \cref{eq: sym_proof1}). Hence, since $[\pi_{\vec{a}_2}^1, \pi_{\vec{a}_3}^1] = 0$,

\begin{gather*}
    \tr (\pi_{\vec{a}_1}^1 \pi_{\vec{a}_3}^1 \rho) = \tr (\pi_{\vec{a}_1}^1 \pi_{\vec{a}_3}^1 \pi_{\vec{a}_2}^1 \rho) = \tr (\pi_{\vec{a}_1}^1 \pi_{\vec{a}_2}^1 \pi_{\vec{a}_3}^1 \rho) = \tr (\pi_{\vec{a}_1}^1 \pi_{\vec{a}_2}^1 \rho) - \tr (\pi_{\vec{a}_1}^1 \pi_{\vec{a}_2}^1 (\id - \pi_{\vec{a}_3}^1) \rho).
\end{gather*}
Since $\tr (\pi_{\vec{a}_2}^1 (\id - \pi_{\vec{a}_3}^1) \rho) = 0$ (from \cref{eq: sym_proof1}), $\tr (\pi_{\vec{a}_1}^1 \pi_{\vec{a}_2}^1 (\id - \pi_{\vec{a}_3}^1) \rho) = 0$, and we have $  \tr (\pi_{\vec{a}_1}^1 \pi_{\vec{a}_3}^1 \rho) = \tr (\pi_{\vec{a}_1}^1 \pi_{\vec{a}_2}^1 \rho)$. Together with \cref{eq: sym_proof_1}, we arrive to
\begin{gather}
\label{eq: sym_proof_3}
    \tr (\pi_{\vec{a}_1}^1 \pi_{\vec{a}_3}^1 \rho) = \tr (\pi_{\vec{a}_1}^1 \rho)= \tr (\pi_{\vec{a}_3}^1 \rho).
\end{gather}

Using \cref{eq: sym_proof_1}, this further yields the following for the triple $\vec{a}_1$, $\vec{a}_3$ and $\vec{a}_4$, which is identical to \cref{eq: sym_proof_2} up to shifting 2 and 3 to 3 and 4 respectively.

\begin{gather}
    \label{eq: sym_proof_4}
          \tr (\pi_{\vec{a}_1}^1 \rho) = \tr (\pi_{\vec{a}_3}^1 \rho) = \tr (\pi_{\vec{a}_4}^1 \rho) = \tr (\pi_{\vec{a}_1}^1 \pi_{\vec{a}_3}^1 \rho) = \tr (\pi_{\vec{a}_3}^1 \pi_{\vec{a}_4}^1 \rho).
\end{gather}
 Running through the same argument, we would obtain 
 \begin{gather*}
    \tr (\pi_{\vec{a}_1}^1 \pi_{\vec{a}_4}^1 \rho) = \tr (\pi_{\vec{a}_1}^1 \rho)= \tr (\pi_{\vec{a}_4}^1 \rho).
\end{gather*}

Repeating the argument in a similar manner for the triples $(\vec{a}_1,\vec{a}_4,\vec{a}_5)$,...,$(\vec{a}_1,\vec{a}_{N-1},\vec{a}_{N})$, we would finally arrive at
 \begin{gather}
 \label{eq: sym_proof_5}
    \tr (\pi_{\vec{a}_1}^1 \pi_{\vec{a}_N}^1 \rho) = \tr (\pi_{\vec{a}_1}^1 \rho)= \tr (\pi_{\vec{a}_N}^1 \rho).
\end{gather}

In this case, we additionally have the commutation $[\pi_{\vec{a}_1}^1,\pi_{\vec{a}_N}^1]=0$ and hence $ \tr (\pi_{\vec{a}_1}^1 \pi_{\vec{a}_N}^1 \rho)= \tr (\pi_{\vec{a}_N}^1 \pi_{\vec{a}_1}^1 \rho)=P(\vec{a}_1=1,\vec{a}_N=1)$. Moreover, $\tr (\pi_{\vec{a}_1}^1 \rho)=P(\vec{a}_1=1)$ and $\tr (\pi_{\vec{a}_N}^1 \rho)=P(\vec{a}_N=1)$. Putting this together with \cref{eq: sym_proof_5}, we obtain
\begin{gather*}
    P(\vec{a}_1=1|\vec{a}_N=1)=P(\vec{a}_N=1|\vec{a}_1=1)=1.
\end{gather*}
This implies the equivalence $\vec{a}_1=1\Leftrightarrow \vec{a}_N=1$, and concludes the proof.
\end{proof}

\end{document}